\documentclass[11pt]{article}
\usepackage[margin=.7in,left=.7in]{geometry}
\usepackage{amsmath}
\usepackage{amsfonts}
\usepackage{amssymb}
\usepackage{amsthm}
\usepackage{graphicx}
\usepackage{color}
\usepackage{xcolor}
\usepackage{hyperref}
\usepackage[all]{xy}

\usepackage{tikz}

\definecolor{aqua}{rgb}{0, 1.0, 1.0}
\definecolor{fuschia}{rgb}{1.0, 0, 1.0}
\definecolor{gray}{rgb}{0.502, 0.502, 0.502}
\definecolor{lime}{rgb}{0, 1.0, 0}
\definecolor{maroon}{rgb}{0.502, 0, 0}
\definecolor{navy}{rgb}{0, 0, 0.502}
\definecolor{olive}{rgb}{0.502, 0.502, 0}
\definecolor{purple}{rgb}{0.502, 0, 0.502}
\definecolor{silver}{rgb}{0.753, 0.753, 0.753}
\definecolor{teal}{rgb}{0, 0.502, 0.502}

\makeatletter
\newdimen\itex@wd%
\newdimen\itex@dp%
\newdimen\itex@thd%
\def\itexspace#1#2#3{\itex@wd=#3em%
\itex@wd=0.1\itex@wd%
\itex@dp=#2ex%
\itex@dp=0.1\itex@dp%
\itex@thd=#1ex%
\itex@thd=0.1\itex@thd%
\advance\itex@thd\the\itex@dp%
\makebox[\the\itex@wd]{\rule[-\the\itex@dp]{0cm}{\the\itex@thd}}}
\makeatother

\makeatletter
\newif\if@sup
\newtoks\@sups
\def\append@sup#1{\edef\act{\noexpand\@sups={\the\@sups #1}}\act}%
\def\reset@sup{\@supfalse\@sups={}}%
\def\mk@scripts#1#2{\if #2/ \if@sup ^{\the\@sups}\fi \else%
  \ifx #1_ \if@sup ^{\the\@sups}\reset@sup \fi {}_{#2}%
  \else \append@sup#2 \@suptrue \fi%
  \expandafter\mk@scripts\fi}
\def\tensor#1#2{\reset@sup#1\mk@scripts#2_/}
\def\multiscripts#1#2#3{\reset@sup{}\mk@scripts#1_/#2%
  \reset@sup\mk@scripts#3_/}
\makeatother

\makeatletter
\newbox\slashbox \setbox\slashbox=\hbox{$/$}
\def\itex@pslash#1{\setbox\@tempboxa=\hbox{$#1$}
  \@tempdima=0.5\wd\slashbox \advance\@tempdima 0.5\wd\@tempboxa
  \copy\slashbox \kern-\@tempdima \box\@tempboxa}
\def\slash{\protect\itex@pslash}
\makeatother

\def\clap#1{\hbox to 0pt{\hss#1\hss}}
\def\mathllap{\mathpalette\mathllapinternal}
\def\mathrlap{\mathpalette\mathrlapinternal}

\def\mathllapinternal#1#2{\llap{$\mathsurround=0pt#1{#2}$}}
\def\mathrlapinternal#1#2{\rlap{$\mathsurround=0pt#1{#2}$}}

\let\oldroot\root
\def\root#1#2{\oldroot #1 \of{#2}}
\renewcommand{\sqrt}[2][]{\oldroot #1 \of{#2}}

\DeclareSymbolFont{symbolsC}{U}{txsyc}{m}{n}
\SetSymbolFont{symbolsC}{bold}{U}{txsyc}{bx}{n}
\DeclareFontSubstitution{U}{txsyc}{m}{n}

\DeclareSymbolFont{stmry}{U}{stmry}{m}{n}
\SetSymbolFont{stmry}{bold}{U}{stmry}{b}{n}

\DeclareFontFamily{OMX}{MnSymbolE}{}
\DeclareSymbolFont{mnomx}{OMX}{MnSymbolE}{m}{n}
\SetSymbolFont{mnomx}{bold}{OMX}{MnSymbolE}{b}{n}
\DeclareFontShape{OMX}{MnSymbolE}{m}{n}{
    <-6>  MnSymbolE5
   <6-7>  MnSymbolE6
   <7-8>  MnSymbolE7
   <8-9>  MnSymbolE8
   <9-10> MnSymbolE9
  <10-12> MnSymbolE10
  <12->   MnSymbolE12}{}

\makeatletter
\def\re@DeclareMathSymbol#1#2#3#4{%
    \let#1=\undefined
    \DeclareMathSymbol{#1}{#2}{#3}{#4}}
\re@DeclareMathSymbol{\neArrow}{\mathrel}{symbolsC}{116}
\re@DeclareMathSymbol{\neArr}{\mathrel}{symbolsC}{116}
\re@DeclareMathSymbol{\seArrow}{\mathrel}{symbolsC}{117}
\re@DeclareMathSymbol{\seArr}{\mathrel}{symbolsC}{117}
\re@DeclareMathSymbol{\nwArrow}{\mathrel}{symbolsC}{118}
\re@DeclareMathSymbol{\nwArr}{\mathrel}{symbolsC}{118}
\re@DeclareMathSymbol{\swArrow}{\mathrel}{symbolsC}{119}
\re@DeclareMathSymbol{\swArr}{\mathrel}{symbolsC}{119}
\re@DeclareMathSymbol{\nequiv}{\mathrel}{symbolsC}{46}
\re@DeclareMathSymbol{\Perp}{\mathrel}{symbolsC}{121}
\re@DeclareMathSymbol{\Vbar}{\mathrel}{symbolsC}{121}
\re@DeclareMathSymbol{\sslash}{\mathrel}{stmry}{12}
\re@DeclareMathSymbol{\bigsqcap}{\mathop}{stmry}{"64}
\re@DeclareMathSymbol{\biginterleave}{\mathop}{stmry}{"6}
\re@DeclareMathSymbol{\invamp}{\mathrel}{symbolsC}{77}
\re@DeclareMathSymbol{\parr}{\mathrel}{symbolsC}{77}
\makeatother

\makeatletter
\def\Decl@Mn@Delim#1#2#3#4{%
  \if\relax\noexpand#1%
    \let#1\undefined
  \fi
  \DeclareMathDelimiter{#1}{#2}{#3}{#4}{#3}{#4}}
\def\Decl@Mn@Open#1#2#3{\Decl@Mn@Delim{#1}{\mathopen}{#2}{#3}}
\def\Decl@Mn@Close#1#2#3{\Decl@Mn@Delim{#1}{\mathclose}{#2}{#3}}
\Decl@Mn@Open{\llangle}{mnomx}{'164}
\Decl@Mn@Close{\rrangle}{mnomx}{'171}
\Decl@Mn@Open{\lmoustache}{mnomx}{'245}
\Decl@Mn@Close{\rmoustache}{mnomx}{'244}
\makeatother

\makeatletter
\DeclareRobustCommand\widecheck[1]{{\mathpalette\@widecheck{#1}}}
\def\@widecheck#1#2{%
    \setbox\z@\hbox{\m@th$#1#2$}%
    \setbox\tw@\hbox{\m@th$#1%
       \widehat{%
          \vrule\@width\z@\@height\ht\z@
          \vrule\@height\z@\@width\wd\z@}$}%
    \dp\tw@-\ht\z@
    \@tempdima\ht\z@ \advance\@tempdima2\ht\tw@ \divide\@tempdima\thr@@
    \setbox\tw@\hbox{%
       \raise\@tempdima\hbox{\scalebox{1}[-1]{\lower\@tempdima\box
\tw@}}}%
    {\ooalign{\box\tw@ \cr \box\z@}}}
\makeatother

\makeatletter
\def\udots{\mathinner{\mkern2mu\raise\p@\hbox{.}
\mkern2mu\raise4\p@\hbox{.}\mkern1mu
\raise7\p@\vbox{\kern7\p@\hbox{.}}\mkern1mu}}
\makeatother

\newcommand{\underoverset}[3]{\underset{#1}{\overset{#2}{#3}}}
\newcommand{\widevec}{\overrightarrow}

\newcommand{\lt}{<}

\newcommand{\R}{\ensuremath{\mathbb R}}
\newcommand{\Z}{\ensuremath{\mathbb Z}}
\newcommand{\Q}{\ensuremath{\mathbb Q}}

\renewcommand{\(}{\begin{equation*}}
\renewcommand{\)}{\end{equation*}}
\newcommand{\bea}{\begin{eqnarray*}}
\newcommand{\eea}{\end{eqnarray*}}

\theoremstyle{italics}
\newtheorem{theorem}{Theorem}[section]
\newtheorem{lemma}[theorem]{Lemma}
\newtheorem{prop}[theorem]{Proposition}
\newtheorem{cor}[theorem]{Corollary}

\theoremstyle{definition}
\newtheorem{defn}[theorem]{Definition}

\newtheorem{example}[theorem]{Example}

\theoremstyle{remark}
\newtheorem{remark}[theorem]{Remark}
\newtheorem{note[theorem]}{Note}

\usepackage{amsfonts}

\newcommand{\Exterior}{\scalebox{.8}{\ensuremath \bigwedge}}

\begin{document}

\title{
Higher T-duality of super M-branes}

 \author{Domenico Fiorenza\thanks{Dipartimento di Matematica, La Sapienza Universit\`a di Roma, Piazzale Aldo Moro 2, 00185 Rome, Italy.}\and
     Hisham Sati\thanks{Division of Science and Mathematics, New York University, Abu Dhabi, UAE.}\and
          Urs Schreiber\thanks{Mathematics Institute, Czech Academy of Science, {\v Z}itna 25, 115 67 Praha 1, Czech Republic.}
     }

\maketitle
\abstract{
  We establish a higher generalization of super $L_\infty$-algebraic T-duality of
  super WZW-terms for super $p$-branes.  In particular, we demonstrate spherical T-duality
  of super M5-branes propagating on exceptional-geometric 11d super spacetime.
}

\tableofcontents

\newpage

\section{Introduction and summary}
\label{Introduction}

The torsion constraints of supergravity and of super $p$-brane models
have remarkably far-reaching consequences.
This is well-known and goes back to \cite{GWZ}, but since it is
key to our study of M-brane phenomena, we briefly recall it, putting it into a geometric
perspective that will be useful for presenting our results.
The following picture illustrates how super-Cartan geometry models a super-manifold
by moving local super-Minkowski frames around it,
and that super-torsion freedom of the super-vielbein $(E^a) := (e^a, \psi^\alpha)$ means
that, moreover, each first-order infinitesimal neighborhood is
supersymmetrically identified with the local model. This is explained in \cite{Lott90, EE}, following
the seminal result of \cite{Guillemin65}.

\vspace{-.7cm}

\begin{center}
 \includegraphics[width=.9\textwidth]{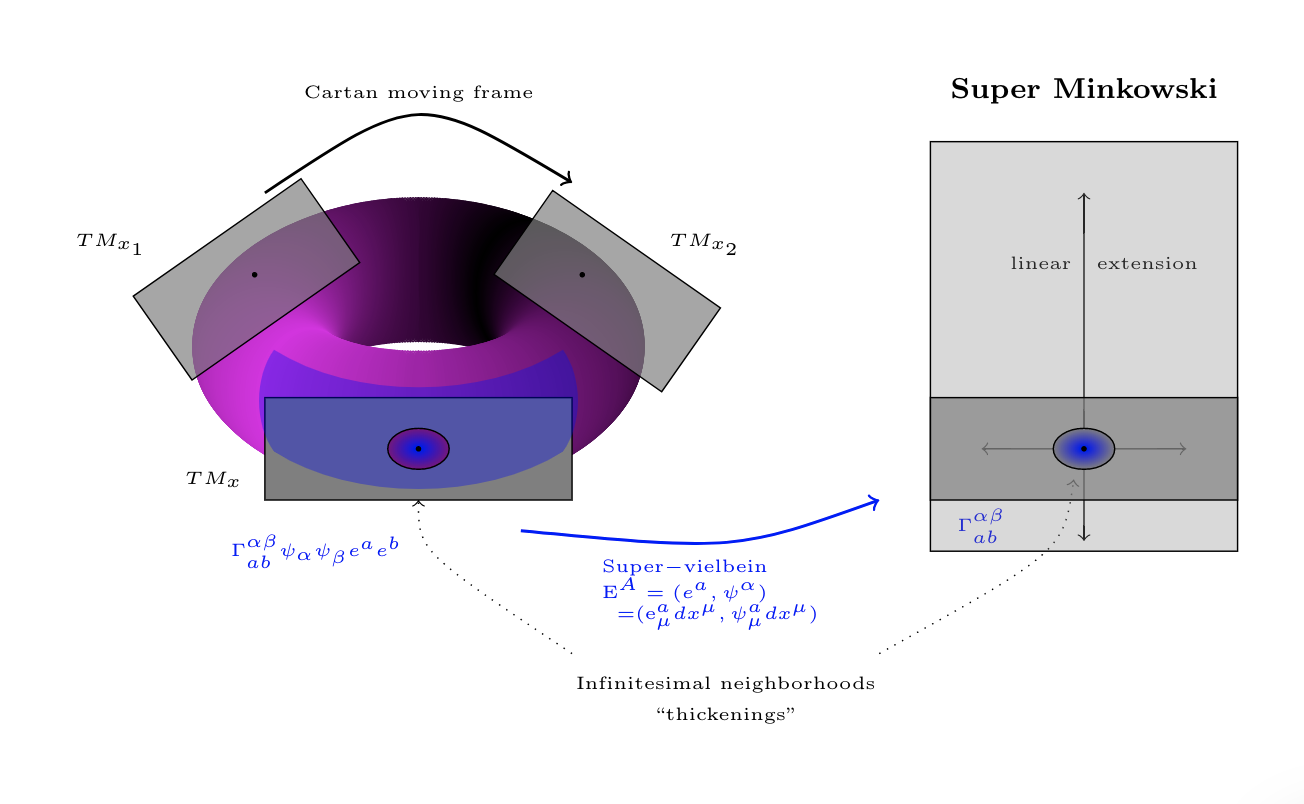}
\end{center}

\vspace{-1cm}
Explicitly, vanishing of the super-torsion in eleven-dimensional supergravity
says, in particular, that the bifermionic component
of the torsion tensor is constrained to be in
each super-tangent space given by (see \cite[(14)]{BST2})
$$
  T^a_{\alpha \beta} = \Gamma^a_{\alpha \beta}
  \,.
$$
In terms of Cartan calculus on super-Minkowski spacetime, this means that (see Section \ref{SuperLInfinityCohomology} below)
\begin{equation}
  \label{VielbeinDifferential}
  d e^a = \overline{\psi} \Gamma^a \psi
  \,,
\end{equation}
which makes it manifest that this identifies each super-tangent space of a supergravity background with
super-Minkowski spacetime not just as a super-vector space, but as a super Lie algebra: the translational
part of the supersymmetry algebra. This is the basis for the powerful super-Cartan-geometric perspective
on supergravity advocated in \cite{DF,CDF}; and in disguise the innocent-looking differential (\ref{VielbeinDifferential})
governs much of the structure of supergravity (also known as \emph{$\tau$-cohomology} \cite{BBLPT90}).
Remarkably, in 11d the constraint of vanishing super-torsion alone is already equivalent to the full supergravity equations of motion \cite{CL, Howe97, FOS},
a first indication that the implications of these super tangent space-wise constraints are far-reaching.

\medskip
Moreover, the bifermionic component
$
  H_{1,2}
  :=
  H_{a \alpha \beta} e^a \wedge \psi^\alpha \wedge \psi^\beta$
of the background field strength $H$, to which the type I superstring \cite{GS}
in dimensions 3, 4, 6 and 10 couples, is constrained in each super tangent-space
to be of the form \cite[(2.11), (2.15)]{BST1} (see expression \eqref{SuperstringCocycleTypeII} below))
\begin{equation}
  \label{BSTString}
  H_{a \alpha \beta} = \Gamma_{a \alpha \beta}
  \,.
\end{equation}
Similarly, the bifermionic component
$
  G_{2,2}
  :=
  G_{a b \alpha \beta} e^a \wedge e^b \wedge \psi^\alpha \wedge \psi^\beta
$
of the background field strength, to which the supermembrane in dimension 4, 5, 7 and 11
couples, is constrained in each super-tangent space
to be given by \cite[(15)]{BST2} (see \eqref{M2M5CocyclesOn11dSuperMinkowski} below)
\begin{equation}
  \label{BSTMembrane}
  H_{a b \alpha \beta} = \Gamma_{a b \alpha \beta}
  \,.
\end{equation}
From the point of view of the kappa-symmetric super-brane sigma-model, this is due to the fact that precisely in these dimensions these forms are
super Lie algebra cocycles for the translational supersymmetry algebra \cite{AETW87, AzTo89, BaezHuerta10, BaezHuerta11},
a fact known as the ``old brane scan''. Strikingly, in 11d the same forms are implied by the super torsion-free super Cartan geometry
that also implies the 11d supergravity equations of motion \cite{CL, Howe97, FOS}.

\medskip
Directly analogous statements, with more complicated local expressions, hold for all the super
D-branes \cite[(3.9)]{CGNSW97} and for the M5-brane \cite[(5)]{LT} \cite[(6)]{BLNPST97}
(recalled below in Example \ref{HigherTDualCorrespondenceForM2Brane}).
A key difference here is that understanding these as super-cocycles requires passage from
ordinary to extended super Minkowski spacetime \cite{CAIB00,IIBAlgebra}. In \cite{FSS13} we
pointed out, following \cite[p. 54]{SSS1} as expanded on in \cite{Huerta12, Huerta14}, that this means
to pass to \emph{homotopy} super Lie algebras, called super \emph{$L_\infty$-algebras} (see Section \ref{SuperLInfinityHomotopyTheory})
and we showed that this perspective completes the ``old brane scan'' to a ``brane bouquet'' of
iterative higher central extensions of super $L_\infty$-algebras:

$$
  \hspace{-1.2cm}
  \xymatrix@=1em{
    &
    &&&& \mathfrak{m}5\mathfrak{brane}
     \ar[d]
    &&&&
    \mbox{ \cite{FSS15} }
    \\
    &
    &&
     && \mathfrak{m}2\mathfrak{brane}
    \ar[dd]
    &&
    \\
    &
    &
    \mathfrak{d}5\mathfrak{brane}
    \ar[ddr]
    &
    \mathfrak{d}3\mathfrak{brane}
    \ar[dd]
    &
    \mathfrak{d}1\mathfrak{brane}
    \ar[ddl]
    &
    & \mathfrak{d}0\mathfrak{brane}
    \ar@{}[ddd]|{\mbox{\tiny (pb)}}
    \ar[ddr]
    \ar@{..>}[dl]
    &
    \mathfrak{d}2\mathfrak{brane}
    \ar[dd]
    &
    \mathfrak{d}4\mathfrak{brane}
    \ar[ddl]
    \\
    &
    &
    \mathfrak{d}7\mathfrak{brane}
    \ar[dr]
    &
    &
    & \mathbb{R}^{10,1\vert \mathbf{32}}
      \ar[ddr]
    &&&
    \mathfrak{d}6\mathfrak{brane}
    \ar[dl]
    &
    \mbox{
      \cite{FSS13, FSS16a}
    }
    \\
    &
    &
    \mathfrak{d}9\mathfrak{brane}
    \ar[r]
    &
    \mathfrak{string}_{\mathrm{IIB}}
    \ar[dr]
    &
    & \mathfrak{string}_{\mathrm{het}}
      \ar[d]
    &&
    \mathfrak{string}_{\mathrm{IIA}}
    \ar[dl]
    &
    \mathfrak{d}8\mathfrak{brane}
    \ar[l]
    \\
    &
    &
    &
    &
    \mathbb{R}^{9,1 \vert \mathbf{16} + {\mathbf{16}}}
    \ar@{<-}@<-3pt>[r]
    \ar@{<-}@<+3pt>[r]
    & \mathbb{R}^{9,1\vert \mathbf{16}} \ar[dr]
    \ar@<-3pt>[r]
    \ar@<+3pt>[r]
    &
    \mathbb{R}^{9,1\vert \mathbf{16} + \overline{\mathbf{16}}}
    \\
    &
    &
    &
    &
    &
    \mathbb{R}^{5,1\vert \mathbf{8}}
    \ar[dl]
    &
    \mathbb{R}^{5,1 \vert \mathbf{8} + \overline{\mathbf{8}}}
    \ar@{<-}@<-3pt>[l]
    \ar@{<-}@<+3pt>[l]
    \\
    &
    &
    &
    &
    \mathbb{R}^{3,1\vert \mathbf{4}+ \mathbf{4}}
    \ar@{<-}@<-3pt>[r]
    \ar@{<-}@<+3pt>[r]
    &
    \mathbb{R}^{3,1\vert \mathbf{4}}
    \ar[dl]
    &&&&
    \mbox{ \cite{HuertaSchreiber} }
    \\
    &
    &
    &
    &
    \mathbb{R}^{2,1 \vert \mathbf{2} + \mathbf{2} }
    \ar@{<-}@<-3pt>[r]
    \ar@{<-}@<+3pt>[r]
    &
    \mathbb{R}^{2,1 \vert \mathbf{2}}
    \ar[dl]
    \\
    &
    &
    &
    &
    \mathbb{R}^{0 \vert \mathbf{1}+ \mathbf{1}}
    \ar@{<-}@<-3pt>[r]
    \ar@{<-}@<+3pt>[r]
    &
    \mathbb{R}^{0\vert \mathbf{1}}
    \\
    \\
    &&& \fbox{Type IIB} \ar@{<->}@/_2pc/[rrrr]_{\mbox{\color{blue}T-Duality}} && \fbox{Type I} && \fbox{Type IIA} && \mbox{ \cite{FSS16} }
  }
$$
\label{TheBraneBouquet}

The full implication of these super tangent space-wise constraints
for super WZW-terms for super $p$-branes has perhaps not been fully appreciated yet.
Notice that any duality in string/M-theory, when fully taking into account all the fermionic degrees
of freedom, will have to respect all these constraints from super tangent space to super tangent space.
This is a strong condition on any duality.

\medskip
Indeed, in \cite{FSS16} we had shown that the super tangent space-wise torsion constraints/super-WZW
terms of the super F1/D$p$ branes in type II super-spacetime already completely reveal the structure of
T-duality on brane charges; we recall this below in Section \ref{OrdinaryTypeIITDuality}.
This structure had previously been proposed under the name ``topological T-duality''
 (\cite{BouwknegtEvslinMathai04, BHM03, BunkeSchick05})
and had been conjectured to underlie the actual T-duality of string theory
(see also Remark \ref{TDualityAxiom} below).

\medskip
The following picture illustrates how a global duality, such as topological T-duality $\mathcal{T}_{global}$, restricts
to a duality on super-tangent spaces, such as the super $L_\infty$-algebraic T-duality $\mathcal{T}_{\mathrm{loc}}$ of \cite{FSS16} (Section \ref{OrdinaryTypeIITDuality}
below):

\vspace{-.7cm}

\begin{center}
\includegraphics[width=.7\textwidth]{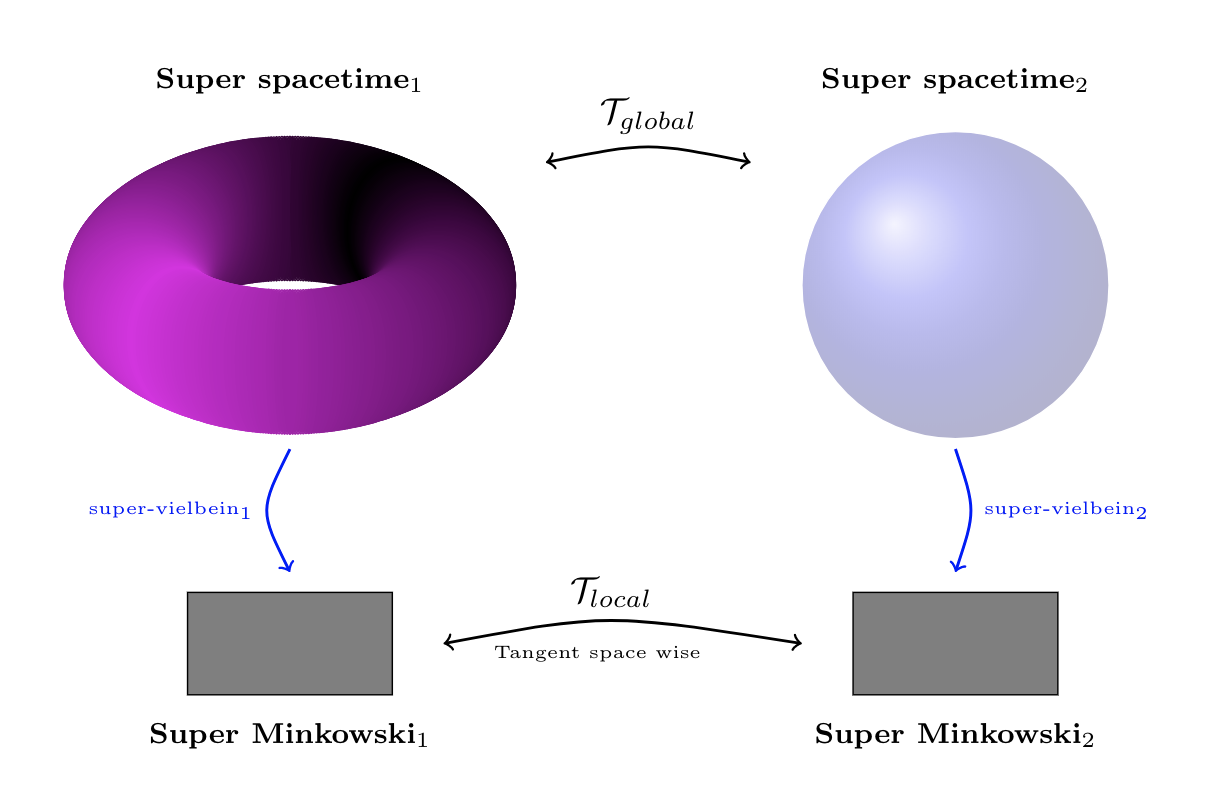}
\end{center}

\vspace{-.7cm}

\noindent Note that the double-headed arrows here indicate not direct maps, but rather ``spans'' of maps,
called {\it correspondences}; see Def. \ref{Correspondences} below.

\medskip
Turning this around, it means that analysis of the super-tangent space-wise super-WZW terms of the super
$p$-branes may be used to systematically discover and analyze previously unknown facts about super
M-brane physics, and hence about the elusive theory of which they are a part. This is what we consider here.
In this paper we establish further extensions of the above brane bouquet which may be organized into the following
diagrammatic table of contents. In the remainder of this introduction we will explain informally what
(some of) the boxed items in this diagram mean.
$$
  \hspace{-.5cm}
  \xymatrix@=5pt{
    &
    &
    \\
    \fbox
    {$
    \begin{array}{cc}
      \mbox{\bf \footnotesize Spherical}
      \\
      \mbox{\bf \footnotesize T-duality}
      \\
      \mbox{\footnotesize  (Prop. \ref{SphericalTDualityOnExtendedSuperspacetime})}
    \end{array}
    $}
    &
    \mathbb{R}^{10,1\vert\mathbf{32}}_{\mathrm{exc},s}
      \ar[rr]^{\mathrm{comp}}
      \ar[dd]_-{ \mathrm{hofib}( \mu_{\mathrm{exc},s} ) }
    &
    & \mathfrak{m}2\mathfrak{brane}
    \ar[ddd]^-{ \mathrm{hofib}(\mu_{{}_{M2}}) }
    &
    &
    \fbox
    {$
    \begin{array}{cc}
      \mbox{\bf \footnotesize  Spherical}
      \\
      \mbox{\bf \footnotesize  T-duality}
      \\
      \mbox{\footnotesize  (Prop. \ref{M5CocycleIsSPhericalTDualToItself})}
    \end{array}
    $}
    \\
    &
    \\
    \fbox
    {$
    \begin{array}{cc}
      \mbox{\bf \footnotesize  528-toroidal}
      \\
      \mbox{\bf \footnotesize  T-duality}
      \\
      \mbox{\footnotesize  (Prop. \ref{ExceptionalToroidalTDuality} (ii))}
    \end{array}
    $}
    &
    \mathbb{R}^{10,1\vert\mathbf{32}}_{\mathrm{exc}}
    \ar[drr]^{ \mathrm{hofib}( \mu_{\mathrm{exc}} ) }
    \ar[dddddd]_-{
      \mathrm{hofib}( \psi \wedge \overline{\psi} )
    }
    &
    &
    &
    &
    \fbox
    {$
    \begin{array}{cc}
      \mbox{\bf \footnotesize  517-toroidal}
      \\
      \mbox{\bf \footnotesize  T-duality}
      \\
      \mbox{\footnotesize  (Prop. \ref{ExceptionalToroidalTDuality} (iii))}
    \end{array}
    $}
    \\
    &&
    &
    \mathbb{R}^{10,1\vert\mathbf{32}}
    \ar[dr]
    \\
    &
    &
    \mathbb{R}^{9,1\vert \mathbf{16} + {\mathbf{16}}}
    &
    \mathbb{R}^{9,1\vert\mathbf{16}}
    \ar@<-3pt>[l]
    \ar@<+3pt>[l]
    \ar@<-3pt>[r]
    \ar@<+3pt>[r]
    \ar[dr]
    &
    \mathbb{R}^{9,1\vert \mathbf{16} + \overline{\mathbf{16}}}
    &
    \fbox
    {$
    \begin{array}{cc}
      \mbox{\bf \footnotesize  Cyclic}
      \\
      \mbox{\bf \footnotesize  T-duality}
      \\
      \mbox{\footnotesize  (Sect. \ref{OrdinaryTypeIITDuality})}
    \end{array}
    $}
    \\
    &
    &
    \mathbb{R}^{5,1\vert \mathbf{8} + {\mathbf{8}}}
    &
    \mathbb{R}^{5,1\vert\mathbf{8}}
    \ar@<-3pt>[l]
    \ar@<+3pt>[l]
    \ar@<-3pt>[r]
    \ar@<+3pt>[r]
    \ar[dl]
    &
    \mathbb{R}^{5,1\vert \mathbf{8} + \overline{\mathbf{8}}}
    &
    \fbox
    {$
    \begin{array}{cc}
      \mbox{\bf \footnotesize  Cyclic}
      \\
      \mbox{\bf \footnotesize  T-duality}
      \\
      \mbox{\footnotesize  (Prop. \ref{TDualityCorrespondenceOver5dSuperMinkowskiSpacetime})}
    \end{array}
    $}
    \\
    &
    &
    \mathbb{R}^{3,1\vert \mathbf{4} + \mathbf{4}}
    &
    \mathbb{R}^{3,1\vert\mathbf{4}}
    \ar@<-3pt>[l]
    \ar@<+3pt>[l]
    \ar[dl]
    &
    \\
    &
    &
    \mathbb{R}^{2,1\vert \mathbf{2} + \mathbf{2}}
    &
    \mathbb{R}^{2,1\vert\mathbf{2}}
    \ar@<-3pt>[l]
    \ar@<+3pt>[l]
    \ar[dl]
    &
    \\
    &
    \mathbb{R}^{0\vert \mathbf{32}}
    \ar@{<-}@<-3pt>[r]^{\ }="t"
    \ar@{<-}@<+3pt>[r]_-{\ }="s"
    &
    \mathbb{R}^{0\vert \mathbf{2}}
    &
    \mathbb{R}^{0 \vert \mathbf{1}}
    \ar@<-3pt>[l]
    \ar@<+3pt>[l]
    &
    \\
    \\
    &
    \fbox{
      {$
      \begin{array}{c}
        \mbox{Exceptional}
      \end{array}
      $}
    }
    &
    \fbox{
      {$
      \begin{array}{c}
        \mbox{Type IIB}
      \end{array}
      $}
    }
    &
    \fbox{
      {$
      \begin{array}{c}
        \mbox{Type I}
      \end{array}
      $}
    }
    &
    \fbox{
      {$
      \begin{array}{c}
        \mbox{Type IIA}
      \end{array}
      $}
    }
    \ar@{..} "s"; "t"
  }
$$

It is noteworthy that this is a derivation of M-theoretic structure from first principles, not involving
any extrapolation from perturbative string theory nor any conjectures or informal analogies from other sources.
Instead, this derivation is the systematic rigorous analysis of the progression of higher extensions in super homotopy theory that is emanating from the
superpoint (see \cite{Schreiber17b} for exposition of our perspective).

\medskip



We now explain this conceptually.  We first observe that the WZW-term of the M5-brane
 sigma-model \cite{BLNPST97, FSS15} exhibits 3-spherical topological T-duality, which is
 interestingly a self-duality; this is Prop. \ref{M5CocycleIsSPhericalTDualToItself} below.
 It may be understood by observing that the relation between the joint M2/M5-cocycle (recalled as Example \ref{HigherTDualCorrespondenceForM2Brane} below)
 $$
   d \mu_{{}_{M5}} + \tfrac{1}{2} \mu_{{}_{M2}} \wedge \mu_{{}_{M2}} = 0
   \,,
 $$
 which is the super $L_\infty$-algebraic avatar of the equations of motion  on the 11d supergravity flux forms
 $$
   d G_7 + \tfrac{1}{2} G_4 \wedge G_4 = 0\;,
 $$
 is a higher analog of the relation
 $$
   d \mu_{{}_{F1}}\vert_{8+1} + c_2^{{}^{\mathrm{IIA}}} \wedge c_2^{{}^{\mathrm{IIA}}} = 0
 $$
 that encodes super-topological T-duality for type II superstrings (recalled in Section \ref{OrdinaryTypeIITDuality} below).
 Equivalently (by Prop. \ref{ClassifyingDataHigherTDuality} below) this is the fact that the genuine M5-brane supercocycle (see \eqref{M5cocycleInTDualityPair} below)
 \begin{equation}
   \label{M5CocycleRecalled}
   \tilde \mu_{{}_{M5}} := 2 \mu_{{}_{M5}} + c_3 \wedge \mu_{{}_{M2}}
   \,,
 \end{equation}
 which is the super $L_\infty$-algebraic avatar of the higher WZW term of the M5-brane sigma-model
 $$
   d \mathbf{L}_{M5,\mathrm{WZW}} \;=\; G_7 + \tfrac{1}{2} C_3 \wedge G_4\;,
 $$
 has an algebraic structure which is a higher degree analog of that of a ``T-dualizable H-flux''
 whose super $L_{\infty}$-algebraic avatar is (see (\ref{SuperstringCocycleTypeII}) below)
 \begin{equation}
   \label{TDuality3CocycleRecalled}
    \mu_{{}_{F1}}^{{}^{\mathrm{IIA/B}}}
    \;=\;
    \mu_{{}_{F1}}\vert_{8+1} \,+\, e^9_{{}_{A/B}} \wedge c_2^{{}^{\mathrm{IIB/A}}}
    \,.
 \end{equation}
 In comparing  expression \eqref{M5CocycleRecalled} with expression \eqref{TDuality3CocycleRecalled},
  one sees that the role of the
 fiberwise Maurer-Cartan 1-form $e^9$ in cyclic type II string T-duality is now taken by the C-field $c_3$.
 This shows that when passing from T-duality for type II string theory (recalled in Section \ref{OrdinaryTypeIITDuality} below)
 to spherical T-duality of M5-branes (Section \ref{SphericalTDualityOfM5BranesOnM2ExtendedSpacetime} below),
 the analog of the role of the 10d super-spacetime fibered over the 9d super spacetime is now the
 2-gerbe (or, rationally, the 3-sphere fibration) that is classified by the M2-brane charge, which is fibered \emph{over 11d spacetime}.
 This phenomenon is also indicated in the table at the beginning of section \ref{Sec-HigherTDualityCorrespondences} below.

$$
  \xymatrix@C=-20pt{
   {
    \fbox
    {
    $
    \begin{array}{cc}
      \mbox{  M2-extended}
      \\
      \mbox{  super-spacetime}
    \end{array}
    $
    }}
    \ar[dr]|{\;\;\; \;\; \mbox{\it \tiny Fibration classified by}    }
    \ar@{<->}@/^2pc/[rr]|{ \mbox{\it \tiny Spherical} \atop \mbox{\it \tiny T-self-duality} }
    &&
     {
    \fbox
    {
    $
    \begin{array}{cc}
      \mbox{  M2-extended}
      \\
      \mbox{  super-spacetime}
    \end{array}
    $
    }}
    \ar[dl]|{\hspace{-.5cm} \mbox{ \it \tiny  M2-brane WZW-term }   }
    \\
    &
     { \fbox{11d super-spacetime} } &
  }
  \;\;\;\;\;\;\;
  \xymatrix{
         \mathfrak{m}2\mathfrak{brane}
    \ar[dr]|{\mathrm{hofib}(\mu_{{}_{M2}})}
    \ar@{<->}@/^2pc/[rr]|{ \mbox{\it \tiny Homotopy } \atop \mbox{\it \tiny fiber-product} }
    &&
     \mathfrak{m}2\mathfrak{brane}
    \ar[dl]|{\mathrm{hofib}(\mu_{{}_{M2}})}
    \\
    & \mathbb{R}^{10,1\vert\mathbf{32}} &
  }
$$
This indeed makes sense, as highlighted before in \cite[Remark 3.11, Sec. 4.4]{FSS13}: the
component of a plain sigma-model field in this 2-gerbe fiber is equivalently the higher gauge
field on the M5-brane's worldvolume.
$$
  \xymatrix{
    &
    &
        &&
    \mathfrak{m}2\mathfrak{brane}
    \ar[d]|{\;\;\; \mathrm{hofib}(\mu_{{}_{M2}})}
    &
    {
    \fbox
    {
    $
    \begin{array}{cc}
      \mbox{\footnotesize  Super 2-gerbe}
      \\
      \mbox{\footnotesize  over super-spacetime}
    \end{array}
    $
    }}
             \\
        {
    \fbox
    {
    $
    \begin{array}{cc}
      \mbox{\footnotesize  M5-brane}
      \\
      \mbox{\footnotesize worldvolume}
    \end{array}
    $
    }}
    &
    \Sigma_{M5}
    \ar@/_1pc/[rrr]_{\mbox{\it \tiny Plain sigma-model field}}
    \ar@{..>}[urrr]|{ \mbox{\it \tiny Sigma-model-}\atop \mbox{\it \tiny \& gauge-field } }
    & && \mathbb{R}^{10,1\vert\mathbf{32}} &
    {
    \fbox
    {
    $
    \begin{array}{cc}
      \mbox{\footnotesize  Target}
      \\
      \mbox{\footnotesize  Super-spacetime}
    \end{array}
    $
    }}
  }
$$
This way higher geometry provides a unification of sigma-model fields with worldvolume gauge fields, and our Prop. \ref{M5CocycleIsSPhericalTDualToItself}
says that this unified perspective reveals spherical T-duality in M-brane theory.

\medskip
 In order to shed more light on this subtle point, we next observe that the decomposed C-field on
 the exceptional super-tangent space of 11d super-spacetime serves as a transgression element for the M2-brane
 WZW terms. This is Prop. \ref{TransgressionElementForM2Cocycle} below, which  provides a re-interpretation
 of the ``hidden'' D'Auria-Fr{\'e} algebra from \cite{DF, BAIPV04} (see also \cite{BdAPV, ADR16, ADR17})
 over which the supergravity C-field decomposes super-equivariantly.
 We expand on that in Section \ref{ExceptionalGenralisedGeometry} below, where we explain how the D'Auria-Fr{\'e} algebra may be
 regarded as providing the supersymmetric refinement of the exceptional generalized geometry for the C-field proposed in \cite{Hull07, PachecoWaldram08}\footnote{
   Finding a suitable supergeometric refinement of exceptional/generalized geometry is stated as an open problem in
   \cite[p. 18]{CederwallEdlundKarlsson13}\cite[p. 4,7]{Cederwall14}. In \cite[p. 10,11]{Bandos17}
   it was proposed that $\mathbb{R}^{10,1\vert \mathbf{32}}_{\mathrm{exc}}$ (in our notation)
    is the answer. But for our argument in
   Section \ref{ExceptionalGenralisedGeometry} the further fermionic extension to the DF-algebra $\mathbb{R}^{10,1\vert \mathbf{32}}_{\mathrm{exc},s}$
   is crucial.
 }
 and how spherical T-duality acts by duality transformation on the resulting super-moduli spaces.
 $$
  \hspace{-8.3cm}
  \xymatrix@R=7pt{
    \fbox{\footnotesize
      Def. \ref{FermionicExtensionOfExceptionalTangentSuperspacetime}
    }
    &
    \mathbb{R}^{10,1\vert \mathbf{32}}_{\mathrm{exc},s}
    \ar[d]|{ \;\;\;\;\mathrm{hofib}( \mu_{\mathrm{exc},s} ) }
    & \simeq_{{}_{\mathbb{R}}} &
    \mathrlap{
      \underset{
        \mbox{
          \tiny
          Exceptional generalized super-geometry
        }
      }{
      \underbrace{
        \mathbb{R}^{10,1}
        \oplus \Exterior^2\big( \mathbb{R}^{10,1\vert \mathbf{32}}  \big)^\ast
        \oplus \Exterior^5\big( \mathbb{R}^{10,1\vert \mathbf{32}}  \big)^\ast
        \oplus \mathbf{32}_{\mathrm{odd}} \oplus \mathbf{32}_{\mathrm{odd}}
      }
      }
    }
    \\
    \fbox{ \footnotesize
        Def. \ref{MaximalCentralExtensionOfNIs32Superpoint},
        Prop. \ref{SpinActionOnExceptionalTangentSuperspacetime}
    }
    &
    \mathbb{R}^{10,1\vert \mathbf{32}}_{\mathrm{exc}}
    \ar[d]|{ \; \;\mathrm{hofib}(\mu_{\mathrm{exc}}) }
    & \simeq_{{}_{\mathbb{R}}} &
    \mathrlap{
      \underset{
        \mbox{\tiny
          Exceptional generalized geometry
        }
      }{
      \underbrace{
        \mathbb{R}^{10,1}
        \oplus \Exterior^2\big( \mathbb{R}^{10,1\vert \mathbf{32}}  \big)^\ast
        \oplus \Exterior^5\big( \mathbb{R}^{10,1\vert \mathbf{32}}  \big)^\ast
      }
      }
      \oplus \mathbf{32}_{\mathrm{odd}}
    }
    \\
    \fbox{\footnotesize
      Ex. \ref{superMinkowskiSuperLieAlgebra}
    }
    &
    \mathbb{R}^{10,1\vert \mathbf{32}}
    &\simeq_{{}_{\mathbb{R}}}&
       \mathrlap{
        \underset{
          \mbox{\tiny Super spacetime}
        }{
        \underbrace{
          \underset{
            \mbox{\tiny Spacetime}
          }{
          \underbrace{
            \mathbb{R}^{10,1}
          }
          }
          \oplus \mathbf{32}_{\mathrm{odd}}
        }
      }
      }
  }
$$
 In particular, using the relation to the Ho{\v r}ava-Witten boundary of the decomposed C-field as in \cite{EvslinSati02}, this explains
 the nature of the extra fermion generator in the D'Auria-Fr{\'e} algebra, which was a concern in \cite{ADR16, ADR17}
 (see Example \ref{InterpretingTheExtraFermionicGenerator} below).
 \footnote{Note that decomposability of the corresponding field strength $G_4$ also arises naturally in
 the classification of backgrounds for 11-dimensional supergravity \cite{FOP}.}

\medskip
This means that sigma-model fields with values in the super 2-gerbe, hence the pairs of ordinary sigma-model fields
and worldvolume higher gauge fields,  may be obtained from plain sigma-model fields into the exceptional
tangent space of 11d super-spacetime. Here, again, the sigma-model field components into the fibers over
spacetime transmute into gauge fields on the brane's worldvolume:
$$
\hspace{-3mm}
  \xymatrix@C=1.6em{
   {
    \fbox
    {
    $
    \begin{array}{cc}
      \mbox{\footnotesize  Exceptional tangent space}
      \\
      \mbox{\footnotesize  over super-spacetime}
    \end{array}
    $
    }}
    &
    \mathbb{R}^{10,1\vert\mathbf{32}}_{\mathrm{exc},s}
    \ar[rrrr]|{ { \mbox{\it \tiny Decomposition} \atop  \mbox{\it \tiny of C-field} } }
    &&&&
    \mathfrak{m}2\mathfrak{brane}
    &
       {
    \fbox
    {
    $
    \begin{array}{cc}
      \mbox{\footnotesize  Super 2-gerbe}
      \\
      \mbox{\footnotesize  over super-spacetime}
    \end{array}
    $
    }}
    \\
    \\
    \Sigma_{M5}
    \ar@{-->}@/_2pc/[uurrrrr]|{ \mbox{\it \tiny Sigma-model-}\atop \mbox{\it \tiny\& gauge-field } }
    \ar@{..>}@/^1pc/[uur]|{ \mbox{\it \tiny U-duality equivariant} \atop
    \mbox{\it \tiny sigma-model- \& gauge-field }  }
    & &&
    &
  }
$$
That the M5-brane indeed has a formulation as a plain sigma-model on the exceptional tangent space over 11d-spacetime this way was recently observed in \cite{Sakatani16}. Related observations
are due to \cite{Yale}.

It may be noteworthy that passing from $\mathfrak{m}2\mathfrak{brane}$ to $\mathbb{R}^{10,1\vert\mathbf{32}}_{\mathrm{exc},s}$ this way, means to trade a (rational) 3-sphere
fibration over 11d superspacetime for a high-dimensional (rational) torus fibration.
This is reminiscent of a recent interest in relating toroidal and spherical backgrounds \cite{CCGS, CSk}.

\medskip
 Putting all the above together we arrive at the following global picture:

 \medskip
$$
\hspace{-3mm}
\boxed{
  \xymatrix@=1.55em{
    &
    &
    {
    \fbox
    {
    $
    \begin{array}{cc}
      \mbox{\it \footnotesize Exceptional tangent space}
      \\
      \mbox{\it \footnotesize over super-spacetime}
    \end{array}
    $
    }}
    &&
    {
    \fbox
    {$\begin{array}{cc}
      \mbox{\it \footnotesize Super 2-gerbe}
      \\
      \mbox{\it \footnotesize over super-spacetime}
    \end{array}$}
    }
    \\
    &
    &
    \mathbb{R}^{10,1\vert\mathbf{32}}_{\mathrm{exc},s}
    \ar[rr]_{\mathrm{comp}_s}^-{ { \mbox{\it \tiny Decomposition} \atop  \mbox{\it \tiny of C-field} } }
    &&
    \mathfrak{m}2\mathfrak{brane}
    \ar[dd]|{\;\;\mathrm{hofib}(\mu_{{}_{M2}})}
    \ar@{<->}@/^2pc/[dr]|>>>>>>>>>>>{\phantom{AA} \atop \phantom{AA}}^>>>>>>>>>>>{
    \mbox{\it \tiny Spherical} \atop \mbox{\it \tiny T-self-duality} }
    &
    \\
    &&&&& \mathfrak{m}2\mathfrak{brane}
    \ar[dl]|{\mathrm{hofib}(\mu_{{}_{M2}})}
    \\
    {
    \fbox
    {
    $
    \begin{array}{cc}
      \mbox{\it \footnotesize M5-brane}
      \\
      \mbox{\it \footnotesize worldvolume}
    \end{array}
    $
    }}
    &
    \Sigma_{M5}
    \ar@/_1pc/[rrr]_{\mbox{\it \tiny Plain Sigma-model-field}}
    \ar@{-->}[uurrr]|{ \mbox{ \it \tiny Sigma-model- \& gauge-field } }
    \ar@{..>}@/^1pc/[uur]^<<<<<<<<<<<<<<{ \mbox{\it \tiny U-duality equivariant}
    \atop \mbox{\it \tiny sigma-model- \& gauge-field }  }
    & && \mathbb{R}^{10,1\vert\mathbf{32}} &
    \\
    &&&&
   {
    \fbox
    {
    $
    \begin{array}{cc}
      \mbox{\it \footnotesize Target}
      \\
      \mbox{\it \footnotesize super-spacetime}
    \end{array}
    $
    }}  }
  }
$$

\vspace{6mm}

Finally we show that the spherical T-duality on $\mathfrak{m}2\mathfrak{brane}$ passes along the decomposition
map to spherical T-duality on the exceptional super spacetime
$\mathbb{R}^{10,1\vert\mathbf{32}}_{\mathrm{exc},s}$; this is Prop. \ref{SphericalTDualityOnExtendedSuperspacetime} below. The key step in establishing this
(via Theorem \ref{HigherTDualityForDecomposedFields} below)
is to show that the decomposed C-field on the exceptional super-spacetime $\mathbb{R}^{10,1\vert \mathbf{32}}_{\mathrm{exc},s}$
still allows to distinguish the 3-spherical wrapping modes of the M5-brane that get exchanged with the non-wrapping modes under
spherical T-duality. Prop. \ref{HCohomologyOfDecomposedCFieldVanishesInSubalgebra} below shows that this is the case, except possibly for
summands in the M5-brane charge twisted supercocycles which are multiples in the gravitino field $\psi$ of the 528-volume form
on the exceptional superspacetime; see Remark \ref{InvolvingTheExceptionalVolumeElement} below.

\medskip

Note that a higher spherical version of T-duality had been established
in vast mathematical generality in \cite{LSW}, and in a special case in \cite{BEM}.
In \cite{BEM} the spherical fiber bundle was suggested to be identified with actual spacetime itself,
not a fibration over spacetime; however, the direct relation to string theory or M-theory had remained unclear.
Notice that the spherical T-duality which we discover takes place entirely in M-theory;
notably it is distinct from the strong coupling lift of ordinary T-duality to M-theory
(see \cite{Russo97}\cite{Schwarz}) which is the non-perturbative version of 10-dimensional T-duality,
with one side inevitably involving a string theory. The super $L_\infty$-algebraic formulation
of this strong coupling T-duality, interpreted as F-theory, we had discussed already in \cite[Section 8]{FSS16}.

\medskip

\medskip

\noindent In {\bf summary}, in this article we offer the following insights:
\begin{enumerate}
\item A new duality in M-theory: spherical T-duality for M5-branes (Section \ref{SphericalTDualityOfM5BranesOnM2ExtendedSpacetime}).

\vspace{-1mm}
\item Toroidal T-duality on exceptional super-spacetime (Section \ref{ToroidalTDualityOnExceptionalMTheorySpacetime});

\vspace{-1mm}
\item Clarification of the supersymmetric exceptional generalized geometry of M-theory (Section \ref{ExceptionalGenralisedGeometry});

\vspace{-1mm}
\item Identification of spherical T-duality as a duality on the
exceptional-generalized super-geometry
(Section \ref{SphericalTDualityOnExceptionalMTheorySpacetime});

\vspace{-1mm}
\item Parity symmetry as an isomorphism in M5-brane charge twisted cohomology akin
to spherical T-duality (Section \ref{ParitySymmetry});

\end{enumerate}

\medskip
\noindent We provide the above physical insights within the proper mathematical setting, developed in Section \ref{sphericaltduality}. We highlight the power of cohomological techniques in the supergeometric setting,
including C-cohomology to study T-duality for decomposed
C-fields in Section \ref{HigherTDualityForDecomposedFormFields}.
Similar techniques in other contexts have allowed, for instance, for eleven-dimensional supergravity to be recovered from
the Spencer cohomology of the Poincar\'e superalgebra \cite{FOS}.

\newpage

The {\bf outline} of the article is as follows: Section 2 is mainly to fix concepts and notations; the reader roughly familiar
with super $L_\infty$-algebras and their cohomology may want to skip this section and start with Section 3, coming back to Section 2 when the
need arises. Section 3 contains the main mathematical results. They are presented in a general form, with all examples from M-brane theory postponed to Section 4. The reader principally interested into these examples is invited to start directly from Section 4, going back to Section 3 for the proofs of the general statements used in Section 4.

\section{Super $L_\infty$-Homotopy theory}
\label{SuperLInfinityHomotopyTheory}

We work in the homotopy theory of super $L_\infty$-algebras as in \cite{FSS13, FSS15, FSS16a, FSS16, T, Schreiber17b}
(see Remark \ref{RemarkOnTerminology} below on differing terminology).
Here we set up the basics that we need in the following sections.

\medskip
The bosonic sector of super $L_\infty$-homotopy theory is
a model for \emph{rational homotopy theory} (for review see e.g. \cite{Hess06} or \cite[section A]{FSS16a}),
where topological spaces as well as spectra parameterized over them \cite{Bra18} are studied in the coarse-grained
perspective that regards two of them as essentially the same as soon as there is a map between them
that induces an isomorphism on all rationalized homotopy groups.
This amounts to disregarding (for the time being) all information contained in torsion-subgroups of homotopy group and to retaining only the information
that may be represented by differential form data.
Notably the Chevalley-Eilenberg algebras of nilpotent $L_\infty$-algebras are Sullivan's models for rational homotopy types \cite{Sullivan77}.

\medskip
For example, in rational homotopy theory the spheres of odd dimension $2n+1$ are equivalent to Eilenberg-MacLane spaces concentrated in this odd degree:
\(
\label{Q-EM}
  \xymatrix{
    K(\Z, 2n+1) \ar[rr]^-{\rm rationalize}
    \ar@{<-}@/_1.8pc/[rrrr]_{\rm rational\;equivalence}
    && K(\Q, 2n+1) && S^{2n+1} \ar[ll]_-{\rm generator}
   }
\)
and both are algebraically represented by the simplest possible Sullivan model, namely by the differential-graded commutative algebra
(dgc-algebra) that has a single generator $c$ in degree $2 +1$, and whose differential vanishes: $d c = 0$.

\medskip
The observation that rational homotopy theory sits inside the homotopy theory of $L_\infty$-algebras
is implicit already in \cite{Quillen69}, but was made fully explicit only in \cite{Hinich},
on which the modern model \cite{pridham} is based.
A review of rational homotopy theory from the perspective of $L_\infty$-algebras is given
 in \cite[section 2]{BuijsFelixMurillo12}.
Therefore, the homotopy theory of super $L_\infty$-algebras may be regarded as a
model for rational supergeometric homotopy theory.

\medskip
In the supergravity literature which goes back to \cite{DF, CDF}, the Chevalley-Eilenberg algebras
of super $L_\infty$-algebras are known
as ``FDA''s, following \cite{Nie82}. In supergravity these serve to neatly unify supersymmetry super Lie algebras
(such as super Poincar{\'e} algebras) with the higher degree form fields that are crucial ingredients of
higher dimensional supergravity theories.

\medskip
From the point of view of super homotopy theory this phenomenon is interpreted \cite{FSS13} as saying that super-Minkowki super Lie algebras
carry a finite number of exceptional $\mathrm{Spin}(p,1)$-invariant cohomology classes (Section \ref{SuperLInfinityCohomology} and Section \ref{EquivariantSuperLInfinityCohomology} below) that iteratively classify a Whitehead-like tower  of higher central extensions (Section \ref{TwistedSuperLInfinityCohomology} below), analogous to
the case of connected covers of Lie groups studied in \cite{SSS2}\cite{SSS3}\cite{9brane}, which interestingly
admit extensions to the Lorentzian case \cite{SS} and to the rational case \cite{SW} as we consider here.
The tower we construct has higher equivariant connections which are precisely the higher WZW terms
of the super-$p$-branes appearing in string/M-theory; see the diagram on page \pageref{TheBraneBouquet}.

\medskip
In \cite{FSS16} we had shown that not only the brane content and brane intersection laws may be ``read off'' from super homotopy this
way, but the ``super-topological'' T-duality between type IIA and type IIB super F1/D$p$-branes may be discovered
(in fact together with the very axioms of topological T-duality themselves)
establishing a kind of reflection symmetry in the above \emph{brane bouquet} diagram.
It is curious that all vertical-going arrows in this diagram are given from first principles:
they are the maximal $R$-symmetry invariant higher extension
in each case \cite{HuertaSchreiber}. We discuss this phenomenon further below in Section \ref{ToroidalTDualityOnExceptionalMTheorySpacetime}.

\newpage

\paragraph{Basic concepts and notation used.}
	\begin{itemize}

\item $\mathfrak{g}$: A (super-) $L_\infty$-algebra, i.e. an algebraic structure akin to a super Lie algebra but
with brackets and higher Jacobi relations of any higher degree.
A (super-) $L_\infty$-algebra  encodes the structure of infinitesimal (super-)symmetries
and of ever higher order infinitesimal (super-)symmetries of infinitesimal (super-)symmetries.
A key example for us is $\mathfrak{g}=$ extended Minkowski super spacetime.
For background on (super-)$L_\infty$-algebras in the context that we use see \cite{SSS1, BaezHuerta11, FSS16}.

\item $b^n \mathfrak{u}_1 \simeq {b^n \R}$:
The super $L_\infty$-algebras which are the higher versions (``deloopings'') of the abelian Lie algebras
 $\mathbb{R} \simeq \mathfrak{u}_1$. The
underlying chain complex has a copy of $\mathbb{R}$ in degree $n$, and all brackets vanish.
These are the rational models for Eilenberg-MacLane spaces $K(\mathbb{Z}, n+1)$. Just as the latter
 classify ordinary integral cohomology, these $L_\infty$-algebra serve as classifying objects
 for super $L_\infty$-cohomology; see Def. \ref{SuperLInfinityCocycles} below.

\item ${\rm CE}(\mathfrak{g})$: The Chevalley-Eilenberg algebra of the super $L_\infty$-algebra
 $\mathfrak{g}$, i.e., a differential graded-commutative (DGC) superalgebra of elements dual to
 $\frak{g}$ whose differential encodes the brackets and higher brackets on $\mathfrak{g}$.
 The CE-algebra may be presented as
 $$
 {\rm CE}(\mathfrak{g})=[{\rm generators}]/(\text{\rm differentials of generators})\;.
 $$
 These are also known as ``FDA''s  in some of the supergravity
literature; see Remark \ref{RemarkOnTerminology} below.

\item  $\mu \in \mathrm{CE}(\mathfrak{g})$ a super $L_\infty$-cocycle, hence a closed
 elements in the Chevalley-Eilenberg algebra.

\item $\hat{\mathfrak{g}}$: A higher central extension of super-$L_\infty$-algebra $\mathfrak{g}$.
A key example is $\widehat{\mathfrak{g}} =\frak{m}2\frak{brane}$, which is the extension of 11d
super Minkowski spacetime by the M2-brane supercocycle $\mu_{{}_{M2}}$.

\item ${\rm hofib}(\mu_n)$: Homotopy fiber of a super $L_\infty$-cocycle $\mu$ of degree $n$.
The latter may be viewed as a map  to the classifying algebra $b^n \R$.
$$
\xymatrix@R=1.5em{
b^{n-1} \mathfrak{u}_1 \ar[rr]\ar[d] && \hat{\frak{g}} \ar[d]|{ \mathrm{hofib}(\mu) } \ar[rr]
&& \ast \ar[d]
\\
\ast\ar[rr]&& \mathfrak{g}  \ar[rr]^-{\mu} && b^{n}\mathfrak{u}_1
}
$$
See Prop. \ref{hofibofSuperLInfinityAlgebras} below.
This is the super $L_\infty$-algebra counterpart of principal bundles obtained as pullback of universal bundles
$$
\xymatrix@R=1.5em{
G \ar[rr]\ar[d] && P\ar[rr] \ar[d] && \,\,EG \ar[d]
\\
\ast\ar[rr] && X \ar[rr] && BG
}
$$
for $G$ a (higher) abelian topological group \cite{NSS12}.

\item $\mu_{{}_M}$: Cocycle corresponding to an M-brane; of degree 4 for the M2-brane and degree 7
for the M5-brane; see Example \ref{HigherTDualCorrespondenceForM2Brane} below.
Jointly these are valued in the rational 4-sphere \cite{top}\cite{FSS15}\cite{FSS16a}.

\item $\mathrm{dd}_n$: Cocycle regarded as the rational version of a higher Dixmier-Douady class (see \cite{DD} for a very readable
account) generalized from to higher degrees, as described in \cite{extended, 3stack, stacky}.
The original class ${\rm dd}_3$ is in turn a generalization of the Chern class of a line bundle.

\item {\it Higher torus}: A product of shifted circles, i.e. of $b^n \mathfrak{u}_1$'s ; see Def. \ref{HigherCupProducts} below.

\item $\widevec{\mathrm{dd}}$: A tuple of cocycles, in our case a $k$-vector of $(2n+1)$-cocycles
classifying an extension by a higher $k$-torus of degree $2n+1$ ; see Def. \ref{HigherCupProducts} below.


\item $H^{i + \mu}$: Cohomology in degree $i$ twisted by the cocycle $\mu$, see Def. \ref{TwistedCohomology} below. Other notations for
twisted cohomology include $H^i_\mu(-)$ and $H^i(-;\mu)$ but the first might be confused with equivariant
cohomology (which we use) and the second would lead to cumbersome notation when we introduce coefficients to
our cohomology groups.

\item $s \mathrm{LieAlg}^\mathrm{fin}_{\mathbb{R}}$: The category of finite dimensional super Lie
algebras, i.e., the collection of all super Lie algebras with appropriate homomorphisms between then.

\item $sL_\infty \mathrm{Alg}^{\mathrm{fin}}_\mathbb{R}$: The category of super $L_\infty$-algebras
of finite type, i.e., the collection of all super $L_\infty$-algebras with appropriate homomorphisms between them.

\item  $\mathrm{sdgcAlg}_{\mathbb{R}}^{\mathrm{op}}$:    The category of super differential graded-commutative
(sCDG) algebras, i.e., the collection of all super CDG-algebras with appropriate homomorphisms between them.

\end{itemize}

\newpage

\subsection{Super $L_\infty$-Cohomology}
\label{SuperLInfinityCohomology}

The cocycles that we encounter are built out of bosonic and fermionic fields
and are closed under an appropriate differential. We now provide the proper setting for describing such fields or cocycles,
namely \emph{super $L_\infty$-algebras}. For details we refer the reader to \cite[section 2]{FSS16}
and references therein; see also Remark \ref{RemarkOnTerminology} below on differing terminology.

\medskip
To every finite-dimensional super Lie algebra $(\mathfrak{g},[-,-])$ one associates its \emph{Chevalley-Eilenberg algebra}
$\mathrm{CE}(\mathfrak{g})$, which is the free $(\mathbb{Z}, \mathbb{Z}/2)$-bigraded-commutative algebra
$\Exterior^\bullet(\mathfrak{g})^\ast$ equipped with the differential $d_{\mathfrak{g}} := [-,-]^\ast$,
which on generators is the linear dual of the super Lie bracket, and from there uniquely extended as graded derivation
of bidegree $(1,\mathrm{even})$.

\begin{example}[Translational supersymmetry super Lie algebra]
\label{superMinkowskiSuperLieAlgebra}
A key class of examples is the Lorentzian supersymmetry super Lie algebras which are specified by a spacetime
dimension $d =p+1$ and a choice of real representation $\mathbf{N}$ (of real dimension $N \in \mathbb{N}$)
of the corresponding Spin group $\mathrm{Spin}(p,1)$. Their translational part may be thought of as the
corresponding $(p+1)$-dimensional and ``$N$-supersymmetric'' super-Minkowski spacetime $\mathbb{R}^{p,1\vert \mathbf{N}}$ equipped with its super-translation
super Lie action on itself. From this point of view the corresponding Chevalley-Eilenberg algebra is generated from
the standard super-left invariant super-vielbein
$$
  (e^a = d x^a + \overline{\theta} \Gamma^a d \theta,\; \psi^\alpha = d \theta^\alpha)
$$
and the CE-differential is given on generators by the torsion
constraint equation (\ref{VielbeinDifferential}):
\begin{equation}
   \label{SupermMinkowskiCE}
   \mathrm{CE}\big(
     \mathbb{R}^{p,1\vert \mathbf{N}}
   \big)
   \;=\;
   \mathbb{R}
   \Big[ \;\;
     \hspace{-5mm}\underset{ \mathrm{deg} = (1,\mathrm{even}) }{\underbrace{(e^a)}}\hspace{-4mm}_{a \in \{0,\cdots, p\}} \; , \;
     \hspace{-4mm}\underset{ \mathrm{deg} = (1,\mathrm{odd}) }{\underbrace{(\psi^\alpha)}}\hspace{-3mm}_{\alpha \in \{1, \cdots, N\}}
   \Big]
   /
   \left(
     {\begin{array}{l} d e^a = \overline{\psi} \Gamma^a \psi \\ d \psi^\alpha = 0\end{array}}
   \right)
   \,.
\end{equation}
Here
$$
  \overline{(-)}\Gamma (-)
  \;:\;
  \mathbf{N} \otimes \mathbf{N}
    \xymatrix{\ar[r]&}
  \mathbb{R}^{p,1}
$$
denotes the bilinear spinor-to-vector pairing that is canonically associated with ever real spin representation.
\end{example}

The operation that takes a finite-dimensional super Lie algebra $\mathfrak{g}$ to its Chevalley-Eilenberg
algebra $\mathrm{CE}(\mathfrak{g})$ turns out to be a fully-faithful embedding into the opposite of
differential $(\mathbb{Z},\mathbb{Z}/2)$-bigraded commutative algebras, hence super DGC-algebras.
\footnote{This may be indicated as
$
  \mathrm{CE}
  \;:\;
  s \mathrm{LieAlg}^\mathrm{fin}_{\mathbb{R}}
  \hookrightarrow
  \mathrm{sdgcAlg}_{\mathbb{R}}^{\mathrm{op}}
  \,.
$
}
This means that a homomorphism of super Lie algebras is equivalently a
super DGC (differential graded commutative) algebra homomorphism of their
CE-algebras in the other direction
$$
  \frac{
    \xymatrix{
      {\phantom{CE(}}
      \mathfrak{g}_1
      {\phantom{)}}
     \ar[rr]
      &&
      {\phantom{CE(}}
      \mathfrak{g}_2
      {\phantom{)}}
    }
  }{
    \xymatrix{
      \mathrm{CE}(\mathfrak{g}_1)
     \ar@{<-}[rr]
      &&
      \mathrm{CE}(\mathfrak{g}_2)
    }
  }
$$
This makes it evident that there is a generalization of finite dimensional super Lie algebras to
\emph{super $L_\infty$-algebras} $\mathfrak{g}$ of finite type which may be \emph{defined}
to be the formal duals of super dgc-algebras $\mathrm{CE}(\mathfrak{g})$ whose underlying
graded-commutative algebra is \emph{free}, i.e., is a super-graded Grassmann-algebra
\cite[Def. 13]{SSS1}.
\footnote{We write this as
$
  \mathrm{CE}
  \;:\;
  s L_\infty \mathrm{Alg}^\mathrm{fin}_{\mathbb{R}}
  \hookrightarrow
  \mathrm{sdgcAlg}_{\mathbb{R}}^{\mathrm{op}}
  \,.
$
}

\begin{remark}[Differing terminology for super $L_\infty$-algebras]
\label{RemarkOnTerminology}
The history of the concept of (super-)$L_\infty$-algebras is a bit interwined, which
tends to hide the great unity of the subject behind the different terminology of disjoint schools.
Traditionally the concept of $L_\infty$-algebras is attributed to Stasheff (see \cite{LS}),
who had introduced $A_\infty$-algebras much earlier. But, in fact, Stasheff indicates that he
got the concept from Zwiebach (see \cite[slide 17]{Stasheff16}), who had discovered infinite-dimensional
bosonic $L_\infty$-algebra in closed string field theory in 1989. However, the \emph{evident}
linear dualization \cite[Def. 13]{SSS1} allows us to interpret it as arising a decade earlier in the
supergravity literature with \cite{Nie82,DF}, where the CE-algebras of finite-type super
$L_\infty$-algebras are referred  to as ``FDA''s. This somewhat non-standard terminology may be
one cause that the  ubiquity of super $L_\infty$-algebra theory in supergravity and superstring
theory remains under-appreciated, even with the recent renewed
interest in $L_\infty$-algebras, for instance in \cite{HohmZwiebach17}.
\end{remark}

\medskip
A first curious fact about (super-)$L_\infty$-algebras is that even if one starts out being interested
just in (super-)Lie algebras, the concept of (super-)$L_\infty$-algebras serves to provide classifying
``spaces'' for (super)Lie algebra cohomology.
This simple but powerful change of perspective is paramount for much of our discussion.
In the following the notation ${\rm dd}_{n+1}$ is meant to indicate the generalization of the
Dixmier-Douady class from degree 3 to degree $n+1$.

\begin{defn}[Super $L_\infty$-cocycles]
  \label{SuperLInfinityCocycles}
  For $n \in \mathbb{N}$, write $b^n \mathfrak{u}_1 \in sL_\infty \mathrm{Alg}^{\mathrm{fin}}_\mathbb{R}$
  for the super $L_\infty$-algebra dually  given by
  $$
    \mathrm{CE}( b^n \mathfrak{u}_1 )
    :=
    \mathbb{R}[\hspace{-4mm}  \underset{ \mathrm{deg} = (n+1, \mathrm{even}) }{\underbrace{\mathrm{dd}_{n+1}}} \hspace{-4mm} ]/( d (\mathrm{dd}_{n+1}) = 0 )
    \,.
  $$
  Hence for $\mathfrak{g} \in sL_\infty \mathrm{Alg}^{\mathrm{fin}}_{\mathbb{R}}$ any super $L_\infty$-algebra, we have that morphisms
  from $\mathfrak{g}$ to $b^n \mathfrak{u_1}$ are in natural bijection to closed elements of degree $n+1$ in the Chevalley-Eilenberg algebra of $\mathfrak{g}$
  $$
   \left\{
      \mathfrak{g}
        \overset{\mu}{\longrightarrow}
      b^n \mathfrak{u}_1
   \right\}
   \;\simeq\;
   \left\{
    \mu \in \mathrm{CE}^{n+1}(\mathfrak{g})
    \,\vert\,
    d_{\mathfrak{g}} \mu = 0
   \right\}
    \,.
  $$
  Such a $\mu$ is a \emph{super $L_\infty$-cocycle} of degree $n+1$ on $\mathfrak{g}$.
\end{defn}

We would like to account for the occurrence of the fields in particular degrees
with a specific spacing. We interpret this as having fields in a certain rational
{\it periodic} cohomology theory, that we collectively call $K(t)$, where
$t$ is the periodicity parameter. This includes and generalizes the set-up of
rational K-theory in \cite{T}, where $t$ is the usual Bott periodicity parameter.
One can account for degrees by taking the suspension, i.e. $K^n=\Sigma^n K$.

\begin{defn}[Periodic super $L_\infty$-cohomology]
\label{PeriodicLInfinityCohomology}
For $t, n \in \mathbb{N}$ with $t \geq 1$,
let $\mathfrak{l}(\Sigma^n K(t))$ be the super $L_\infty$-algebra defined by
$$
  \mathrm{CE}(\mathfrak{l}(\Sigma^n K(t)))
  :=
  \mathbb{R}\big[
    \big(\hspace{-5mm}
      \underset{\mathrm{deg} = (2 k t + n, \mathrm{even}) }{\underbrace{ \omega_{2 k t + n} }}
      \hspace{-5mm}
      , k \in \mathbb{Z}
    \big)
  \big]
  /
  \left(
    d \omega_{(2 k t + n)} = 0
  \right)\;.
$$
\end{defn}

\medskip
Hence for $\mathfrak{g} \in \mathrm{sL}_\infty\mathrm{Alg}$ any super $L_\infty$-algebra, a morphism
$$
  \xymatrix{
    \mathfrak{g}
    \ar[rr]^-{\omega_{2 t \bullet + n }}
    &&
    \mathfrak{l}( \Sigma^n K(t) )
  }
$$
is equivalently a sequence of super $L_\infty$-cocycles on $\mathfrak{g}$, according to Def. \ref{SuperLInfinityCocycles} of degrees $n$ mod $2t$.
We write
\begin{equation}
  \label{PeriodictSuperLInfinityCohomology}
  H^{\bullet \;\mathrm{mod}\; 2t}(\mathfrak{g}/K)
\end{equation}
for the corresponding periodic cohomology groups.

\subsection{Equivariant super $L_\infty$-cohomology}
\label{EquivariantSuperLInfinityCohomology}

The setting we have will involve an action of a group. We now describe
the proper way  to account for that in our framework.

\begin{defn}[Quotient of super $L_\infty$-algebra by group action]
  \label{SuperLInfinityAlgebraWithGroupAction}
  For $\mathfrak{g}$ a super $L_\infty$-algebra, an \emph{action} of a group $K$ on $\mathfrak{g}$
  is a linear group action on the underlying graded vector space which preserves the bi-grading
  $$
    \rho
    \;:\;
    K \times \mathfrak{g}_\bullet \longrightarrow \mathfrak{g}_\bullet
    \,,
  $$
  such that its induced dual action
  $$
    (\rho(-))^\ast
    \;:\;
    K
      \times
    \Exterior^\bullet( \mathfrak{g}^\ast )
      \longrightarrow
    \Exterior^\bullet( \mathfrak{g}^\ast )
  $$
  is compatible
  with the Chevalley-Eilenberg differential $d_{\mathrm{CE}}$, in that for all $k \in K$ we have
  $$
    d_{\mathrm{CE}} \circ \rho(k)^\ast
    =
    \rho(k)^\ast \circ d_{\mathrm{CE}}
    \,.
  $$
  This means that the subspace of $K$-invariant elements in the CE-algebra is a sub-dgc-algebra,
  to be denoted
  \begin{equation}
    \label{KInvariantCEComplex}
    \xymatrix{
    \mathrm{CE}(\mathfrak{g})^K
    \; \ar@{^{(}->}[r] &
    \mathrm{CE}(\mathfrak{g})\;.
    }
  \end{equation}

  \medskip
  We may think of this as the CE-algebra of the quotient $\mathfrak{g}/K$, which is thereby defined.
  Accordingly, given two super $L_\infty$-algebras $\mathfrak{g}_1$, $\mathfrak{g}_2$
  equipped with actions by groups $K_1$ and $K_2$, respectively, then we say that a homomorphism
  $$
    \mathfrak{g}_1/K_1
    \xrightarrow{\;\;\;\phi\;\;\;}
    \mathfrak{g}_2/K_2
  $$
  between them is equivalently a dgc-algebra homomorphism the other way around, between
  their invariant CE-algebras (\ref{KInvariantCEComplex}):
  $$
    \mathrm{CE}\left(\mathfrak{g}_1\right)^{K_1}
    \xleftarrow{\;\;\;\phi^\ast\;\;\;}
    \mathrm{CE}\left( \mathfrak{g}_2\right)^{K_2}
    \,.
  $$
\end{defn}

\medskip
\begin{example}[Invariant super $L_\infty$-cocycle]
  \label{InvariantCocycle}
  If $K = 1$ is the trivial group, then every super $L_\infty$-algebra $\mathfrak{g}$ carries a unique action by that group
  and is canonically identified with the quotient $\mathfrak{g}/1$.
Under this identification, if $\mathfrak{g}$ now is equipped with a general group action, then a homomorphism of the form
  $$
    \mu/K
    \;:\;
    \mathfrak{g}/K
      \longrightarrow
    b^{n}\mathfrak{u}_1
  $$
  is equivalently a super $L_\infty$-cocycle according to Def. \ref{SuperLInfinityCocycles} which in addition is $K$-invariant, namely a plain homomorphism $\mu$
  such that
  $$
    \rho(k)^\ast(\mu) = \mu
    \phantom{AAA}
    \xymatrix@R=1.6em{
      \mathfrak{g}
      \ar[rr]^{\rho(k)}
      \ar[dr]_\mu
      &&
      \mathfrak{g}
      \ar[dl]^\mu
      \\
      & b^n \mathfrak{u}_1
    }
  $$
  for all  $k \in K$.
  Notice that, more generally, one may consider super $L_\infty$-cocycles which are not necessarily $K$-invariant, but which
  are \emph{$K$-equivariant}. This means first of all that the cocycle is $K$ invariant only up to specified homotopies
  $$
    \eta_k \;:\;\rho(k)^\ast(\mu) \Rightarrow \mu
    \phantom{AAA}
    \xymatrix{
      \mathfrak{g}
      \ar[rr]^{\rho(k)}_{\ }="s"
      \ar[dr]_\mu^{\ }="t"
      &&
      \mathfrak{g}
      \ar[dl]^\mu
      \\
      & b^n \mathfrak{u}_1
      \ar@{=>}^{\eta_k} "s"; "t"
    }
  $$
  such that, moreover, these homotopies are compatible up to specified higher homotopies
  $$
    \eta_{k_1} \cdot  \rho(k_1)^\ast \eta_{k_2}
      \xymatrix{\ar@{=>}[r] &}
    \eta_{k_1 k_2}
  $$
  and so on.
  In a broader context of higher supergeometry one may sum this up by saying all this data is equivalently a homomorphism
  out of the \emph{homotopy quotient} of $\mathfrak{g}$ by $K$, denoted
  \begin{equation}
    \label{EquivariantCocycle}
    (\mu, \eta, \cdots)
    \;:\;
    \mathfrak{g}//K
      \longrightarrow
    b^{n} \mathfrak{u}_1
    \,.
  \end{equation}
\end{example}

\begin{defn}[$K$-invariant super $L_\infty$-cohomology]
  \label{KInvariantLInfinityCohomology}
  Given a super $L_\infty$-algebra $\mathfrak{g}$ equipped with an action by a group $K$ (Def. \ref{SuperLInfinityAlgebraWithGroupAction}), then
  its \emph{$K$-invariant super $L_\infty$-cohomology}
  $$
    H^\bullet(\mathfrak{g}/K)
    \;:=\;
    H^\bullet\left(
      \mathrm{CE}(\mathfrak{g})^K
    \right)
  $$
  is the $(\mathbb{Z} \times (\mathbb{Z}/2))$-bigraded cochain cohomology groups of the $K$-invariant
   subcomplex (\ref{KInvariantCEComplex}) of its Chevalley-Eilenberg algebra.
\end{defn}
\begin{remark}[Different notions of equivariant cohomology]
  {\bf (i)} One may also consider the group $H^\bullet(\mathfrak{g}//K)$ of equivalence classes of equivariant cocycles (\ref{EquivariantCocycle})
  as well as the subgroup $\left(H^\bullet(\mathfrak{g})\right)^K \hookrightarrow H^\bullet(\mathfrak{g}) := H^\bullet( \mathrm{CE}(\mathfrak{g}) )$ of
  those cohomology classes in the full Chevalley-Eilenberg complex which are invariant under $K$.
  There are canonical comparison maps to these from the group of Def. \ref{KInvariantLInfinityCohomology}
  \begin{equation}
    \label{ComparisonMapsBetweenEquivariantCohomologyGroups}
    H^\bullet(\mathfrak{g}/K)
      \longrightarrow
    H^\bullet(\mathfrak{g}//K)
      \longrightarrow
    H^\bullet( \mathfrak{g} )^K
    \,,
  \end{equation}
  where the first one regards an invariant cocycle as an equivariant cocycle with trivial equivariance data, and the
  second forgets the choice of equivariance data.

 \item {\bf (ii)} For (super-)Lie algebras (i.e., (super-)Lie 1-algebras) the study of the composite comparison map in (\ref{ComparisonMapsBetweenEquivariantCohomologyGroups})
  is the topic of the Hochschild-Serre spectral sequence, which may be used to extract sufficient conditions for the total
  comparison map to be an isomorphism. One such sufficient condition is that $K$ is a compact topological group (which is however not the case for the
  Lorentzian spin groups $K = \mathrm{Spin}(p,1)$ of interest below.)

 \item {\bf (iii)} But notice that from the point of view of equivariant homotopy theory the group $H^\bullet( \mathfrak{g} )^{K}$ has no intrinsic meaning in itself,
  since its elements are just ``in-coherently equivariant''  cocycles.

\item {\bf (iv)}  In contrast, the group $H^\bullet(\mathfrak{g}/K)$ does have intrinsic meaning in equivariant homotopy theory, despite superficial appearance,
  namely in the context of what is called \emph{Bredon equivariant homotopy theory} \cite[section 5.1]{Rezk14}.
  This is the group we will be considering here.
\end{remark}

\subsection{Twisted super $L_\infty$-cohomology}
\label{TwistedSuperLInfinityCohomology}

Our setting will also involve twists, so twisted versions of the above constructions are needed.
The main statements below in Theorem \ref{HigherTDuality}, Theorem
\ref{HigherTDualityForDecomposedFields} and Cor.
 \ref{HigherTopologicalTDualityInTwistedCohomology}
(with various examples in section \ref{Examples})
establish isomorphisms in \emph{twisted} invariant super $L_\infty$-cohomology.
On Chevalley-Eilenberg algebras this concept of twisted cohomology is straightforward,
made explicit by Def. \ref{TwistedCohomology} below.

\medskip
A key example of twisted super $L_\infty$-cocycles are the super-WZW-terms for the
F1/D$p$-branes on type II super-Minkowski spacetime \cite{FSS16a, FSS16}, recalled
below in Section \ref{OrdinaryTypeIITDuality}. These may be extracted, via Prop. \ref{TwistedCohomologyInjects} below, from untwisted cocycles
on higher central extensions (Prop. \ref{hofibofSuperLInfinityAlgebras} below) of
super-Minkowski spacetimes, as found originally in \cite{CAIB00, IIBAlgebra}.
This transformation of Prop. \ref{TwistedCohomologyInjects} is an example of a general equivalence \cite[Theorem 4.39]{NSS12}
between twisted cohomology and non-twisted but higher equivariant cohomology on the
extension classified by the twist, we make this homotopy-theoretic perspective explicit in Prop. \ref{FromTwistedCocyclesToUntwistedCocyclesOnTheExtensionClassifiedByTheTwist} below.

\begin{defn}[Twisted super $L_{\infty}$-algebra cohomology]
  \label{TwistedCohomology}
  Let $\mathfrak{g}$ be a super $L_\infty$-algebra equipped with the action of a group $K$ (Def. \ref{SuperLInfinityAlgebraWithGroupAction}), and let
  $\mathfrak{g}/K
      \xrightarrow{\;\;\mu\;\;}
      b^{2t}\mathfrak{u}_1
 $,
   i.e.,
   $
   \mu \in \mathrm{CE}(\mathfrak{g})^K
  $
  be a $K$-invariant cocycle (Example \ref{InvariantCocycle}).
  Then
  the Chevalley-Eilenberg differential $d_{\mathfrak{g}}$ plus the wedge product with $\mu$
  defines a differential of degree 1 mod $2t$
   $$
    d_{\mathfrak{g}}
    +
    \mu\wedge
    \;\colon\;
    \underset{k \in \mathbb{Z}}{\bigoplus} \mathrm{CE}^{2 k t +\bullet}(\mathfrak{g})^K
    \xymatrix{\ar[r] &}
    \underset{k \in \mathbb{Z}}{\bigoplus}
    \mathrm{CE}^{2 k t + \bullet + 1}(\mathfrak{g})^K
    \,.
  $$
  The cochain cohomology of this differential is the \emph{$K$-invariant $\mu$-twisted super $L_\infty$-cohomology} of $\mathfrak{g}$
  $$
    H^{n + (\mu)}(\mathfrak{g}/K)
    \;:=\;
    \frac{
      \mathrm{ker}
      \Big(
        \underset{k \in \mathbb{Z}}{\bigoplus} \mathrm{CE}^{2 k t + n }(\mathfrak{g})^K
          \xrightarrow{\;d_{\mathfrak{g}} + \mu\wedge\; }
        \underset{k \in \mathbb{Z}}{\bigoplus} \mathrm{CE}^{2 k t + n + 1 }(\mathfrak{g})^K
      \Big)
    }{
      \mathrm{im}
      \Big(
        \underset{k \in \mathbb{Z}}{\bigoplus} \mathrm{CE}^{2 k t + n-1 }(\mathfrak{g})^K
         \xrightarrow{\;d_{\mathfrak{g}} + \mu\wedge\; }
        \underset{k \in \mathbb{Z}}{\bigoplus} \mathrm{CE}^{2 k t + n }(\mathfrak{g})^K
      \Big)
    }\,.
  $$
\end{defn}
The concept of twisted super $L_\infty$-cohomology (Def. \ref{TwistedCohomology}) is closely related to the
non-twisted cohomology of higher central extensions:
\begin{prop}[Homotopy fiber functor {\cite[Theorem 3.8]{FSS13}}, based on {\cite[Theorem 3.1.13]{FRS13}}]
  \label{hofibofSuperLInfinityAlgebras}
  Let $\mathfrak{g} \in sL_\infty\mathrm{Alg}$ be a super $L_\infty$-algebra and let
  $
    \mu \;:\; \mathfrak{g} \longrightarrow b^{n} \mathfrak{u}_1
  $
  be an $n+1$-cocycle on it. Then

  \item {\bf (i)} A model for its \emph{homotopy fiber}
  \begin{equation}
    \label{HomotopyFiberProjection}
    \xymatrix@R=1.2em{
      \hat{\mathfrak{g}}
      \ar[d]^{\mathrm{hofib}(\mu)}
      \\
      \mathfrak{g}
    }
  \end{equation}
  is the super $L_\infty$-algebra dually given by adjoining to the CE-algebra of $\mathfrak{g}$ a generator $b$
  of degree $n$ which trivializes the cocycle:
  $$
    \mathrm{CE}(
      \hat{\mathfrak{g}}
    )
    \;:=\;
    \mathrm{CE}(\mathfrak{g})[b]/(d b = \mu)
    \,.
  $$
\item {\bf (ii)}   This construction clearly extends to a functor
  $$
    \mathrm{hofib}\;:
    \xymatrix{
    sL_\infty\mathrm{Alg}/b^n \mathfrak{u}_1
    \ar[r] &}
    sL_\infty\mathrm{Alg}
  $$
  from super $L_\infty$-algebras over $b^n \mathfrak{u}_1$ to plain super $L_\infty$-algebras.
  \item {\bf (iii)} If $\mathfrak{g}$ is equipped with an action by a group $K$ (Def. \ref{KInvariantCEComplex}) and if the cocycle is
$K$-invariant, $\mu \in \mathrm{CE}(\mathfrak{g})^K$ (\ref{KInvariantCEComplex}), then $\widehat{\mathfrak{g}}$
inherits a $K$-action, such that the projection (\ref{HomotopyFiberProjection}) respects the $K$-actions.

We also say that $\hat{\mathfrak{g}}$ in Def. \ref{hofibofSuperLInfinityAlgebras} is the \emph{higher central extension of $\mathfrak{g}$ classified by $\mu$}.
\end{prop}
\begin{remark}
  The consideration of higher central extensions of supersymmetry super Lie algebras
  may be identified \cite{FSS13}, under the translation provided by Prop. \ref{hofibofSuperLInfinityAlgebras}, as the core tool
  for supergravity and superstring theory that was established in \cite{DF, CDF}, there referred to as the
  ``Free Differential Algebra" or ``FDA'' approach.
\end{remark}

The following proposition compares the concepts of twisted super $L_\infty$-cohomology (Prop. \ref{TwistedCohomology})
with the non-twisted cohomology of higher central extensions of super $L_\infty$-algebras (Prop. \ref{hofibofSuperLInfinityAlgebras})
at the purely algebraic level.

\begin{prop}[Twisted cohomology maps into the periodic cohomology of the higher central extension]
  \label{TwistedCohomologyInjects}
  Let $\mathfrak{g}$ be a super $L_\infty$-algebra equipped with an action by a group $K$ (Def. \ref{SuperLInfinityAlgebraWithGroupAction}) and
  let $\mu \in \mathrm{CE}(\mathfrak{g})^K$, $d_{\mathfrak{g}} \mu = 0$ be a $K$-invariant cocycle
  of degree $2t +1$ (Example \ref{InvariantCocycle}) for $t \geq 1$. Then there is an injection
  $$
    H^{\bullet + \mu}( \mathfrak{g}/K )
    \xymatrix{\ar@{->}[r]&}
    H^\bullet( \widehat{\mathfrak{g}}/K )
  $$
  from the $K$-invariant $\mu$-twisted cohomology of $\mathfrak{g}$ (Def. \ref{TwistedCohomology}) to the non-twisted
  periodic $K$-invariant cohomology (Def. \ref{PeriodicLInfinityCohomology})
  of the higher central extension $\widehat{\mathfrak{g}}$ classified by $\mu$ according to Prop. \ref{hofibofSuperLInfinityAlgebras}.
\end{prop}
\begin{proof}
  For any degree $n \in \mathbb{Z}$ consider the following linear map on cochains:
  \begin{equation}
    \label{CochainMapForInclusionOfTwistedCohomologyIntoNonTwistedPeriodicCohomologyOfHigherCentralExtension}
    \hspace{-1.8cm}
    \xymatrix{
      \underset{k \in \mathbb{Z}}{\bigoplus} \mathrm{CE}^{2 k t + n}(\mathfrak{g})
      \ar[rrrrrrr]^-\phi
      &&&&&&&
      \underset{k \in \mathbb{Z}}{\bigoplus} \mathrm{CE}^{2 k t + n}(\widehat{\mathfrak{g}})
      }
      \end{equation}
        \vspace{-4mm}
     $$
     \xymatrix{\left\{
        \omega_{2kt + n}
      \right\}_{ k \in \mathbb{Z} }
      \ar@{|->}[rrr]^-{\phi}
      \ar[dd]_{ d_{\mathfrak{g} + \mu\wedge} }
      &&&
      \Big\{
        \Big(
          e^b \wedge \underset{j \in \mathbb{Z}}{\sum} \omega_{2 j t + n}
        \Big)_{ 2k t + n }
      \Big\}_{k \in \mathbb{Z} }
      \ar[d]^{  d_{\widehat{\mathfrak{g}}} }
     \\
     &&&
      \Big\{
               \Big(
          {
          \underset{
            \mu \wedge e^b
          }{
          \underbrace{
            d_{\widehat{\mathfrak{g}}}(e^b)
          }}
          \wedge \underset{j \in \mathbb{Z}}{\sum} \omega_{2 j t + n}
          +
          e^b \wedge \underset{j \in \mathbb{Z}}{\sum} d_{\mathfrak{g}}\omega_{2 j t + n}
          }
        \Big)_{ 2k t + n }
      \Big\}_{k \in \mathbb{Z} }
      \\
      \left\{
        {
          d_{\mathfrak{g}}\omega_{2kt + n}
          + \mu \wedge \omega_{2kt + n-1}
        }
      \right\}_{ k \in \mathbb{Z} }
      \ar@{|->}[rrr]_-{\phi}
      &&&
      \Big\{
        \Big(
          e^b \wedge \underset{j \in \mathbb{Z}}{\sum}
          \big(
            {
              d_{\mathfrak{g}}\omega_{2 j t + n}
              + \mu \wedge \omega_{2j t + n-1}
            }
          \big)
        \Big)_{ 2k t + n }
      \Big\}_{k \in \mathbb{Z} }
      \ar@{=}[u]
    }
  $$
  which intertwines the twisted CE-differential $d_{\mathfrak{g}} + \mu\wedge $ of $\mathfrak{g}$ with the plain CE-differential $d_{\widehat{\mathfrak{g}}}$ of $\widehat{\mathfrak{g}}$, as shown.
  It is clear that this is an injective linear map, hence an injective chain map from the $(d_{\mathfrak{g}} = +\mu\wedge)$-complex to the
  $d_{\widehat{\mathfrak{g}}}$-complex, whose image is closed unde the pre-image of $d_{\mathfrak{g}}$. This implies the claim.
\end{proof}

The following example show that the map of Prop. \ref{TwistedCohomologyInjects} is in general not surjective,
hence that there are in general cohomology classes not in its image.

\begin{example}[Non-twisted periodic cohomology of higher central extension is strictly larger than twisted cohomology]
  Under the assumptions of Prop. \ref{TwistedCohomologyInjects},
  consider two elements $\omega_{2t+ n}, \omega_{4t + n} \in \mathrm{CE}(\mathfrak{g})$
  and consider the following three equations:
 $$
 d_{\mathfrak{g}} \omega_{4t}  = - \mu \wedge \omega_{2t}
 \;,\qquad
d_{\mathfrak{g}} \omega_{2 t} = 0
 \;,\qquad
    \mu \wedge \omega_{4 t} = 0\;.
 $$
  The combination of the first two of these conditions is equivalent to the statement that the element
  $$
    \omega_{4t + n} + b \wedge \omega_{2 t + n} \in \mathrm{CE}(\widehat{\mathfrak{g}})
  $$
  is a $d_{\widehat{\mathfrak{g}}}$-cocycle. On the other hand, the combination of all three conditions
 is equivalent to the stronger statement that also the element
  $$
    b \wedge \omega_{4t + n} + \tfrac{1}{2} b \wedge b \wedge \omega_{2 t + n} \in \mathrm{CE}(\widehat{\mathfrak{g}})
  $$
  is a $d_{\widehat{\mathfrak{g}}}$-cocycle, which in turn is equivalent to the statement that the tuple
  $$
    \left( \omega_{2t + n}, \omega_{4t + n } \right)
  $$
  is a $(d_{\mathfrak{g}} + \mu\wedge)$-cocycle.
\end{example}

We now give the  homotopy-theoretic explanation of the phenomenon seen in Prop. \ref{TwistedCohomologyInjects},
exhibiting it as a special case of a general statement \cite[Theorem 4.39]{NSS12} about twisted cohomology.
\begin{defn}
[Coefficients for twisted $L_\infty$-cocycles]
For
$n\geq 0$ and
$t \geq 1$,
write $\mathfrak{l}(\Sigma^n K(t)/b^{t}\mathfrak{u}_1 )
\in sL_\infty\mathrm{Alg}^{\mathrm{fin}}_{\mathbb{R}}$ for
$$
  \mathrm{CE}\big(
    \mathfrak{l}(\Sigma^n K(t)/b^{2 t-1}\mathfrak{u}_1 )
  \big)
  :=
  \mathbb{R}
  \big[ \hspace{-2mm}
    \underset{
      \mathrm{deg} = (2t+1, \mathrm{even})
    }{
      \underbrace{h_{2t+1}}
    }
    \hspace{-3mm} , \;\;\;
    \big( \hspace{-1mm}
      \underset{ (2 k t + n, \mathrm{even}) }{\underbrace{ \omega_{2 k t + n} }}
      \hspace{-3mm}, k \in \mathbb{Z}
    \big)
  \big]/
  \Big(
    \begin{array}{l}
      d h_{2 t + 1} = 0,
      \\
      d \omega_{2 (k+1) t + n} = - h_{2 t + 1} \wedge \omega_{2 k t + n}
    \end{array}
  \Big)
$$
and write
\begin{equation}
  \label{BundleMorphismFromTwistedKt}
  \xymatrix@R=.0em{
    \mathfrak{l} \Sigma^n K(t)
    \ar[rr]^{\rho}
    &&
    b^{2t } \mathfrak{u}_1
    \\
    h_{2t+1} && \mathrm{dd}_{2t +1} \ar@{|->}[ll]
  }
\end{equation}
for the canonical morphism.
Then for $\mathfrak{g} \in SL_\infty \mathrm{Alg}^{\mathrm{fin}}_{\mathbb{R}}$ any super $L_\infty$-algebra,
and for
$$
  \mu \;:\; \mathfrak{g} \longrightarrow b^2t \mathfrak{u}_1
$$
a super $L_\infty$-cocycle on $\mathfrak{g}$ according to Def. \ref{SuperLInfinityCocycles}, a
\emph{$\mu$-twisted cocycle} $\omega_\bullet$ in degree $n$ mod $2t$ is a homomorphism over $b^{2t}\mathfrak{u}_1$ (see \cite[Def. 4.21]{NSS12}):
\begin{equation}
\hspace{-3mm}
  \label{TwistedLInfinityCocycles}
  \left\{
    \raisebox{20pt}{
    \xymatrix@=1em{
      \mathfrak{g}
      \ar[dr]_-{\mu}
      \ar[rr]^-{\omega_\bullet}
        &&
      \mathfrak{l}(\Sigma^p K(n)/b^{n-1}\mathfrak{u}_1)
      \ar[dl]^\rho
      \\
      & b^{2t} \mathfrak{u}_1
    }}
  \right\}
  \;\simeq\;
  \left\{
    \omega_{2k t + n}
    \;\in\;
    \mathrm{CE}(\mathfrak{g})
    \;\vert\;
    d_{\mathfrak{g}} \omega_{2 (k+1) t + n} + h_{2t+1} \wedge \omega_{ 2 k t + n } = 0,
    \;\;
    k \in \mathbb{Z}
  \right\}\;.
\end{equation}
\end{defn}

The homotopy fiber of the canonical morphism (\ref{BundleMorphismFromTwistedKt}) is the coefficient $\mathfrak{l}\Sigma^n K(t)$ for untwisted periodic
cohomology from Def. \ref{PeriodicLInfinityCohomology}
$$
  \xymatrix{
    \mathfrak{l}\Sigma^n K(t)
    \ar[rr]^-{\mathrm{hofib}(\rho)}
    &&
    \mathfrak{l}\Sigma^n K(t)/b^{2t-1}
    \ar[d]^\rho
    \\
    && b^{2t}\mathfrak{u}_1\;.
  }
$$

\begin{defn}
[CE-algebra of degree $n$, rational, $2t$-periodic, $(2t+1)$-twisted (generalized) cohomology]
\label{defn-lSigmaKt}
For $n \geq 0$ and
$t \geq 1$,  the super $L_\infty$-algebra
$$
  \mathfrak{l}(\Sigma^n K(t))_{\mathrm{res}}
    \;\in\;
  sL_{\infty} \mathrm{Alg}^{\mathrm{fib}}_{\mathbb{R}}
$$
is defined dually by
$$
  \mathrm{CE}
  \left(
    \mathfrak{l}(\Sigma^n K(t))_{\mathrm{res}}
  \right)
  \;:=\;
  \mathbb{R}
  \Big[ \hspace{-2mm}
    \underset{
      \mathrm{deg} = (2t, \mathrm{even})
    }{
    \underbrace{
      b
    }
    }
    ,
    \underset{
       (2t+1, \mathrm{even})
    }{
      \underbrace{h_{2t+1}}
    },
    \big( \hspace{-2mm}
      \underset{ (2 k t + n, \mathrm{even}) }{\underbrace{ \omega_{2 k t + n} }}
      \hspace{-2.5mm}, k \in \mathbb{Z}
    \big)
  \Big]/
  \left(
    \begin{array}{l}
      d b_{2t} = h_{2t+1},\;\;\;
      d h_{2 t + 1} =0,
      \\
      d \omega_{2 (k+1) t + n} =  - h_{2t + 1} \wedge \omega_{ 2k t + n }
    \end{array}
  \right)\;.
$$
\end{defn}

For the proof of Prop. \ref{FromTwistedCocyclesToUntwistedCocyclesOnTheExtensionClassifiedByTheTwist}
below we need the following fibration resolution of this homotopy fiber:

\begin{lemma}[Fibration resolution of coefficients for untwisted periodic cohomology]
\label{FibrationResolutionForUnTwistedPeriodicCohomology}
The algebra
\newline
$\mathrm{CE}
  \left(
    \mathfrak{l}(\Sigma^n K(t))_{\mathrm{res}}
  \right)$  of Def. \ref{defn-lSigmaKt}
 provides a fibration resolution of the homotopy fiber inclusion
\begin{equation}
  \label{FibrationResolutionForHomotpyFiberOfTwistedCoefficients}
  \xymatrix@R=2pt{
    \mathfrak{l}( \Sigma^n K(t) )
    \ar@/^2pc/[rr]^{\mathrm{hofib}(\rho)}
   \; \ar@{^{(}->}[r]^{\simeq}
    &
    \mathfrak{l}( \Sigma^n K(t))_{\mathrm{res}}
    \ar@{->>}[r]^-{ \mathrm{hofib}(\rho)_{\mathrm{res}} }
    &
    \mathfrak{l}(\Sigma^n K(t)/b^{2t-1}\mathfrak{u}_1)
    \\
    \omega_{2 k t + n} & \omega_{2 kt + n} \ar@{|->}[l] & \omega_{2 k t + n} \ar@{|->}[l]
    \\
    0 & h \ar@{|->}[l] & h \ar@{|->}[l]
    \\
    0 & b \ar@{|->}[l]
    \\
    \\
    \mathfrak{l}( \Sigma^n K(n) )
    \ar@{<-}[r]_{\simeq}^\phi
    &
    \mathfrak{l}( \Sigma^n K(t))_{\mathrm{res}}
    \\
    \omega_{2 k t + n} \; \ar@{|->}[r] & \;
    \Big( e^b \wedge \underset{j \in \mathbb{Z}}{\sum} \omega_{2 j t + n}\Big)_{2 k t + n}\;.
  }
\end{equation}
\end{lemma}
\begin{proof}
  That $\mathfrak{l}(\Sigma^n K(t))_{\mathrm{res}}$ provides a fibration resolution
  as claimed follows just as the proof of Prop. \ref{hofibofSuperLInfinityAlgebras} from \cite[Theorem 3.1.13]{FRS13}.
  From this the other statements follow by inspection.
\end{proof}

\begin{remark}
At the bottom of \eqref{FibrationResolutionForHomotpyFiberOfTwistedCoefficients}
we indicated a homotopy inverse $\phi$ of the resolution, which will be of use below.
In terms of this resolution, the long homotopy fiber sequence of $\rho$ starts out as
the following ordinary fibration sequence:
$$
  \xymatrix@R=2pt{
    b^{2t-1}\mathfrak{u}_1
    \ar[rr]^b
    &&
    \mathfrak{l}(\Sigma^n K(t))_{\mathrm{res}}
    \ar[rr]^-{ \mathrm{hofib}(\rho)_{\mathrm{res}} }
    &&
    \mathfrak{l}(\Sigma^n K(t)/b^{2t-1}\mathfrak{u}_1)
    \ar[rr]^-\rho
    &&
    b^{2t}\mathfrak{u}_1
    \\
    0
    &&
    h_{2t + 1} \ar@{|->}[ll]
    &&
    h_{2t + 1} \ar@{|->}[ll]
    &&
    \mathrm{dd}_{2t+1} \ar@{|->}[ll]
    \\
    0
    &&
    \omega_{2 kt + n} \ar@{|->}[ll]
    &&
    \omega_{2 kt + n} \ar@{|->}[ll]
    &&
    \\
    b
    &&
    b \ar@{|->}[ll]
    &&
  }
$$
\end{remark}

\begin{prop}[Map of twisted cohomology of $\mathfrak{g}$ into non-twisted periodic cohomology of $\widehat{\mathfrak{g}}$ is a homotopy pullback]
\label{FromTwistedCocyclesToUntwistedCocyclesOnTheExtensionClassifiedByTheTwist}
The inclusion from Prop. \ref{TwistedCohomologyInjects} of the twisted super $L_\infty$-cohomology on $\mathfrak{g}$
into the non-twisted periodic cohomology of $\widehat{\mathfrak{g}}$ is equivalently the image on cohomology classes of forming the
homotopy pullback along the homotopy fiber inclusion of the projection morphism $\rho$ (see the mapping \eqref{BundleMorphismFromTwistedKt}):
$$
  [\mathrm{hofib}(\rho)^\ast]
  \;:\;
  H^{\bullet + \mu}( \mathfrak{g}/K )
  \xymatrix{\ar@{->}[r]&}
  H^{\bullet \;\mathrm{mod}\; 2t }( \widehat{\mathfrak{g}}/K )
  \,.
$$
\end{prop}
\begin{proof}
Consider a twisted cocycle in degree $n$ mod $t$
$$
  \omega_\bullet
  :=
  \left\{
    \omega_{2 k t + n}
  \right\}_{ k \in \mathbb{Z} }
  \,.
$$
Via the homomorphism \eqref{TwistedLInfinityCocycles} we may regard this equivalently as a
super $L_\infty$-homomorphism of the form
$$
  \xymatrix@R=1.5em{
    \mathfrak{g}
    \ar[rr]^-{ \omega_\bullet }
    \ar[dr]_-{\mu}
    &&
    \mathfrak{l}\left( \Sigma^n K(t)/b^{2t-1}\mathfrak{u}_1 \right)\;.
    \ar[dl]^-{\rho}
    \\
    &
    b^{2t}\mathfrak{u}_1
  }
$$
By forming homotopy pullbacks and using the pasting law for homotopy pullbacks this induces
a homotopy-commutative diagram of the form
$$
  \xymatrix@C=8em{
    b^{2t-1}\mathfrak{u}_1
    \ar@/^1.7pc/[rr]^b
    \ar[r]_>{\ }="s3"
    \ar[d]
    &
    \widehat{\mathfrak{g}}
    \ar[r]^{\hspace{-1.2cm} (\mathrm{hofib}(\rho))^\ast(\omega_{\bullet}) }_>{\ }="s2"
    \ar[d]
    &
    \mathfrak{l}(\Sigma^n K(t))
    \ar[r]_>{\ }="s1"
    \ar[d]^{\mathrm{hofib}(\rho)}
    &
    \ast \ar[d]
    \\
    \ast
    \ar[r]^<{\ }="t3"
    &
    \mathfrak{g}
    \ar@/_1.5pc/[rr]_{\mu}
    \ar[r]^-{\;\;\;\omega_\bullet}^<{\ }="t2"
    &
    \mathfrak{l}(\Sigma^n K(t)/b^{2t-1}\mathfrak{u}_1)
    \ar[r]^-{\;\;\rho}^<{\ }="t1"
    &
    b^{2t} \mathfrak{u}_1
    \ar@{=>} "s1"; "t1"
    \ar@{=>} "s2"; "t2"
    \ar@{=>} "s3"; "t3"
  }
$$
Hence this induces, in particular, a super $L_\infty$-homomorphism
$$
  (\mathrm{hofib}(\rho))^\ast(\omega_\bullet)
  \;:\;
  \widehat{\mathfrak{g}}
  \longrightarrow
  \mathfrak{l}( \Sigma^n K(t) )
$$
which, via Def. \ref{PeriodicLInfinityCohomology}, represents a periodic cohomology class on the higher
 central extension $\widehat{\mathfrak{g}}$:
$$
  [(\mathrm{hofib}(\rho))^\ast(\omega_\bullet)]
  \;\in\;
  H^{\bullet\;\mathrm{mod}\;t}(\widehat{\mathfrak{g}})
  \,.
$$
So it is now sufficient to show that this cohomology class is represented by $\phi(\omega_\bullet)$ according to
formula \eqref{CochainMapForInclusionOfTwistedCohomologyIntoNonTwistedPeriodicCohomologyOfHigherCentralExtension}.
In order to obtain such an explicit formula for $(\mathrm{hofib}(\rho))^\ast(\omega_\bullet)$ we may apply the
resolution of Lemma \ref{FibrationResolutionForUnTwistedPeriodicCohomology} to represent, up to weak equivalence,
the middle homotopy pullback in the above diagram by an ordinary pullback, as shown in the middle of the following diagram:
$$
  \xymatrix@C=6em{
    &&&
    \mathfrak{l}(\Sigma^{n} K(t))
    \ar@{^{(}->}[d]_\simeq
    \\
    b^{n-1}\mathfrak{u}_1
    \ar[r]
    \ar[d]
    &
    \widehat{\mathfrak{g}}
    \ar@/^1.5pc/[urr]^{\vspace{-1cm} \hspace{-1cm} \big( e^b \wedge
    \underset{ j \in \mathbb{Z} }{\sum} \omega_{2 j t + n}  \big)_{2t \bullet + n} }
    \ar[rr]^{ (\mathrm{hofib}(\rho)_{\mathrm{res}})^\ast(\omega_\bullet) }
    \ar[d]
    &&
    \mathfrak{l}(\Sigma^n K(t))_{\mathrm{res}}
    \ar[r]
    \ar@/^1pc/[u]^{\phi}
    \ar@{->>}[d]_{ \mathrm{hofib}(\rho)_{\mathrm{res}} }
    &
    \ast \ar[d]
    \\
    \ast
    \ar[r]
    &
    \mathfrak{g}
    \ar@/_1.5pc/[rrr]_{\mu}
    \ar[rr]^-{\omega_\bullet}
    &&
    \mathfrak{l}(\Sigma^n K(t)/b^{t-1}\mathfrak{u}_1)
    \ar[r]^-\rho
    &
    b^{2t} \mathfrak{u}_1\;.
  }
$$
Inspection reveals that the resolved homotopy pullback of $\omega_\bullet$ that is obtained thereby is simply given by
$$
  \xymatrix@R=1pt{
    \widehat{\mathfrak{g}}
    \ar[rrr]^-{ (\mathrm{hofib}(\rho)_{\mathrm{res}})^\ast(\omega_\bullet) }
    &&&
    \mathfrak{l}( \Sigma^n K(t) )_{\mathrm{res}}
    \\
    \omega_{2kt + n}
    &&&
    \omega_{2kt + n}
    \ar@{|->}[lll]
    \\
    \mu
    &&&
    h
    \ar@{|->}[lll]
    \\
    b
    &&&
    b
    \ar@{|->}[lll]
  }
$$
In order to manifestly identify this with a periodic cocycle we may postcompose with the weak equivalence $\phi$
from Lemma \ref{FibrationResolutionForUnTwistedPeriodicCohomology}, as shown in the above diagram. By the formula (\ref{FibrationResolutionForHomotpyFiberOfTwistedCoefficients}) for $\phi$ this implies the claim:
$$
  \begin{aligned}
    {[ \mathrm{hofib}(\rho)^\ast(\omega_\bullet) ]}
    & =
    [ \phi \circ (\mathrm{hofib}(\rho)_{\mathrm{res}})^\ast(\omega_\bullet) ]
    \\
    & =
    \Big[
      e^b
        \wedge
      \Big(
        \underset{j \in \mathbb{Z}}{\sum} \omega_{2jt + n}
      \Big)_{2 t \bullet + n}
    \Big]\;.
  \end{aligned}
$$

\vspace{-8mm}
\end{proof}

\subsection{Transgression elements}

The following simple but crucial structure in super $L_\infty$-homotopy theory will play a key role in Section \ref{SphericalTDualityOnExceptionalMTheorySpacetime}
and Section \ref{ExceptionalGenralisedGeometry} below:

\begin{defn}[Transgression of super $L_\infty$-cocycles, {\cite[Def. 21]{SSS1}} {\cite[Def. 4.1.20]{FSS12}}]
\label{TransgressionElements}
Consider the homotopy fiber sequence
$$
  \xymatrix{
    \mathfrak{f}\;
    \ar@{^{(}->}[r]^{\iota}
    & \mathfrak{g}
    \ar@{->>}[d]^{\pi}\
    \\
    & \mathfrak{b}
  }
$$
for $\pi$ a fibration
and let $\mu_{\mathfrak{b}} \in \mathrm{CE}(\mathfrak{b})$
be a cocycle. Then a \emph{transgression} of $\mu_{\mathfrak{b}}$ to a cocycle in the fiber
$\mu_{\mathfrak{f}} \in \mathrm{CE}(\mathfrak{f})$
is an element
$
  \mathrm{cs} \in \mathrm{CE}(\mathfrak{g})
$
such that
\begin{equation}
  \label{TransgressionElementCondition}
  \xymatrix{
    \mu_{\mathfrak{f}}
    &&&
    \mathrm{CE}(\mathfrak{f})
    \\
    \mathrm{cs}
      \ar@{|->}[r]^{d}
      \ar@{|->}[u]_{\iota^\ast}
    & d \mathrm{cs}
    &&
    \mathrm{CE}(\mathfrak{g})
    \ar[u]^{\iota^\ast}
    \\
    &
    \mu_{\mathfrak{b}}
    \ar@{|->}[u]^{\pi^\ast}
    &&
    \mathrm{CE}(\mathfrak{b})\;.
    \ar[u]^{\pi^\ast}
    \ar@/_2pc/[uu]_0
  }
\end{equation}
\end{defn}

We will be interested in  elements which arise in the following manner.

\begin{example}[Transgression elements in higher central extensions]
  \label{HigherCentralExtensionTransgressionElement}
  For $n \in \mathbb{N}$, $n \geq 1$, let
  $$
    \xymatrix{
      b^{n-1}\mathfrak{u}_1
      \; \ar@{^{(}->}[r]
      &
      \widehat{\mathfrak{b}}
      \ar[d]_{\mathrm{hofib}(\mu)}
      \\
      & \mathfrak{b} \ar[rr]^-{\mu} && b^n \mathfrak{u}_1
    }
  $$
  be a homotopy fiber sequence. Then the element $b \in \mathrm{CE}(\widehat{\mathfrak{g}})$ from
  Def. \ref{hofibofSuperLInfinityAlgebras} is a transgression element in the sense of Def. \ref{TransgressionElements}.
$$
  \xymatrix{
    b
    &&&
    \mathrm{CE}(b^{n-1} \mathfrak{u}_1)
    \\
    b
      \ar@{|->}[r]^{d}
      \ar@{|->}[u]_{\iota^\ast}
    & d b
    &&
    \mathrm{CE}(\widehat{\mathfrak{b}})
    \ar[u]_{\iota^\ast}
    \\
    &
    \mu
    \ar@{|->}[u]_{\pi^\ast}
    &&
    \mathrm{CE}(\mathfrak{b})
    \ar[u]_{\pi^\ast}
  }
$$
In fact this is the \emph{universal} transgression element for $\mu$: For every other transgression element $\mathrm{cs}$ of $\mu$
on any other fibration $\xymatrix{ \mathfrak{g} \ar@{->>}[r] & \mathfrak{b}}$ as in diagram \eqref{TransgressionElementCondition},
there is a unique morphisms of fibrations
\begin{equation}
  \label{ClassifyingMolrphismsForTransgressionElements}
  \xymatrix@R=1.5em{
    \mathfrak{g}
    \ar@{->>}[rr]^-{ \phi }
    \ar@{->>}[dr]_{\pi}
    &&
    \widehat{\mathfrak{b}}
    \ar@{->>}[dl]^{ \mathrm{hofib}(\mu) }
    \\
    &
    \mathfrak{b}
  }
\end{equation}
such that
$
  \mathrm{cs} = \phi^\ast(b)
  $.
\end{example}

Of particular interest are the following transgression elements for higher cocycles
existing on ordinary (i.e., not higher) super Lie algebras.
\begin{prop}
 [Left-invariant differential forms on a super Lie group represent CE-elements]
  \label{LeftInvariantDifferentialFormsOnLieGroupRepresentCEElements}
   Let $G$ be a super Lie group with super Lie algebra $\mathfrak{g}$. Then
  restriction to the super tangent space at the neutral element constitutes an isomorphism of
  differential $(\mathbb{Z} \times \mathbb{Z}/2)$-graded-commutative algebras
  $$
    \xymatrix{
      \Omega^\bullet(G)_{\mathrm{LI}}
      \ar[r]^-\simeq
      &
      \mathrm{CE}(\mathfrak{g})
    }
  $$
  between the sub-DGCA-algebra of left-invariant differential forms inside the super de Rham algebra
  of the super Lie group and the Chevalley-Eilenberg algebra of the super Lie algebra.
\end{prop}
We will leave this isomorphism notationally implicit.
The application of this isomorphism in the following simple setup is what drives the
exceptional generalized geometry that we discover below in Section \ref{SphericalTDualityOnExceptionalMTheorySpacetime}.

\begin{example}[Primitives from sections via transgression elements]
  \label{PrimitivesFromSections}
 Consider a transgression of super Lie algebra cocycles (Def. \ref{TransgressionElements})
  $$
  \raisebox{20pt}{
  \xymatrix@R=1.5em{
    \mathfrak{g}
    \ar[rr]^-{ \phi }
    \ar@{->>}[dr]_{\pi}
    &&
    \widehat{\mathfrak{b}}
    \ar@{->>}[dl]^{ \mathrm{hofib}(\mu) }
    \\
    &
    \mathfrak{b}
  }
  }
  $$
  for $\mathfrak{g}$ and $\mathfrak{b}$ both super Lie 1-algebras.
 We are interested in studying the corresponding fibration at the level of
 the Lie groups. Consider a Lie integration to
  left-invariant super-differential forms on corresponding super Lie groups according to Prop. \ref{LeftInvariantDifferentialFormsOnLieGroupRepresentCEElements}:
  $$
  \xymatrix@R=1.5em{
    & F \ar@{^{(}->}[d]_-{\iota}
    &
    d \iota^\ast b = 0
    \\
    & G
    \ar[d]_{\pi}
    & b \in \Omega^\bullet(G)_{\mathrm{LI}}, \phantom{AA} d b = \pi^\ast(\mu)
    \\
    & B & \mu \in \Omega^\bullet(B)_{\mathrm{LI}}.
  }
$$
Then every smooth section $\sigma$ of $\pi$, regarded as a fibration of super-manifolds
$
  \xymatrix{
    G
    \ar[r]_\pi
 &
    B
    \ar@/_.5pc/[l]_\sigma
  }
$
induces a primitive
$
  \sigma^\ast(b) \in \Omega^\bullet(B)
$
of $\mu$:
$$
    d \sigma^\ast ( b )
     =
    \sigma^\ast ( d b )
  =
    \underset{\mathrm{id}}{\underbrace{\sigma^\ast \pi^\ast}} \mu
  = \mu
    \,.
$$
Of course $\sigma^\ast(b)$ here is necessarily \emph{not} left-invariant if $[\mu] \neq 0 \in H^\bullet(\mathfrak{g})$.
\end{example}


\section{Higher T-duality}
\label{sphericaltduality}

In this section we establish the higher and super-geometric generalization (Theorem \ref{HigherTDuality} and Cor.
 \ref{HigherTopologicalTDualityInTwistedCohomology} below) of toroidal ``topological T-duality''
as originally proposed in \cite{BHM03} and as derived from analysis of the super-WZW terms of type II F1/D$p$-branes in \cite{FSS16}.

\medskip
After considering some basics of rational higher torus fibrations in section \ref{HigherTopusFibrations},
we introduce the concept of \emph{higher T-duality correspondences} between twisting
cocycles on such higher torus fibrations in Section \ref{Sec-HigherTDualityCorrespondences}
and we prove, in Section \ref{HigherTDualityTransformations}, that these induce isomorphisms
in the corresponding twisted cohomology groups. These are the higher T-duality isomorphisms as such,
generalizing Hori's formula for the T-duality isomorphisms on Ramond-Ramond (RR) charges twisted
by F1-brane charges (recalled in Section  \ref{OrdinaryTypeIITDuality} below) to higher twisting degree,
notably by M5-brane charges (Section \ref{SphericalTDualityOfM5BranesOnM2ExtendedSpacetime} below).

\medskip
Finally in section \ref{HigherTDualityForDecomposedFormFields} we consider a phenomenon in higher
T-duality that has no analog in ordinary T-duality, namely its passage to ``decomposed'' higher form fields
(Theorem \ref{HigherTDualityForDecomposedFields} below); the realization of this effect in M-brane
theory is discussed in Section \ref{SphericalTDualityOnExceptionalMTheorySpacetime}.

\subsection{Higher torus fibrations}
\label{HigherTopusFibrations}

Where a torus is, topologically, a Cartesian product of circles $S^1 \simeq U(1)$, by a ``higher torus''
we shall mean here a Cartesian product of shifted circles $B^n U(1)$, i.e. an Eilenberg-MacLane space
$K(\Z, n+1)$. Rationally, we think of this via the corresponding $L_\infty$-algebras
$b^{n}\mathfrak{u}_1$. In every odd degree these are indistinguishable, up to rational weak equivalence, from the corresponding
odd-dimensional spheres (see the discussion at the beginning of Section \ref{SuperLInfinityHomotopyTheory}).
 As such the rational \emph{spherical} T-duality for M-branes discussed below
in Sections \ref{SphericalTDualityOfM5BranesOnM2ExtendedSpacetime}, \ref{SphericalTDualityOver7dSpacetime} and \ref{SphericalTDualityOnExceptionalMTheorySpacetime}
involves higher torus fibrations in the following sense.

\medskip
Beyond its plain homotopy type, the Riemannian structure on a flat torus is equivalently encoded in the
corresponding universal cup product on the tuples of universal first Chern classes that classify torus-principal bundles.
This cohomological incarnation of the Riemannian structure on a torus immediately generalizes to higher tori, in this sense,
in terms of cup products on tuples of universal higher Dixmier-Douady classes.  This is what the following definitions formalizes.

\begin{defn}[Higher tori and cup products]
  \label{HigherCupProducts}
  For $n, k \in \mathbb{N}$ two natural numbers, consider the
  $L_\infty$-algebra
  $$
    b^n (\mathfrak{u}_1)^k
    \;\in\;
    sL_\infty
  $$
  whose Chevalley-Eilenberg algebra has $k$ generators in degree $n+1$ and vanishing differential:
  $$
    \mathrm{CE}( b^n (\mathfrak{u}_1)^k )
    \;=\;
    \Big(
      \mathbb{R}\big[ \underset{ \mathrm{deg} = n+1 }{\underbrace{\mathrm{dd}_{n}^{(1)}}},  \cdots, \underset{ \mathrm{deg} = n+1 }{\underbrace{\mathrm{dd}_n^{(k)} }} \big],
      \, d = 0
    \Big)
    \,.
  $$
  For
  $$
   \left\langle -,-\right\rangle
   \;:\;
   \mathbb{R}^k \otimes \mathbb{R}^k
     \longrightarrow
   \mathbb{R}
  $$
  a non-degenerate bilinear pairing, symmetric if $n+1$ is even, skew-symmetric otherwise, on the vector space spanned by these generators,
  we say that the corresponding \emph{universal cup product} is the morphism
  \begin{equation}
    \label{UniversalCupProduct}
    \xymatrix@R=3pt{
      b^n (\mathfrak{u}_1)^k
      \times
      b^n (\mathfrak{u}_1)^k
      \ar[rr]^-{ \left\langle (-) \cup (-) \right\rangle }
      &&
      b^{2n+1} \mathfrak{u}_1
      \\
      \left\langle
        \mathrm{pr}_1^\ast \big( \widevec{\mathrm{dd}}_n \big)
          \wedge
        \mathrm{pr}_2^\ast \big( \widevec{\mathrm{dd}}_n \big)
      \right\rangle
      \ar@{<-|}[rr]
       &&
      \mathrm{dd}_{2n+1}
   }
  \end{equation}
\end{defn}

\begin{defn}[Higher torus fibrations]
\label{HigherTorisFibrationsAndFiberIntegration}
 Consider the  higher cup product pairing \eqref{UniversalCupProduct}
 (Def. \ref{HigherCupProducts}).
 Then for $\mathfrak{g} \in sL_\infty\mathrm{Alg}$ a super $L_\infty$-algebra
 and for
 $$
   \vec \mu
   \;:\;
   \mathfrak{g}
     \longrightarrow
   b^n(\mathfrak{u}_1)^k
 $$
 a $k$-tuple of $(n+1)$-cocycles, the corresponding \emph{higher torus fibration}
  is the homotopy fiber of $\mu$
 \begin{equation}
   \label{HigherTorusFibration}
   \xymatrix@R=1.2em{
     \hat{\mathfrak{g}}
     \ar[d]^{\pi := \mathrm{hofib}(\vec{\mu})}
     \\
     \mathfrak{g}
   }
 \end{equation}
 with Chevalley-Eilenberg algebra given via Prop. \ref{hofibofSuperLInfinityAlgebras} as
 $
   \mathrm{CE}(\hat {\mathfrak{g}})
   \;=\;
   \mathrm{CE}(\mathfrak{g})[ \vec{b} ]( d \vec b = \vec \mu )
   $.
\end{defn}

 This factors of course as a sequence of plain higher central extensions.
 In particular, we have the following.

\begin{defn}[Fiber integration]
 \label{FiberIntegration}
 Let
 $$
   \xymatrix@R=1.2em{
     \hat{\mathfrak{g}}
     \ar[d]^{\pi := \mathrm{hofib}(\mu)}
     \\
     \mathfrak{g}
   }
 $$
 be a higher central extension according to Prop. \ref{hofibofSuperLInfinityAlgebras},
 classified by a cocycle of even degree $\mu\colon \mathfrak{g}\to b^{2n-1}\mathbb{R}$. Let
 \[
 \int\colon \mathrm{CE}(b^{2n-2}\mathbb{R})=\left(\mathbb{R}[b],db=0\right)\longrightarrow \mathbb{R}[\deg b]
 \]
 be the morphism of cochain complexes defined by $\int b=1$, where $b$ has degree $2n-1$. Then the morphism
 ``$\int$" extends to a morphism of graded vector spaces, called the \emph{fiber integration morphism},
 $$
   \pi_\ast
   \;:\;
   \mathrm{CE}(\hat{\mathfrak{g}})
   \longrightarrow
   \mathrm{CE}(\mathfrak{g})[\deg b]
 $$
 as the composition
 \[
 \mathrm{CE}(\hat{\mathfrak{g}})\cong \mathrm{CE}(\mathfrak{g})\otimes \mathrm{CE}(b^{2n-2}\mathbb{R})\xrightarrow{\;\;\;\mathrm{id}\otimes \int\;\;\;} \mathrm{CE}(\mathfrak{g})[\deg b]\;.
 \]
 \end{defn}

\medskip
 \begin{remark}[Wrapping and non-wrapping modes]
 \label{wrap-unwrap}
 As $\mathrm{CE}(\hat{\mathfrak{g}}) = \mathrm{CE}(\mathfrak{g})[b]/(d b= \mu)$,
 any cochain $\omega$ on $\hat{\mathfrak{g}}$ can be uniquely written as
 \begin{equation}
   \label{CochainCoefficientsWrapping}
   \omega = \omega_{\mathrm{nw}} + b \wedge \omega_{\mathrm{w}}
 \end{equation}
for suitable ``non-wrapping'' and ``wrapping'' coefficients
 $\omega_{\mathrm{nw}}, \omega_{\mathrm{w}} \in \mathrm{CE}(\mathfrak{g})$. Under fiber integration
  \begin{equation}
   \label{FiberIntegrationMap}
   \pi_\ast
     \;:\;
   \omega \longmapsto (-1)^{\mathrm{deg}(b)} \omega_{\mathrm{w}}
   \,,
 \end{equation}
 so that $\pi_*$ picks the coefficient of the ``wrapping component'', up to a sign.
\end{remark}

\medskip
\begin{remark}
The fiber integration morphism does indeed respect the differentials and so it is a $\deg(b)$-graded morphism of chain complexes $\mathrm{CE}(\hat{\mathfrak{g}})
   \to
   \mathrm{CE}(\mathfrak{g})$. Namely, one has
\begin{equation}
  \label{FiberIntegrationRespectsDifferentials}
  \begin{aligned}
    \pi_\ast( d \omega )
    & =
    \pi_\ast\big(
       d \omega_{\mathrm{nw}}
       +
       \underset{\mu}{\underbrace{d b}} \wedge \omega_{\mathrm{nw}}
       +
       (-1)^{\mathrm{deg}(b)} b \wedge d \omega_{\mathrm{nw}}
    \big)
    \\
    & =
    (-1)^{\mathrm{deg}(b)} d \omega_{\mathrm{nw}}
    \\
    & =
    (-1)^{\mathrm{deg}(b)} d \pi_\ast( \omega )\;.
  \end{aligned}
\end{equation}
\end{remark}

\medskip
We may think of an ordinary $k$-torus $T^{k}$ as being an $S^1$-fibration over the $(k-1)$-torus in $k$ different ways
$$
  \xymatrix{
    S^1 \; \ar@{^{(}->}[r] & T^k
      \ar[d]^{\pi_j}
      \\
    & T^{k-1}
  }
$$
and of course the same remains true for higher tori, according to Def. \ref{HigherCupProducts}.
Accordingly from every higher $k$-torus there are $k$ \emph{partial fiber integration} maps that
fiber integrate, in the sense of Def. \ref{FiberIntegration}, over one of the factors, Def. \ref{FiberIntegration}
below. This is needed in the key condition (\ref{HigherTDualityPairingCondition}) on higher T-duality
correspondences below in Def. \ref{HigherTDualityCorrespondences}.
\begin{defn}
[Partial fiber integration]
\label{FiberIntegration}
For any choice
 $j \in \{1, \cdots, k\}$, $\pi$  factors as
 $$
   \xymatrix{
     \hat{\mathfrak{g}}
     \ar[d]^{ \pi^{(j)} := \mathrm{hofib}( \mu^{(j)} ) }
     \\
     \ar[d]^{ \pi^{(1,\cdots, j-1,j+1,\cdots,k)} :=  \mathrm{hofib}( \mu^{(1)}, \cdots, \mu^{(j-1)}, \mu^{(j+1)}, \cdots, \mu^{(k)} ) }
     \\
     \mathfrak{g}
   }
 $$
 so that there are $k$ different intermediate fiber integration maps (Def. \ref{FiberIntegration}), defined by
 (see Remark \ref{wrap-unwrap})
 \begin{equation}
   \label{IntegrationOverSingleFiberFactor}
   \pi^{(j)}_\ast
   \big(
     \omega_{\mathrm{nw}}
     +
     \big\langle
       \vec b \wedge \vec \omega_{\mathrm{w}}
     \big\rangle
     +
     \cdots
   \big)
   \;:=\;
   \omega^{(j)}_{\mathrm{w}}
 \end{equation}
 or
 $$
   \vec\pi_\ast(\omega)
   \;:=\;
   \vec\omega_{\mathrm{w}}
   \,,
 $$
 for short. This is well-defined by the assumption that the pairing $\langle -,-\rangle$ is non-degenerate.
\end{defn}

\subsection{Higher T-duality correspondences}
\label{Sec-HigherTDualityCorrespondences}

The higher topological T-duality itself, below in Section \ref{HigherTDualityTransformations}, is a non-trivial
 isomorphism between two different twisted cohomology groups. But, first of all, the twists on the two sides of this isomorphisms need to be T-dual themselves. This is encoded in the concept of a \emph{higher T-duality correspondence} which we discuss now (Def. \ref{HigherTDualityCorrespondences} below)
being a special case of a correspondence of twisting cocycles (Def. \ref{Correspondences} below.)

\medskip
In application to super $p$-brane physics in Section \ref{Examples} below, the twisting cocycles correspond to the
charges of a given brane species that may end of some other branes (e.g. the fundamental string in type II or the
M5-brane in 11d) and the twisted cohomology classes twisted thereby correspond to the charges of the branes it
may end on (e.g. the D-branes or (possibly) the M9-brane, respectively); see \cite[Section 3]{FSS13} for details on this
homotopy-theoretic incarnation of the brane intersection laws.

\vspace{5mm}
\hspace{-13mm}
{\small
\begin{tabular}{|c||c|c|c|}
  \hline
  {\bf Concepts}
  &
  \begin{tabular}{c}
    {\bf Higher T-duality}
    \\
    {\bf correspondence}
  \end{tabular}
  &
  \begin{tabular}{c}
    {\bf Pull-tensor-push}
    \\
   {\bf through}
   \\
    {\bf correspondence}
  \end{tabular}
  &
  {\bf Higher T-duality}
  \\
  \hline
  & Def. \ref{HigherTDualityCorrespondences} & Def. \ref{def-pullpush} & Theorem \ref{HigherTDuality}, Cor. \ref{HigherTopologicalTDualityInTwistedCohomology}
  \\
  \hline
  \hline
  {\bf Examples}
  \\
  \hline
  \hline
  String theory
    &
  \begin{tabular}{c}
    \emph{Type IIA/B string} (\ref{CorrespondenceTDualityTypeII}, \ref{SuperstringCocycleTypeII})
    \\
$
      \mu_{{}_{F1}}^{\mathrm{IIA/B}}
      =
      \underset{
        \mathrm{\footnotesize basic}
      }{
      \underbrace{
        \mu_{{}_{F1}}\vert_{8+1}
      }}
      +
      e^9_{{}_{A/B}}
      \wedge
      \underset{
        \!\!\!\!\!\!\!
        { \tiny \begin{array}{c}
           \mathrm{ dual}
          \\
          \mathrm{extension}
          \\
          \mathrm{class}
        \end{array} }
      \!\!\!\!\!\!\!\!}{
        \underbrace{ c_2^{{}^{\mathrm{IIB/A}}} }
      }
    $
  \end{tabular}
  & D-branes & \begin{tabular}{c} Hori's formula for \\ Buscher rules for RR-fields \\ (\ref{TypeIIT}) from \cite[Prop. 6.4]{FSS16} \end{tabular}
  \\
  \hline
  M-theory
  &
    \begin{tabular}{c}
      \emph{M5-brane} (\ref{CorrespondenceSphericalTDualityM5}, \ref{M5cocycleInTDualityPair} )
      \\
      $
      \widetilde \mu_{{}_{M5}}
      =
      \underset{
        \!\!\mathrm{basic}\!\!
      }{
      \underbrace{
        2\mu_{{}_{M5}}
      }}
      +
      c_3 \wedge
      \underset{
        \!\!\!\!\!\!\!\!\!\!\!
        { \tiny \begin{array}{c}
          \mathrm{dual}
          \\
          \mathrm{extension}
          \\
          \mathrm{class}
        \end{array} }
      \!\!\!\!\!\!\!\!\!\!}{
        \underbrace{ \mu_{{}_{M2}} }
      }
      $
    \end{tabular}
  & Remark \ref{M5TwistedCocyclesFromHeteroticGreenSchwarz} & Sections \ref{SphericalTDualityOfM5BranesOnM2ExtendedSpacetime}, \ref{SphericalTDualityOnExceptionalMTheorySpacetime}
  \\
  \hline
\end{tabular}
}

\medskip
\begin{defn}[Higher Poincar\'e form]
\label{higherPoincareForm}
Let $n, k \in \mathbb{N}$, and let
$
    b^n (\mathfrak{u}_1)^k
    \times
    b^n (\mathfrak{u}_1)^k
    \xrightarrow{{ \left\langle (-) \cup (-) \right\rangle }}
    b^{2n+1} \mathfrak{u}_1
$
be a choice of cup product (\ref{UniversalCupProduct}) as in Def. \ref{HigherCupProducts}.
Moreover, let  $\mathfrak{g}$ be a super-$L_\infty$-algebra and let
$$
\xymatrix{
  \widevec{\mathrm{dd}}^A, \widevec{\mathrm{dd}}^B
  \;:\;
  \mathfrak{g}
  \ar[r] &
  b^{2n+1}(\mathfrak{u}_1)^k
  }
$$
be two cocycles, classifying two higher extensions $\hat{\mathfrak{g}}^A$ and $\hat{\mathfrak{g}}^B$, respectively, via Prop. \ref{hofibofSuperLInfinityAlgebras}:
\footnote{We will use the notation $(-)^{A/B}$ to indicate the two cases $A$ and $B$ respectively.
}
\begin{equation}
  \label{CEAlgebraOfddABExtensions}
  \mathrm{CE}\big(\widehat{\mathfrak{g}}^{A/B}\big)
  \;=\;
  \mathrm{CE}(\mathfrak{g})\big[ \vec b^{A/B} \big]/\big( d \big( \vec b^{A/B} \big) =
   \widevec{\mathrm{dd}}^{A/B} \big)
  \,.
\end{equation}
Then on the fiber product
\begin{equation}
  \label{FiberProductForPoincareForm}
  \raisebox{20pt}{
  \xymatrix@R=1em{
    &
    \hat{\mathfrak{g}}^A \times_{\mathfrak{g}} \hat{\mathfrak{g}}^B
    \ar[dl]_{p_A}
    \ar[dr]^{p_B}
    \\
    \hat{\mathfrak{g}}^A
    \ar[dr]_{\pi^A}
    &&
    \hat{\mathfrak{g}}^B
    \ar[dl]^{\pi^B}
    \\
    & \mathfrak{g}
  }
  }
\end{equation}
whose CE-algebra is
$$
  \mathrm{CE}
  \left(
    \hat{\mathfrak{g}}^A \times_{\mathfrak{g}} \hat{\mathfrak{g}}^B
  \right)
  \;=\;
  \mathrm{CE}(\mathfrak{g}) \big[ \vec b^A, \vec b^B  \big]
  /\big(
    d (\vec b^{A}) = \widevec{\mathrm{dd}}^{A},
    d (\vec b^{B}) = \widevec{\mathrm{dd}}^{B}
  \big)
$$
there is the cochain
\begin{equation}
  \label{HigherPoincarForm}
  \mathcal{P}
  :=
  \big\langle \vec b^A \wedge \vec b^B \big\rangle
  \;\;\in
  \mathrm{CE}\big( \hat{\mathfrak{g}}^A \times_{\mathfrak{g}} \hat{\mathfrak{g}}^B \big)
  \,,
\end{equation}
which we call the \emph{ higher Poincar{\'e} form} on this fiber product.
\end{defn}
\begin{defn}[Correspondence of twisting cocycles]
\label{Correspondences}
In the context of Def. \ref{higherPoincareForm}, if
$$
  h^{A/B} \;:\; \hat{\mathfrak{g}}^{A/B} \longrightarrow b^{4n+2} \mathfrak{u}_1
$$
are two $(4k+3)$-cocycles on the two extensions, respectively, we say that they
are \emph{in correspondence}, or that the diagram
\begin{equation}
  \label{Correspondence}
  \raisebox{20pt}{
  \xymatrix@R=1.5em{
    &
    &
    \hat{\mathfrak{g}}^A \times_{\mathfrak{g}} \hat{\mathfrak{g}}^B
    \ar[dl]_{p_A}
    \ar[dr]^{p_B}
    \\
    b^{4n+2}\mathfrak{u}_1
    &
    \hat{\mathfrak{g}}^A
    \ar[l]_{\hspace{6mm}h^A}
    &&
    \hat{\mathfrak{g}}^B
    \ar[r]^-{h^B}
    &
    b^{4k+2}\mathfrak{u}_1
  }
  }
\end{equation}
is a \emph{correspondence between} the cocycles, if the Poincar{\'e} form (\ref{HigherPoincarForm})
trivializes the differences of their pullbacks to the fiber product:
\begin{equation}
  \label{PoincareFormTrivializesDifferenceBetweenhBAndhA}
  d \mathcal{P} = (p_B)^\ast(h^B) - (p_A)^\ast(h^A)
  \,,
\end{equation}
\end{defn}

\begin{remark}
The above definition, Def. \ref{Correspondences}, may naturally be stated more homotopy-theoretically:
The algebraic condition (\ref{PoincareFormTrivializesDifferenceBetweenhBAndhA})
  witnesses a 2-dimensional diagram in the $(\infty,1)$-category of super $L_\infty$-algebras
  of the following form:
  $$
    \xymatrix@R=1.5em{
      &
      \hat{\mathfrak{g}}^A \times_{\mathfrak{g}} \hat{\mathfrak{g}}^B
      \ar[dl]_{ p_A }
      \ar[dr]^{ p_B }_{\ }="s"
      \\
      \hat{\mathfrak{g}}^A
      \ar[dr]_{h^A}^{\ }="t"
      &&
      \hat{\mathfrak{g}}^B \;,
      \ar[dl]^{h^B}
      \\
      & b^{4n+2} \mathfrak{u}_1
      \ar@{=>}^{\mathcal{P}} "s"; "t"
    }
  $$
  where $\mathcal{P}$ is viewed as a homotopy interpolating between the two
 compositions.
\end{remark}

\medskip
\begin{defn}[Higher T-duality correspondence]
\label{HigherTDualityCorrespondences}
We call a correspondence of cocycles according to Def. \ref{Correspondences}
$$
  \xymatrix@R=1.5em{
    &&
    \hat{\mathfrak{g}}^A \times_{\mathfrak{g}} \hat{\mathfrak{g}}^B
    \ar[dl]_{p_A}
    \ar[dr]^{p_B}
    \\
    b^{4n+2}\mathfrak{u}_1
    &
    \hat{\mathfrak{g}}^A
    \ar[dr]_{\pi^A := \mathrm{hofib}(\widevec{\mathrm{dd}}^A) \phantom{AAA} }
    \ar[l]_-{h^A}
    &&
    {\hat{\mathfrak{g}}}^B
    \ar[dl]^{\phantom{AAA}\pi^B := \mathrm{hofib}( \widevec{\mathrm{dd}}^B) }
    \ar[r]^-{h^B}
    &
    b^{4n+2}\mathfrak{u}_1
    \\
    &
    & \mathfrak{g}
    \ar[dl]^-{\widevec{\mathrm{dd}}^A }
    \ar[dr]_-{ \widevec{\mathrm{dd}}^B }
    \\
    &
    b^{2n+1} (\mathfrak{u}_1)^k && b^{2n+1} (\mathfrak{u}_1)^k
  }
$$
a \emph{higher topological T-duality correspondence}
if the partial integration (Def. \ref{FiberIntegration}) of the
  A-cocycle along the A-fibration coincides with minus the
  B-cocycle on the base, and conversely:
\begin{equation}
  \label{HigherTDualityPairingCondition}
    \vec \pi^A_\ast (h^A) = - \widevec{\mathrm{dd}}^B
    \qquad \text{and} \qquad
    \vec \pi^B_\ast (h^B) = - \widevec{\mathrm{dd}}^A
  \,.
\end{equation}
\end{defn}
We denote this situation by either of the following two notations:
$$
\xymatrix{
  (\widevec{\mathrm{dd}}^A,h^A) \;\; \ar@{<->}[r]^{\cal{T}} &\;\; (\widevec{\mathrm{dd}}^B, h^B)
  }
\qquad \text{or}
\qquad
\xymatrix{
  (\widehat{\mathfrak{g}}^A, h^A) \;\; \ar@{<->}[r]^{\cal{T}} & \;\; (\widehat{\mathfrak{g}}^B, h^B)
  }
  \,.
$$

\begin{remark}[The T-duality axiom]
  \label{TDualityAxiom}
 {\bf (i)}  Notice that in relations (\ref{HigherTDualityPairingCondition}), by usual abuse of notation, we are notationally suppressing the pullback of $\vec{\mathrm{dd}}^{B/A}$
  to the codomain of $\vec \pi^{A/B}_\ast$ (Def. \ref{FiberIntegration}).

\item {\bf (ii)}  This condition (\ref{HigherTDualityPairingCondition}) is the evident generalization of the
  condition for toroidal topological T-duality that has been considered before:
The latter was conjecturally proposed in \cite[(2.1)]{BHM03} and
  argued for from actual T-duality in \cite[around (2.5)]{BHM}.
  Its evident lift to integral cohomology was considered in \cite[(2.3)]{BunkeRumpfSchick06}.
  We had \emph{derived} this condition from analysis of type II superstring super-WZW terms in
  \cite[(4) and (5)]{FSS16}.

\item {\bf (iii)}  In the language of ``dg-manifolds'', this condition for ordinary T-duality
  was considered in \cite[Def. 3.2]{LupercioRengifoUribe14}.
  Notice that dg-manifolds are
  just $L_\infty$-algebras that vary over a base manifold, in a suitable sense: namely dg-manifolds are \emph{$L_\infty$-algebroids} \cite[Section A.1]{SSS3}.
  This way the setup in \cite{LupercioRengifoUribe14} is closely related to our super $L_\infty$-algebraic discussion. The key difference is
  the higher generalization of T-duality which we consider and, crucially, the generalization to
  super-dg-geometry: It is the \emph{higher} and \emph{fermionic} cocycles on super-$L_\infty$-algebras
  which induce all the nontrivial
  structure in higher T-duality of super $p$-branes, both for ordinary T-duality of type II superstrings
  from \cite{FSS16} (recalled as Section \ref{OrdinaryTypeIITDuality} below),
   as well as in the spherical T-duality of M-branes discussed below in Section \ref{Examples}.
\end{remark}
It will be convenient to have notation for the set of all possible T-duality correspondences over a given base:
\begin{defn}
[Set of higher T-duality correspondences]
Given a super $L_\infty$-algebra $\mathfrak{g}$ we write
\begin{equation}
  \label{AssigningSetsOfHigherTDualityCorrespondences}
  \mathfrak{g} \longmapsto \mathrm{TCorr}_n(\mathfrak{g})
\end{equation}
for the set of all higher T-duality correspondences over $\mathfrak{g}$; by
Prop. \ref{hofibofSuperLInfinityAlgebras}
this is a contravariant functor in $\mathfrak{g}$.
\end{defn}

\medskip
We offer the following more succinct way to \emph{induce} the data of T-duality
correspondences purely from cocycle data on $\mathfrak{g}$.

\begin{prop}[Classifying data for higher topological T-duality correspondences]
  \label{ClassifyingDataHigherTDuality}
  Given a super $L_\infty$-algebra $\mathfrak{g}$, then the set $\mathrm{TCorr}_n(\mathfrak{g})$ (\ref{AssigningSetsOfHigherTDualityCorrespondences})
  of higher T-duality correspondences above it
  is in natural bijection with the set of choices of cocycles $\widevec{\mathrm{dd}}^{A/B}$ \emph{on $\mathfrak{g}$} equipped with a trivialization of (minus) their
  cup product (Def. \ref{HigherCupProducts}):
  $$
    \mathrm{TCorr}_n(\mathfrak{g})
    \;\simeq\;
    \Big\{
      \underset{ \mathrm{deg} = 2n }{\underbrace{\widevec{\mathrm{dd}}^{A/B}_{\mathfrak{g}}}}\;,
      \underset{ 4n-1}{\underbrace{h_{\mathfrak{g}}}}
      \in \mathrm{CE}(\mathfrak{g})
      \;  \vert \;
      d \big( \widevec{\mathrm{dd}}^{A/B} \big) = 0\,,\;
      d h_{\mathfrak{g}}
        = \
      -
      \big\langle
        \widevec{\mathrm{dd}}^A_{\mathfrak{g}} \wedge \widevec{\mathrm{dd}}^B_{\mathfrak{g}}
      \big\rangle
    \Big\}
    \,.
  $$
\end{prop}\begin{proof}
  First we observe that given the data $(\widevec{\mathrm{dd}}^{A/B}, h^{A/B})$ underlying a T-duality correspondence,
  then the two conditions saying that this does indeed constitute a T-duality correspondence (Def. \ref{HigherTDualityCorrespondences}),
  \begin{enumerate}
    \item $d \mathcal{P} = (p_B)^\ast(h^B) - (p_A)^\ast(h^A) $;
    \item $(\vec\pi^{A/B})_\ast( h^{A/B} ) = - \widevec{\mathrm{dd}}^{A/B}$,
  \end{enumerate}
  are equivalent to the statement that $h^{A/B}$ has the following form:
  \begin{equation}
    \label{DecompositonOfHInTDualityCorrespondence}
    h^{A/B}
    \;=\;
    (\pi^{A/B})^\ast(h_{\mathfrak{g}})
    +
    \Big\langle
      \vec b^{A/B} \wedge (\pi^{A/B})^\ast\big( \widevec{\mathrm{dd}}^{B/A} \big)
    \Big\rangle
    \phantom{AAAA}
    \text{for some}
    \; h_{\mathfrak{g}} \in \mathrm{CE}(\mathfrak{g})
    \,.
  \end{equation}
  To see this, observe
  by the definition of fiber integration, Def. \ref{FiberIntegration},
  and noticing that the prefactor in (\ref{FiberIntegrationMap}) is $-1$ in the present case,
  that the second condition
  above is equivalent to
  $$
    h^{A/B}
    \;=\;
    (\pi^{A/B})^\ast(q^{A/B})
    +
    \Big\langle
      \vec b^{A/B} \wedge (\pi^{A/B})^\ast \big(\widevec{\mathrm{dd}}^{B/A}\big)
    \Big\rangle
    \phantom{AAA}
    \text{for some}
    \; q^{A/B} \in \mathrm{CE}(\mathfrak{g})
    \,.
  $$
  With this the first condition is equivalent to
  $$
    \begin{aligned}
      0
      & =
      (p^B)^\ast (h^B) - (p^A)^\ast( h^A ) - d \mathcal{P}
      \\
      & =
      (p^B)^ \ast(\pi^{B})^\ast(q^{B}) - (p^A)^\ast(\pi^{A})^\ast(q^{A})
      +
      \underset{
       = 0
      }{
      \underbrace{
        \left(
          (p^B)^\ast
          \left\langle
            \vec b^{B} \wedge (\pi^B)^\ast(\widevec{\mathrm{dd}}^{A})
          \right\rangle
          -
          (p^A)^\ast
          \left\langle
            \vec b^{A} \wedge (\pi^A)^\ast( \widevec{\mathrm{dd}}^{B})
          \right\rangle
        \right)
        -
        d \mathcal{P}
      }
      }
    \end{aligned}
    \,,
  $$
  where the term over the brace vanishes by the definition of $\mathcal{P}$ (expression \eqref{HigherPoincarForm}).
  But since $(p^A)^ \ast(\pi^{A})^\ast = (p^B)^ \ast(\pi^{B})^\ast$, by the fiber product diagram (\ref{FiberProductForPoincareForm}), and since the pullback operation along
  a fibration is injective,
  this is equivalent to $q^A = q^B$. Hence $h_{\mathfrak{g}} := q^{A/B}$.
  This proves the claim.

  Now from this claim it is immediate that Def. \ref{HigherTDualityCorrespondences} implies
  $$
    \begin{aligned}
      0
      & =
      (p^{A/B})^\ast ( d h^{A/B}) )
      \\
      & =
      (p^{A/B})^\ast
      \left(
         d
         \Big(
           (\pi^{A/B})^\ast(h_{\mathfrak{g}})
           +
           \Big\langle
             \vec b^{A/B} \wedge (\pi^{A/B})^\ast \big( \widevec{\mathrm{dd}}^{B/A} \big)
           \Big\rangle
         \Big)
      \right)
      \\
      & =
      (p^{A/B})^\ast
      (\pi^{A/B})^\ast
      \left(
        d h_{\mathfrak{g}}
        +
        \big\langle
          \widevec{\mathrm{dd}}^A \wedge \widevec{\mathrm{dd}}^B
        \big\rangle
      \right)\;.
    \end{aligned}
  $$
  Consequently, we have the equation $d h_{\mathfrak{g}} = - \mathrm{dd}^A \wedge \mathrm{dd}^B$, because the pullback operation
  along higher central extensions is injective, by Prop. \ref{hofibofSuperLInfinityAlgebras}.
  Conversely, given this equation, then the same computation shows that setting
  $$
    h^{A/B} := (\pi^{A/B})^\ast( h_\mathfrak{g} ) + b^{A/B} \wedge (\pi^{A/B})^\ast( \mathrm{dd}^{B/A} )
  $$
  defines cocycles. Then, as in the proof of the claim above, it follows that these satisfy the two conditions above.
\end{proof}
This immediately implies that there is a universal $L_\infty$-algebra that serves as the classifying space for
higher topological T-duality correspondences, provided we give it the right context first.
\begin{defn}
[Super $L_\infty$ T-fold algebra]
For $n,k \in \mathbb{N}$, $n \geq 1$ write
$$
  b\mathfrak{tfold}_n
  \;\in\;
  sL_\infty\mathrm{Alg}
$$
for the super $L_\infty$-algebra
dually defined as having the following Chevalley-Eilenberg algebra:
\begin{equation}
 \label{btfoldnAlgebra}
 \mathrm{CE}(b\mathfrak{tfold}_n)
  \;
  :=
  \;
  \mathbb{R}\Big[
    \underset{\mathrm{deg} = 2n }{\underbrace{ \widevec{\mathrm{dd}}^A }}
    \hspace{-1mm}, \;
    \underset{ 2n}{\underbrace{  \widevec{\mathrm{dd}}^B}} \;,
    \underset{   4n -1 }{\underbrace{h}}
  \Big]/
  \Big(
       d  \big( \widevec{\mathrm{dd}}^{A/B} \big) = 0, \;
       d h
         =
       - \big\langle \widevec{\mathrm{dd}}^A \wedge \widevec{\mathrm{dd}}^B \big\rangle
  \Big)\;.
\end{equation}
\end{defn}

\begin{prop}[Higher T-duality $L_\infty$-algebra]
\label{HigherTDualityLInfinityAlgebra}
For $\mathfrak{g} \in s L_{\infty}\mathrm{Alg}$ any super $L_\infty$-algebra there is a natural bijection
$$
  \mathrm{Hom}( \mathfrak{g}, b \mathfrak{tfold}_n )
  \;\simeq\;
  \mathrm{TCorr}_n(\mathfrak{g})
$$
between the set of super $L_\infty$-homomorphisms of the form $\mathfrak{g} \to b\mathfrak{tfold}_n$
and the set of higher T-duality correspondences over $\mathfrak{g}$ (Def. \ref{HigherTDualityCorrespondences}).
\end{prop}
\begin{proof}
This is the composite of natural bijections
\begin{align*}
  \mathrm{Hom}( \mathfrak{g}, b \mathfrak{tfold}_n )
  \; &\simeq\;
  \Big\{
    \hspace{-2mm}\underset{\mathrm{deg} = 4n-1}{\underbrace{h_{\mathfrak{g}}}}
    \hspace{-2mm}, \;
    \underset{  2n }{\underbrace{ \widevec{\mathrm{dd}}^{A/B}_{\mathfrak{g}}}}
    \in \mathrm{CE}(\mathfrak{g})
    \;\vert\;
    d \big(\widevec{\mathrm{dd}}^{A/B}\big) = 0,
    d h_{\mathfrak{g}}
      =
    - \left\langle \widevec{\mathrm{dd}}^A_{\mathfrak{g}} \wedge \widevec{\mathrm{dd}}^B_{\mathfrak{g}} \right\rangle
  \Big\}
  \\
  \; &\simeq\;
  \mathrm{TCorr}_n(\mathfrak{g})
  \,,
\end{align*}
where the first one is given by expression \eqref{btfoldnAlgebra}
and the second is given by Prop. \ref{ClassifyingDataHigherTDuality}.
\end{proof}

\subsection{Higher T-duality transformations}
\label{HigherTDualityTransformations}

Finally we discuss how every correspondence of cocycles (Def. \ref{Correspondences}) induces a pull-push transformation on cochains (Def. \ref{def-pullpush} below),
which is an isomorphism on the corresponding higher twisted cohomology groups (Theorem \ref{HigherTDuality} below).
Applied to the case of higher T-duality correspondences (Def. \ref{HigherTDualityCorrespondences}) this yields the genuine
higher topological T-duality (Corollary \ref{HigherTopologicalTDualityInTwistedCohomology} below).

\begin{defn}[Pull-push through correspondences]
\label{def-pullpush}
  Consider a correspondence (Def. \ref{Correspondences})
  $$
    \xymatrix@R=1.5em{
      &&
      \hat{\mathfrak{g}}^A \times_{\mathfrak{g}} \hat{\mathfrak{g}}^b
      \ar[dl]_{p_A}
      \ar[dr]^{p_B}
      \\
      b^{4n+2}\mathbb{R}
      &
      \hat{\mathfrak{g}}^A
      \ar[l]_-{\;\;\;h^A}
      &&
      {\hat{\mathfrak{g}}}^B
      \ar[r]^-{h^B}
      &
      b^{4n+2}\mathbb{R}\;.
    }
  $$
  Then we say that the \emph{pull-push transform through the correspondence} is
  the linear map $(p_A)_\ast \circ e^{\mathcal{P}} \circ (p_B)^\ast $ from the cochains on $\hat{\mathfrak{g}}^B$ to those of $\hat{\mathfrak{g}}^A$
  which is the composite of
  \begin{enumerate}
    \item pullback $(p_B)^\ast$ to the fiber product,
    \item multiplication with the exponential $e^{\mathcal{P}} := 1 + \mathcal{P} + \tfrac{1}{2} \mathcal{P} \wedge \mathcal{P} + \cdots$  of the Poincar{\'e} form (\ref{HigherPoincarForm}), and
    \item fiber integration $(p_A)_\ast$ (Def. \ref{FiberIntegration}).
  \end{enumerate}
\end{defn}

\begin{theorem}[Pull-push through correspondences is isomorphism on twisted cohomology]
  \label{HigherTDuality}
 {\bf (i)} The pull-push through correspondences of Def. \ref{def-pullpush}
  induces an isomorphism between the corresponding twisted
  cohomology groups (Def. \ref{TwistedCohomology}):
  $$
    \xymatrix{
      H^{(\bullet + 2n+1) + h^A}(\hat{\mathfrak{g}}^A)
      \ar@{<-}[rrrr]^-{{\cal T} := (p_A)_\ast e^{ \mathcal{P} }  (p_B)^\ast}_-{\simeq}
      &&&&
      H^{\bullet+h^B}(\hat{\mathfrak{g}}^B)
    }
    \,.
  $$
 \item {\bf (ii)} Moreover, if the base super $L_\infty$-algebra $\mathfrak{g}$ is equipped with
 an action by a group $K$ (Def. \ref{SuperLInfinityAlgebraWithGroupAction}),
  and if the cocycles $\vec{\mathrm{dd}}^{A/B} \in \mathrm{CE}(\mathfrak{g})^K$ are $K$-invariant (\ref{KInvariantCEComplex}),
  so that also the extensions $\widehat{\mathfrak{g}}^{A/B}$ are canonically equipped with a $K$-actions (by Prop. \ref{hofibofSuperLInfinityAlgebras})
  and if finally the twisting cocycles $h^{A/B} \in \mathrm{CE}(\widehat{\mathfrak{g}}^{A/B})^K$ are $K$-invariant, then
  this isomorphism restricts to an isomorphism of $K$-invariant twisted cohomology groups
  $$
    \xymatrix{
      H^{(\bullet + 2n+1) + h^A}(\hat{\mathfrak{g}}^A/K)
      \ar@{<-}[rrrr]^-{{\cal T} := (p_A)_\ast e^{ \mathcal{P} }  (p_B)^\ast}_-{\simeq}
      &&&&
      H^{\bullet+h^B}(\hat{\mathfrak{g}}^B/K)
    }
    \,.
  $$
\end{theorem}
\begin{proof}
First consider the case that $k = 1$ in Def. \ref{higherPoincareForm}. This means that there is a single generator $b^B$ and a single generator $b^A$.
Hence in this case the Poincar{\'e} form $\mathcal{P}$ given in expression
(\ref{HigherPoincarForm}) is a linear multiple of $b^A \wedge b^B$. Without
restriction of generality we may take the linear multiple to be one. Then
$$
  e^{\mathcal{P}} = 1 + b^A \wedge b^B
  \,.
$$
Now,  let (see Remark \ref{wrap-unwrap})
\begin{equation}
  \label{DecompositionOfCochainIntoWrappingAndNonWrappingComponent}
  \omega
  \;=\; \omega_{\mathrm{nw}} + b^B \wedge \omega_{\mathrm{w}}
  \;\in \mathrm{CE}( \hat{\mathfrak{g}}^B)
\end{equation}
be any, possibly inhomogeneous, cochain on $\hat{\mathfrak{g}}^B$,
where on the right we are displaying its unique coefficients $\omega_{\mathrm{nw}}, \omega_{\mathrm{w}} \in \mathrm{CE}(\mathfrak{g})$
with respect to the fibration $\pi^B$ as in expression (\ref{CochainCoefficientsWrapping}).

We directly compute the action of the transformation on this cochain in terms of generators:
\begin{equation}
  \label{PullTensorPushActionOnGeneralElement}
  \begin{aligned}
    (p_A)_\ast e^{ \mathcal{P} }  (p_B)^\ast
    (\omega)
    & =
    (p_A)_\ast e^{ \mathcal{P} }  (p_B)^\ast
    \big(
      \omega_{\mathrm{nw}}
      +
      b^B \wedge \omega_{\mathrm{w}}
    \big)
    \\
    & =
    (\pi^B)_\ast
    \big(\;
      \underset{e^{\mathcal{P}}}{\underbrace{(1 + b^A \wedge b^B)}}
      \wedge
      (
        \omega_{\mathrm{nw}}
        +
        b^B \wedge \omega_{\mathrm{w}}
      )
    \big)
    \\
    & =
    - \omega_{\mathrm{w}} +   b^A \wedge \omega_{\mathrm{nw}}
    \;\;\;
    \in \mathrm{CE}(\hat{\mathfrak{g}}^A)
    \,.
  \end{aligned}
\end{equation}
This says that the transform just swaps the wrapping/non-wrapping components of cochains, up to a sign.
Hence it is manifestly a linear isomorphism on cochains.

For general $k$ the argument is directly analogous: Generally, each cochain is expanded in coefficients of
monomials $b^{B}_{j_1} \wedge \cdots \wedge b^{B}_{j_r}$, and the operation
$(\pi^B)_\ast e^{\mathcal{P}}$ amounts to re-interpreting these as coefficients of the corresponding
Hodge dual powers of the $b^A$'s. Since every monomial in the $b^B$'s has a unique Hodge dual
monomial in the $b^A$'s, this is still a linear isomorphism on cochains.

Therefore, to conclude it is sufficient to see that
the transform operation intertwines the twisted differentials
$$
  (d + h^A)
    \circ
  (\pi_A)_\ast e^{\mathcal{P}} (\pi_B)^\ast
  =
  -
  (\pi_A)_\ast e^{\mathcal{P}} (\pi_B)^\ast
    \circ
  (d + h^B)
  \,.
$$
We may check this as follows:
\begin{equation}
  \label{PullPushIntertwinedTwistedDifferentials}
  \begin{aligned}
    (d + h^A)
    \;
    (p_A)_\ast e^{\mathcal{P}} (p_B)^\ast (\omega)
    &
    =
    -
    (p_A)_\ast
    (d + h^A)
    e^{\mathcal{P}}
    (p_B)^\ast (\omega)
    \\
    & =
    -
    (p_A)_\ast
    e^{\mathcal{P}}
    (d + \underset{h^B}{\underbrace{h^A + d \mathcal{P}}})
    (p_B)^\ast (\omega)
    \\
    & =
    -
    (p_A)_\ast
    e^{\mathcal{P}}
    (p_B)^\ast
    \;
    (d + h^B)
   (\omega)
   \,.
  \end{aligned}
\end{equation}
Here in the first step we used that the plain differential graded-commutes with fiber integration by (\ref{FiberIntegrationRespectsDifferentials}),
as does multiplication by $h^A$, trivially. Then
under the brace we applied the defining condition (\ref{PoincareFormTrivializesDifferenceBetweenhBAndhA}) for a correspondence.
\end{proof}

\begin{cor}[Higher topological T-duality in twisted cohomology]
  \label{HigherTopologicalTDualityInTwistedCohomology}
Since every higher T-duality correspondence (Def. \ref{HigherTDualityCorrespondences}) is in particular
  a correspondence of cocycles in the sense of Def. \ref{Correspondences}, its induced
  pull-push transform (Def. \ref{def-pullpush}) is an isomorphism in higher twisted cohomology, by Theorem \ref{HigherTDuality}.
    We call this the actual \emph{higher topological T-duality} induced by the T-duality correspondence:
  $$
    \xymatrix@R=5pt{
      H^{(\bullet + 2n+1) + h^A}(\hat{\mathfrak{g}}^A, \mathbb{R})
      \ar@{<-}[rrrr]^-{ {\cal T} := (p_A)_\ast e^{ \mathcal{P} }  (p_B)^\ast}_-{\simeq}
      &&&&
      H^{\bullet+h^B}(\hat{\mathfrak{g}}^B, \mathbb{R})
      \\
      [- \omega_{\mathrm{w}} + b^A \wedge \omega_{\mathrm{nw}})]
      \;
      \ar@{<-|}[rrrr]
      &&&&
      \;[\omega_{\mathrm{nw}} + b^B \wedge \omega_{\mathrm{w}}]\;.
    }
  $$
  \end{cor}

  \begin{remark}[component-wise analysis of higher T-duality]
  \label{ComponentWiseAnalysisOfHigherTDuality}
  While the slick computation in (\ref{PullPushIntertwinedTwistedDifferentials}) implies that this map on cochains
  indeed respects the twisted differentials, it is instructive to check this alternatively
  in components, in terms of the characteristic condition (\ref{HigherTDualityPairingCondition})
  on a higher T-duality correspondence. This will pave the way for the generalization of higher
  T-duality to higher T-duality of decomposed form fields in Theorem
  \ref{HigherTDualityForDecomposedFields} below:

  Condition (\ref{HigherTDualityPairingCondition}) implies that the cocycles $h^{A/B}$ decompose as in (\ref{DecompositonOfHInTDualityCorrespondence}).
  Using this and collecting coefficients of $b^A$ and $b^B$, one obtains the respect for the twisted
  differentials under $\mathcal{T}$ as follows:
  \begin{equation}
  \label{HigherTDualityTransformationInComponents}
  \hspace{-4mm}
   \begin{aligned}
     (d + h^A) {\cal T}(\omega)
    & =
    (d + h^A)
    (- \omega_{\mathrm{w}} + b^A \wedge \omega_{\mathrm{nw}})
    \\
    & =
    -
    \Big(
      -
      \underset{
        = \left((d + h^B)\omega\right)_{\mathrm{w}}
      }{
      \underbrace{
      \big(
        -
        d \omega_{\mathrm{w}}
        -
        h_{\mathfrak{g}} \wedge \omega_{\mathrm{w}}
        +
        \mathrm{dd}^A \wedge \omega_{\mathrm{nw}}
      \big)
      }
      }
      +
      b^A
      \wedge
      \underset{
        = \left((d + h^B)\omega\right)_{\mathrm{nw}}
      }{
      \underbrace{
      \big(
        d \omega_{\mathrm{nw}}
        +
        h_{\mathfrak{g}} \wedge \omega_{\mathrm{nw}}
        +
        \mathrm{dd}^B \wedge \omega_{\mathrm{w}}
      \big)
      }}\;
    \Big)
    \\
    & =
    -{\cal T}\big( (d + h^B)(\omega) \big)
    \,.
  \end{aligned}
\end{equation}
\end{remark}

\subsection{Higher T-duality for decomposed form fields}
\label{HigherTDualityForDecomposedFormFields}

Corollary \ref{HigherTopologicalTDualityInTwistedCohomology} shows that every higher topological T-duality correspondence (Def. \ref{HigherTDualityCorrespondences}) induces an isomorphism in higher twisted super $L_\infty$-cohomology (Def. \ref{TwistedCohomology}), obtained there as an example of a general class of pull-push transforms through correspondences of cocycles
(Theorem \ref{HigherTDuality}). But the computation in \eqref{HigherTDualityTransformationInComponents} shows that this T-duality isomorphism
may alternatively be understood without reference to either the correspondence space or the Poincar{\'e} form on it (Def. \ref{higherPoincareForm}).
Instead, a brief reflection on (\ref{HigherTDualityTransformationInComponents}) reveals that this alternative proof relies, apart from the T-duality condition \eqref{HigherTDualityPairingCondition} itself,
only on the fact that in the Chevalley-Eilenberg algebras of $\widehat{g}^{A/B}$ (see \eqref{CEAlgebraOfddABExtensions}) every cochain
has a \emph{unique} decomposition of the form $\omega = \omega_{\mathrm{nw}} + b^{A/B} \wedge \omega_{\mathrm{w}}$ with
\emph{unique} coefficients $\omega_{\mathrm{nw}}, \omega_{\mathrm{w}}$; see expression  \eqref{DecompositionOfCochainIntoWrappingAndNonWrappingComponent}.

\medskip
While such a unique decomposition is of course immediate for generators $b^{A/B}$ in a free graded-commutative
algebra (re-amplified as Example \ref{TrivialHCohomologyOfGeneratorInHigherCentralExtension} below)
it is not restricted to this situation. Indeed, the same happens equivalently (Prop. \ref{VanishingHCohomologyEquivalentToUniqueExpansionInH} below) in algebras on which the \emph{C-cohomology} (Def. \ref{CCohomology} below) of the given element vanishes. Hence in this situation higher topological T-duality generalizes; this is Theorem \ref{HigherTDualityForDecomposedFields} below. Below in Section \ref{SphericalTDualityOnExceptionalMTheorySpacetime} we discover a curious example of this
theorem in M-brane physics.

\begin{defn}[C-cohomology]
  \label{CCohomology}
  Let $\mathcal{A}$ be a graded-commutative algebra in characteristic zero, and let
  $
    C \in \mathcal{A}
  $
  be an element in odd degree, hence multiplicatively nilpotent:
  $$
    C^2 := C\cdot C = 0 \;\in\;\mathcal{A}
    \,.
  $$
  Then the cochain cohomology of the resulting complex with differential $C \cdot (-)$
  we call the \emph{C-cohomology} $H(\mathcal{A},C)$ of $C$:
  $$
    H(\mathcal{A},C)
    :=
    \frac{
       \mathrm{ker}(C\cdot(-))
    }
    {
      \mathrm{im}(C\cdot(-))
    }
    \,.
  $$
\end{defn}
For closed elements of degree-3 in a de Rham algebra of differential forms, this cohomology is called `H-cohomology'' in \cite[p. 19]{Cav05}, but of course the concept as such is elementary and appears elsewhere under different names or under no special name, e.g. \cite[p. 1]{Sev05}. Since the letter ``$H$'' is
alluding to the NS-NS field strength for the string, which is closed, for emphasis we use ``$C$'' to allude instead to the supergravity C-field,
which is not, in general, closed. So our formulation is more general than previous ones.
With the decomposition into wrapping and non-wrapping modes (Remark \ref{wrap-unwrap}), we have
the following.

\begin{prop}[Vanishing C-cohomology equivalent to unique expansions in $C$]
  \label{VanishingHCohomologyEquivalentToUniqueExpansionInH}
  Let $\mathcal{A}$ be a graded-commutative algebra in characteristic zero, and let
  $
    C \in \mathcal{A}
  $
  be an element in odd degree.
  Then the following are equivalent:
    \item {\bf (i)} There exists a linear subspace $\mathcal{A}_0{\hookrightarrow} \mathcal{A}$ such that
     every element $\omega \in \mathcal{A}$ has an expansion
     $$
       \omega = \omega_{\mathrm{nw}} + C \cdot \omega_{\mathrm{w}}
     $$
     for \emph{unique} $\omega_{\mathrm{w}}, \omega_{\mathrm{nw}} \in \mathcal{A}_0$.
    \item {\bf (ii)}
      The C-cohomology of $C$ (Def. \ref{CCohomology}) vanishes.
\end{prop}
\begin{proof}
  Having a unique expansion of the form $\omega_{\mathrm{nw}} + C \cdot \omega_{\mathrm{w}}$ with $\omega_{\mathrm{w}}, \omega_{\mathrm{nw}} \in \mathcal{A}_0$ for every $\omega$ in $\mathcal{A}$
  is equivalent to saying that the linear map
  \begin{align*}
  \phi_C\colon \mathcal{A}_0\oplus\mathcal{A}_0&\longrightarrow \mathcal{A}\\
  ( \omega_{\mathrm{nw}},\omega_{\mathrm{w}})&\longmapsto \omega_{\mathrm{nw}} + C \cdot \omega_{\mathrm{w}}
  \end{align*}
  is an isomorphism.
  In one direction, assume that $\phi_C$ is an isomorphism, and
  consider an element $\omega \in \mathcal{A}$ which is a C-cocycle, i.e., such that $C \cdot \omega = 0$.
  By the nilpotency of $C$ this implies $C \cdot \omega_{\mathrm{nw}} = 0$, and so $\phi_C(0,\omega_{\mathrm{nw}})=\phi_C(0,0)$. As we are assuming $\phi_C$ is an isomorphism, we get
  $\omega_{\mathrm{nw}} = 0$. This in turn means that $\omega = C \cdot \omega_{\mathrm{w}}$,
  hence that $\omega$ is a C-coboundary. This shows that the
  C-cohomology vanishes.

  Conversely, assume that the C-cohomology vanishes, i.e., $\mathrm{ker}(C\cdot(-))=\mathrm{im}(C\cdot(-))$, and let $\mathcal{A}_0\subseteq \mathcal{A}$ be a linear complement of $\mathrm{im}(C\cdot(-))$, so that we have a linear direct sum decomposition $\mathcal{A}=\mathcal{A}_0\oplus \mathrm{im}(C\cdot(-))$. As $C\cdot C=0$, this implies that $C\cdot \mathcal{A}=C\cdot\mathcal{A}_0$, and so the map $\phi_C\colon \mathcal{A}_0\oplus \mathcal{A}_0\to \mathcal{A}$ is surjiective. If $( \omega_{\mathrm{nw}},\omega_{\mathrm{w}})\in \ker\phi_C$, then $\omega_{\mathrm{nw}}+C\cdot\omega_{\mathrm{w}}=0$ with $\omega_{\mathrm{nw}}\in \mathcal{A}_0$ and $C\cdot\omega_{\mathrm{w}}\in \mathrm{im}(C\cdot(-))$. Since $\mathcal{A}_0$ is a linear complement of $\mathrm{im}(C\cdot(-))$ in $\mathcal{A}$, this gives
  $\omega_{\mathrm{nw}}=0$ and $C\cdot\omega_{\mathrm{w}}=0$. The second equation gives $\omega_{\mathrm{w}}\in \mathcal{A}_0\cap \mathrm{ker}(C\cdot(-))= \mathcal{A}_0\cap\mathrm{im}(C\cdot(-))=0$. So $\phi_C$ is injective and therefore an isomorphism.
  This finishes the proof.

  Alternatively, we may argue for the converse direction more abstractly as follows.
  That the C-cohomology vanishes means that
  we have a long exact sequence of vector spaces
  \[
  \cdots\longrightarrow \mathcal{A}\xrightarrow{\;\;C\cdot\;\;} \mathcal{A}\xrightarrow{\;\;C\cdot\;\;}
   \mathcal{A}\xrightarrow{\;\;C\cdot\;\;}\mathcal{A}\longrightarrow \cdots.
  \]
   Since exact sequences of vector spaces split, this fits into the commutative diagram
  $$
    \hspace{-1mm}
    \xymatrix@R=2em{
      \cdots
      \ar[r]
      &
      \mathcal{A}
      \ar[rr]^{C \cdot}
      \ar@{->>}[dr]_{C\cdot}
      &&
      \mathcal{A}
      \ar[rr]^{C \cdot}
      \ar@{->>}[dr]_{C\cdot}
      &&
      \mathcal{A}
        \ar[r]
      &
      \cdots
      \\
      &&
      (C)
      \ar@{^{(}->}[ur]^{\scriptscriptstyle{~~\iota\,\,}}
      \ar@{^{(}->}[dr]^{\scriptscriptstyle{(\mathrm{id},0)\,\,}}
      \ar@{_{(}->}[dl]_-{\scriptscriptstyle{~~(0,(C\cdot)^{-1})\,\,}}
      &&
      (C)
      \ar@{^{(}->}[dr]^{\scriptscriptstyle{(\mathrm{id},0)\,\,}}
      \ar@{^{(}->}[ur]^{\scriptscriptstyle{~~\iota\,\,}}
      \ar@{_{(}->}[dl]_-{\scriptscriptstyle{~~(0,(C\cdot)^{-1})\,\,}}
      \\
      \cdots
      \ar[r]
      &
      (C)
      \oplus
      \mathcal{A}/(C)
      \ar[uu]^{\iota+\sigma}_{\simeq}
      \ar[rr]_{\footnotesize
        \begin{pmatrix} 0 & 0 \\ C\cdot & 0 \end{pmatrix}      }
      &&
      (C)
      \oplus
      \mathcal{A}/(C)
      \ar[uu]^{\iota+\sigma}_{\simeq}
      \ar[rr]_{\footnotesize
          \begin{pmatrix} 0 & 0 \\ C\cdot & 0 \end{pmatrix}
            }
      &&
      (C) \oplus \mathcal{A}/(C)
      \ar[uu]^{\iota+\sigma}_{\simeq}
      \ar[r]
      &
      \cdots
    }
  $$
  where $C=\mathrm{im}(C\cdot)$ is the ideal of $\mathcal{A}$ generated by the element $C$, the morphism $\iota\colon(C)\hookrightarrow \mathcal{A}$ is the inclusion, and $\sigma\colon \mathcal{A}/(C)\hookrightarrow \mathcal{A}$ is a linear section of the natural projection $\mathcal{A}\to \mathcal{A}/(C)$. The isomorphism $(C\cdot)^{-1}\colon (C)\xrightarrow{\sim} \mathcal{A}/(C)$ is the inverse of the natural isomorphism
   $$
      \xymatrix{
      \mathcal{A}/(C)= \mathcal{A}/\mathrm{im}(C\cdot)=\mathcal{A}/\ker(C\cdot)
      \ar[r]^-{C\cdot }_-{\simeq}&
       \mathrm{im}(C\cdot)=(C)
       }
  $$
  induced by the vanishing of C-cohomology. Under this identification between $(C)$ and $\mathcal{A}/(C)$, the isomorphisms in the above diagram become
  \begin{equation}
    \label{VanishingCCohomologyIsomorphism}
    \xymatrix{
      \mathcal{A}/(C) \oplus \mathcal{A}/(C)
      \ar[rr]^-{ (C\cdot)+ \sigma }_-\simeq
      &&
      \mathcal{A}
    }
    \,.
  \end{equation}
 As a result, setting $\mathcal{A}_0=\sigma(\mathcal{A}/(C))$, one sees that there is a unique decomposition of elements of $\mathcal{A}$ as claimed.
 \end{proof}

The following example
highlights the case in which C-cohomology of an element $C$ vanishes
for the trivial reason that the element $C$ is a free generator.
The purpose of the concept of C-cohomology here is to allow for generalization away from this example.

\begin{example}[Trivial C-cohomology of fiber generator in higher central extension]
  \label{TrivialHCohomologyOfGeneratorInHigherCentralExtension}
  Let
  $$
    \xymatrix{
      b^{2n}\mathfrak{u}_1
      \ar@{^{(}->}[r]
      & \widehat{\mathfrak{g}}
      \ar[d]_-{\mathrm{hofib}(\mathrm{dd})}
      \\
      &
      \mathfrak{g}
      \ar[rr]^-{\mathrm{dd}}
      &&
      b^{2n+1} \mathfrak{u}_1
    }
  $$
  be a higher central extension, classified by a cocycle $\mathrm{dd}$ in even degree. Then the
  element
  $
    b \in \mathrm{CE}(\widehat{\mathfrak{g}})
    $
  of degree $2n+1$, from Prop. \ref{hofibofSuperLInfinityAlgebras}, has vanishing C-cohomology (Def. \ref{CCohomology}).
This is because $b$ is a free generator in the underlying graded-commutative algebra of
  $\mathrm{CE}(\widehat{\mathfrak{g}}) = \mathrm{CE}(\mathfrak{g})[b]/(d b = \mathrm{dd})$, which
  by definition means that every element $\omega \in \mathrm{CE}(\widehat{\mathfrak{g}})$ has a unique expansion
  $$
    \omega =\omega_{\mathrm{nw}} + b \wedge \omega_{\mathrm{w}}
  $$
  for unique
  $$
  \xymatrix{
    \omega_{\mathrm{nw}}, \; \omega_{\mathrm{w}}
    \;\in\;
    \mathrm{CE}( \mathfrak{g} )
      \;\ar@{^{(}->}[r] &
    \mathrm{CE}( \mathfrak{g} )[b]
    }\,,
  $$
  as required by Prop. \ref{VanishingHCohomologyEquivalentToUniqueExpansionInH}.
 This unique expansion is used notably in the definition of fiber integration along higher central extensions in Def. \ref{FiberIntegration}.
\end{example}

\medskip

With appeal to vanishing C-cohomology, we may thus generalize the concept of higher T-duality:

\begin{theorem}[Higher T-duality for decomposed fields]
  \label{HigherTDualityForDecomposedFields}
 Consider a higher self-T-duality correspondence according to Def. \ref{HigherTDualityCorrespondences}
$$
  \xymatrix@R=2.5em{
    b^{4n+2}\mathfrak{u}_1
    &
    \hat{\mathfrak{g}}
    \ar[dr]_{\pi := \mathrm{hofib}({\mathrm{dd}})}
    \ar[l]_-{h}
    &&
    {\hat{\mathfrak{g}}}
    \ar[dl]^{\pi := \mathrm{hofib}( {\mathrm{dd}}) }
    \ar[r]^-{h}
    &
    b^{4n+2}\mathfrak{u}_1
    \\
    &
    & \mathfrak{g}
    \ar[dl]_{ {\mathrm{dd}} }
    \ar[dr]^{ {\mathrm{dd}} }
    \\
    &
    b^{2n+1} \mathfrak{u}_1 && b^{2n+1}  \mathfrak{u}_1
  }
$$
possibly equipped with the action of a group $K$ (Def. \ref{SuperLInfinityAlgebraWithGroupAction}) such that all
cocycles
are $K$-invariant (Example \ref{InvariantCocycle})
Let
$
  \xymatrix{
    \mathfrak{e}
    \ar[r]
    &
    \mathfrak{g}
  }
$
be a fibration over the base, equipped with a $K$-invariant transgression element (Def. \ref{TransgressionElements})
$$
  \mathrm{cs} \;\in\; \mathrm{CE}(\mathfrak{e})^K
$$
for $\mathrm{dd}$, where such that
  \item {\bf (a)} the C-cohomology (Def. \ref{CCohomology}) of the transgression element $\mathrm{cs}$ vanishes in $\mathrm{CE}(\mathfrak{e})^K$;
  \item {\bf (b)} the  inclusion $\iota$ from Prop. \ref{VanishingHCohomologyEquivalentToUniqueExpansionInH} may be found
    such as to contain the CE-algebra of the base
    $$
      \xymatrix{
        \mathrm{CE}(\mathfrak{e})^G/(\mathrm{cs})
       \; \ar@{^{(}->}[rr]^-{\iota}
        &&
        \mathrm{CE}(\mathfrak{e})^{K}\;.
        \\
        &
        \mathrm{CE}(\mathfrak{g})^K
        \ar@{_{(}->}[ul]
        \ar@{^{(}->}[ur]
      }
    $$
Then there is an isomorphism of $(\phi^\ast(h)$-twisted cohomology groups (Def. \ref{TwistedCohomology}) which covers the higher T-duality isomorphism from Example \ref{HigherTopologicalTDualityInTwistedCohomology},
in that we have a diagram:
$$
  \xymatrix@R=5pt{
    H^{ (\bullet + 2n +1) + h^A}(\widehat{\mathfrak{g}}^A)^K
    \ar[ddd]_-{(\phi^A)^\ast}
    \ar@{<-}[rrr]^-{ {\cal T} }_-{\simeq}
    &&&
    H^{\bullet+h^B}(\widehat{\mathfrak{g}})^K
    \ar[ddd]^-{(\phi^B)^\ast}
    \\
    \\
    \\
    H^{ (\bullet + 2n + 1)+ h}( \mathfrak{e})^K
    \ar@{<-}[rrr]^-{ {\cal T}_{\mathrm{comp}} }_-{\simeq}
    &&&
    H^{\bullet+h^B}(\mathfrak{e}^{B})^K
    \\
    [- \omega_{\mathrm{w}} + (\phi)^\ast(b) \wedge \omega_{\mathrm{nw}})]
    \;
    \ar@{<-|}[rrr]
    &&&
    \;[ \omega_{\mathrm{nw}} + (\phi)^\ast(b) \wedge \omega_{\mathrm{w}} ]\;.
  }
$$
Here the vertical morphisms come from pullback along the classifying morphisms (\ref{ClassifyingMolrphismsForTransgressionElements})
$$
  \xymatrix{
    \mathfrak{e}
      \ar[rr]^-{\phi}
    &&
      \widehat{\mathfrak{g}}
  }
$$
according to Example \ref{HigherCentralExtensionTransgressionElement}
\end{theorem}
\begin{proof}
  By Prop. \ref{VanishingHCohomologyEquivalentToUniqueExpansionInH}, the vanishing of the
  C-cohomology
  implies that the decompositions $\omega_{\mathrm{nw}} + \mathrm{cs}^{A/B} \wedge \omega_{\mathrm{w}}$
  are unique. Consequently, the linear map on cochains
  $$
    \xymatrix@R=5pt{
       - \omega_{\mathrm{w}} + \mathrm{cs}^A \wedge \omega_{\mathrm{nw}}
      \;
      \ar@{<-|}[rr]
      &&
      \;\omega_{\mathrm{nw}} + \mathrm{cs}^B \wedge \omega_{\mathrm{w}}
    }
  $$
  is a well-defined linear isomorphism.
  With this it is now sufficient to see that this linear isomorphism on cochains intertwines the
  twisted differentials. But
  $$
    d \mathrm{cs}
    \;=\;
    \mathrm{dd}
    \,,
  $$
  by the defining condition on transgression elements,
  this follows verbatim by the same computation as in expression (\ref{HigherTDualityTransformationInComponents}).
\end{proof}


\section{Higher T-duality of M-Branes }
\label{Examples}

We now present and discuss examples of the super $L_\infty$-algebraic higher T-duality that we introduced in Section \ref{sphericaltduality}.

First we observe in Section \ref{OrdinaryTypeIITDuality} that the super-topological T-duality of F1/D$p$-branes on 10d type II super-Minkowski spacetimes
established in \cite{FSS16} is an example. This serves to put the generalization to higher T-dualities in the following examples into perspective.

To illustrate that there are further examples even of ordinary (i.e. non-higher) super-topological T-duality we observe in Section \ref{SelfTDualityIn6d} that there is
super-topological T-self-duality for superstrings on 6d super spacetime.

After this warmup, we pass attention to the higher T-duality of genuine interest here:

\vspace{1mm}
\noindent {\bf (i)} First we observe in Section \ref{SphericalTDualityOfM5BranesOnM2ExtendedSpacetime} that the fact that the joint supercocycle for the M2/M5 brane takes values in the rational 4-sphere \cite{FSS15} implies, by Prop. \ref{ClassifyingDataHigherTDuality}, spherical self-T-duality of the M5-brane on the
M2-brane extended superspacetime.

\vspace{1mm}
\noindent {\bf (ii)} In order to gain a deeper understanding of what this means, we turn attention in Section \ref{ToroidalTDualityOnExceptionalMTheorySpacetime} to 11d exceptional superspacetime and show that it exhibits 528-toroidal T-duality over the superpoint, and 517-toroidal T-duality over ordinary 11d superspacetime.

\vspace{1mm}
\noindent {\bf (iii)} Then in Section \ref{SphericalTDualityOnExceptionalMTheorySpacetime} we first recall the decomposition of the C-field over the 11d exceptional tangent superspacetime due to \cite{DF,BAIPV04}. Then we compute the C-cohomology of the decomposed C-field in Prop. \ref{HCohomologyOfDecomposedCFieldVanishesInSubalgebra}
and thus establish that the spherical T-duality of M5-branes passes to the exceptional tangent superspacetime (Prop. \ref{SphericalTDualityOnExtendedSuperspacetime} below).

\vspace{1mm}
\noindent {\bf (iv)} To conclude the role of exceptional super-spacetime in spherical T-duality, we explain
in section \ref{ExceptionalGenralisedGeometry} how the decomposed C-field on the exceptional super tangent spacetime
realizes the proposal of \cite{Hull07, Wes03} (see Remark \ref{ExceptionalSuperTangentSpacetime} below) that M-theoretic field configurations should have a moduli space in exceptional generalized geometry. We observe that spherical T-duality implements a duality relation on these moduli spaces which renders
duality-equivalent the decomposition of the C-field at different values of the parameter $s$.

\vspace{1mm}
\noindent {\bf (v)} In Section \ref{SphericalTDualityOver7dSpacetime} we observe that the mechanism of spherical T-duality immediately passes to various Kaluza-Klein (KK) compactifications of 11d superspacetime, notably it passes to minimal 7d superspacetime with its decomposition of the C-field due to \cite{ADR16}.

\vspace{1mm}
\noindent {\bf (vi).}
Finally we prove in Section \ref{ParitySymmetry} that, different from but akin to spherical T-duality, also the parity symmetry of the 11d supergravity action functional
lifts to isomorphism on M5-brane-charge twisted cohomology on the exceptional super tangent spacetime. Since they thus act on the same
spaces of brane charges, we may think of spherical T-duality and parity symmetry to jointly constitute a new system of M-theoretic duality relations.

\subsection{Ordinary T-Duality of super F1/D$p$-Branes on 10d type II superspacetime}
\label{OrdinaryTypeIITDuality}

Ordinary \emph{T-duality} is a fundamental symmetry in string theory (see \cite{GPR}\cite{AAL95} for standard review),
important both for the inner structure of the theory as well as for its phenomenology.
Previously the formalization of ``topological T-duality'' -- which is meant to be the restriction of T-duality
to the brane charges, disregarding metric information --
had been inspired by, but not directly derived from string theory:
for circle bundles in \cite{BouwknegtEvslinMathai04, BunkeSchick05}
and, more generally for torus bundles in \cite{BHM03, BunkeRumpfSchick06}.

\begin{itemize}
\item  For $n =0$ our definition of higher T-duality correspondences in
Def. \ref{HigherTDualityCorrespondences} structurally reduces to this formulation of
ordinary topological T-duality (see Remark \ref{TDualityAxiom})
and our Corollary \ref{HigherTopologicalTDualityInTwistedCohomology}
reduces to the corresponding isomorphism in degree-3 twisted cohomology, rationally.

\item   A key difference is that even for $n =0$ our formulation captures also the super-geometry,
where the super-WZW terms of the super $p$-branes take values. This had allowed us in \cite{FSS16}
to show that the super WZW-terms of the type II F1/D$p$ branes constitute the archetypical example
of ``super-topological T-duality'' (\cite[Thm. 5.3, Rem. 5.4]{FSS16}):
\end{itemize}

First, the T-duality correspondence from Def. \ref{HigherTDualityCorrespondences}
in this case is that of \cite[Prop. 6.2]{FSS16}:
\begin{equation}
  \label{CorrespondenceTDualityTypeII}
  \xymatrix{
    && & \mathbb{R}^{8 + (1,1),1\vert 32}
      \ar[dl]_{p_A}
      \ar[dr]^{p_B}
    \\
    b^2 \mathfrak{u}_1
    &&
    \mathbb{R}^{9,1\vert \mathbf{16} + \overline{\mathbf{16}}}
    \ar[dr]_{\mathrm{hofib}( c_2^{\mathrm{IIA}} )}
    \ar[ll]_-{\mu_{{}_{F1}}^{\mathrm{IIA}}}
    &&
    \mathbb{R}^{9,1\vert \mathbf{16} + {\mathbf{16}}}
    \ar[dl]^{\mathrm{hofib}( c_2^{\mathrm{IIB}} )}
    \ar[rr]^-{\mu_{{}_{F1}}^{\mathrm{IIB}} }
    &&
    b \mathfrak{u}_1
    \\
    &
    &
    &
    \mathbb{R}^{8,1\vert \mathbf{16} + \mathbf{16}}
    \ar[dl]|{ c_2^{\mathrm{IIB}} }
    \ar[dr]|{ c_2^{\mathrm{IIA}} }
    \\
    &&
    b \mathfrak{u}_1 && b \mathfrak{u}_1
  }
\end{equation}
Here $\mathbb{R}^{9,1\vert \mathbf{16} + \overline{\mathbf{16}}}$ and $\mathbb{R}^{9,1\vert \mathbf{16} + {\mathbf{16}}}$ denote the 10d Type IIA/B
super-Minkowski spacetimes (Example \ref{superMinkowskiSuperLieAlgebra}) carrying the 3-cocycles
\begin{equation}
  \label{SuperstringCocycleTypeII}
  \begin{aligned}
    \mu_{{}_{F1}}^{{}^{\mathrm{IIA/B}}}
    & =
    i \overline{\psi} \Gamma^{{}^{\mathrm{IIA/B}}}_a \Gamma_{10} \psi \wedge e^a_{{}_{A/B}}
    \\
    & = \mu_{{}_{F1}}\vert_{8+1} \,+\, e^9_{{}_{A/B}} \wedge c_2^{{}^{\mathrm{IIB/A}}}
  \end{aligned}
\end{equation}
corresponding to the type IIA/B superstring super-WZW terms, respectively, i.e. the type II version of (\ref{BSTString}).
Both the super-spacetime are super Lie algebra extensions (via Prop. \ref{hofibofSuperLInfinityAlgebras}, hence rational circle fibrations) over the 9d type II super-Minkowski spacetime
$\mathbb{R}^{8,1\vert \mathbf{16} + \mathbf{15}}$ (\cite[Prop. 2.14]{FSS16}), and their (homotopy-)fiber product as such is the ``doubled'' superspacetime
$\mathbb{R}^{8+(1,1),1\vert 32}$ (\cite[Section 6]{FSS16}).

\medskip
Moreover, the induced isomorphism on twisted cohomology from Corollary \ref{HigherTopologicalTDualityInTwistedCohomology} via Theorem \ref{HigherTDuality}, in this case is that of \cite[Prop. 6.4]{FSS16}
\begin{equation}
  \label{TypeIIT}
  \xymatrix@R=4pt{
    H^{0 + \mu_{{}_{F1}}^{\mathrm{IIA}}}\big( \mathbb{R}^{9,1 \vert \mathbf{16} + \overline{\mathbf{16}}}/\mathrm{Spin}(9,1) \big)
    \ar@{<-}[rrr]_-{\simeq}^{  (p_A)_\ast \circ  e^{\mathcal{P}} \circ p_B^\ast }
    &&&
    H^{1 + \mu_{{}_{F1}}^{\mathrm{IIB}}}\big( \mathbb{R}^{9,1 \vert \mathbf{16} + {\mathbf{16}}}/\mathrm{Spin}(9,1) \big)
    \\
    \left(
      {\begin{array}{c}
        \mu_{{}_{D0}}, \;\mu_{{}_{D2}}
        \\
        \mu_{{}_{D4}}, \;\mu_{{}_{D6}}
        \\
        \mu_{{}_{D8}}, \;\mu_{{}_{D10}}
      \end{array}}
  \right)
    \ar@{<-|}[rrr]
    &&&
     \left(
      {\begin{array}{c}
        \mu_{{}_{D1}}, \; \mu_{{}_{D3}}
        \\
        \mu_{{}_{D5}}, \; \mu_{{}_{D7}}
        \\
        \mu_{{}_{D9}}
      \end{array}}
    \right)
  }
\end{equation}
which takes the cocycles in even/odd-degree $\mathrm{Spin}(9,1)$-invariant $\mu_{F1}^{IIA/B}$-twisted cohomology (Def. \ref{TwistedCohomology})
corresponding to the super-WZW terms for the super D$p$-branes (\cite[Section 4]{FSS16}) into each other, as indicated.
Structurally, by pull-tensor-push through the doubled super-spacetime, this
is \emph{Hori's formula for the Buscher rules for RR-fields} \cite[(1.1)]{Hori99}
refined to the superspace components of the RR-fields. Both generalize and globalize the original
rules \cite{Bu1}\cite{Bu2}.

\medskip
Since in following we will be lifting this ordinary T-duality to exceptional spaces and involving higher gerbes,
We comment
on related literature, which in a certain sense this generalizes:
Global topological structures in T-duality
had also been considered in \cite{AABL, Hu}. A geometric description of T-duality may be given by
identifying the cotangent bundles of the original and the dual manifold, exhibiting the duality
as a symplectomorphism of the string phase spaces \cite{KS}\cite{AL}. This
has been extended to cotangent bundles of the corresponding loop spaces \cite{BHM}.
T-duality is also described in generalized geometry \cite{Per}\cite{Linds}\cite{CG}, in
non-geometry \cite{GMPW}, in doubled geometry \cite{Hu2}.
Topological T-duality is extended to include automorphisms determined by the twists,
which can be viewed as a topological approximation to a gerby gauge transformation of spacetime \cite{Pan}.
The effect of the gauging on the B-field terms in the sigma model lead to restrictions
on the corresponding curvature \cite{HS1, HS2}.
Treating the 2-form $B$ as a gerbe connection captures the gauging obstructions
and the global constraints on the T-duality \cite{BHM}.

\subsection{Self T-duality on 6d super-spacetime}
\label{SelfTDualityIn6d}

As an other example of an ordinary (i.e. not higher) super-topological T-duality
we observe that there is an example of a cyclic topological T-duality for superstrings on
a D0-brane extension of 5d super-Minkowski spacetime
(Prop. \ref{TDualityCorrespondenceOver5dSuperMinkowskiSpacetime} below).
To put this in perspective, we first recall some background.

\medskip
$D = 5$ simple supergravity can be obtained directly as a Calabi-Yau compactification
of D = 11 supergravity \cite{CCDF}\cite{FKM}\cite{FMS} on a  threefold $Y$ with
Hodge number $h_{1,1} = 1$, together with the truncation of scalar
multiplets.  Therefore, the two objects in $D=5$, namely the string and the dual
D$0$-brane, have 11-dimensional origins. The first may be viewed as an
M5-brane wrapped around the unique 4-cycle of $Y$, while the latter
 is an M2-brane wrapped around the unique 2-cycle dual to the 4-cycle.
The fact that we are getting T-duality for the lower cocycles is perhaps then
not surprising and can be naturally explained by the above direct relation
between the branes and by uniqueness of the cycles on which they wrap.

\medskip
Note that this theory resembles $D=11$ supergravity in many respects. This is illustrated by
using extended symmetries in \cite{MO} and higher gauge symmetries in \cite{tcu}.
Hence this simpler model might give an insight into the unsolved interesting problems of
M-theory.

\medskip
In 5d there is a direct analog of what in 11d is the $S^4$-valued supercocycle of M2/M5-branes (\ref{M2M5CocyclesOn11dSuperMinkowski})
from example \ref{HigherTDualCorrespondenceForM2Brane}:
\begin{prop}[Two-sphere valued supercocycle in 5d]
  \label{2SphereValuedSupercocycleIn5d}
  In $\mathrm{CE}\left( \mathbb{R}^{4,1\vert \mathbf{8}+\mathbf{8}} \right) $
  consider the cochains
  $$
    \mu_{{}_{D0}}^{5d}
    :=
    \overline{\psi}_A \psi_A
   \qquad
   \text{and}
   \qquad
    \mu_{\mathrm{string}}^{5d}
    :=
    \overline{\psi}_A \Gamma_a \psi_A \wedge e^a
    \,.
  $$
  \item {\bf (i)} Then we have the relations
  $$
    d \mu_{{}_{D0}}^{5d} = 0
    \qquad
    \text{and}
    \qquad
    d \mu_{\mathrm{string}}^{5d} = \mu_{{}_{D0}}^{5d} \wedge \mu_{{}_{D0}}^{5d}
    \,.
  $$
  \item {\bf (ii)} Hence jointly these constitute
  a rational 2-sphere valued coycle.
  $$
    \xymatrix{
      \mathbb{R}^{4,1\vert \mathbf{8} + \mathbf{8} }
        \ar[rrr]^-{ (\mu_{{}_{D0}}^{5d}, \; \; \mu_{\mathrm{string}}^{5d}) }
        &&&
      \mathfrak{l}(S^2)
    }
  $$
\end{prop}
\begin{proof}
  Due to the relation $d e^a = \overline{\psi}_A \Gamma^a \psi_A$, the
  condition to be shown is equivalently
  $$
    \overline{\psi}_A \Gamma_a \psi_A
    \wedge
    \overline{\psi}_B \Gamma^a \psi_B
    =
    \overline{\psi}_A \psi_A \wedge \overline{\psi}_B \psi_B
    \,.
  $$
  But this is precisely the Fierz identity that holds from \cite[(III.5.50.a)]{CDF}.
\end{proof}

Consequently, we can discuss the appropriate T-duality in this setting.

\begin{prop}[Cyclic T-duality on D0-brane extension of 5d superspacetime]
  \label{TDualityCorrespondenceOver5dSuperMinkowskiSpacetime}
  There is a T-duality self-correspondence (Def. \ref{HigherTDualityCorrespondences}) on the D0-brane extension of 5d super-Minkowski spacetime
  of the form
$$
  \xymatrix@R=1.7em{
    b^{2}\mathfrak{u}_1
    &
    \widehat{\mathbb{R}}^{4,1\vert \mathbf{8} + \mathbf{8}}
    \ar[drr]_-{\hspace{-2mm}\mathrm{hofib}( \overline{\psi}_A \psi_A) }
    \ar[l]_-{ h_3  }
    &&&&
    \widehat{\mathbb{R}}^{4,1\vert \mathbf{8} + \mathbf{8}}
    \ar[dll]^-{\;\;\mathrm{hofib}( \overline{\psi}_A \psi_A) }
    \ar[r]^-{ h_3 }
    &
    b^{2}\mathfrak{u}_1
    \\
    &
    &&
    \mathbb{R}^{4,1 \vert \mathbf{8} + \mathbf{8} }
    \ar[dll]^-{ \overline{\psi}_A \psi_A }
    \ar[drr]_-{ \overline{\psi}_A \psi_A }
    \\
    &
    b \mathfrak{u}_1
    &&&&
    b \mathfrak{u}_1\;.
  }
$$
\end{prop}
\begin{proof}
  By Prop. \ref{2SphereValuedSupercocycleIn5d} used in Prop. \ref{ClassifyingDataHigherTDuality}
  the result is established.
\end{proof}

\subsection{Spherical T-duality of M5-branes on M2-extended M-theory spacetime}
\label{SphericalTDualityOfM5BranesOnM2ExtendedSpacetime}

We discuss spherical topological T-duality in the sense of Section \ref{sphericaltduality} realized for M5-branes
on the M2-brane extended 11d super-spacetime.
We first recall that the M2/M5-brane geometry is governed by the following data.

\begin{example}[M2/M5-Brane cocycle]
\label{HigherTDualCorrespondenceForM2Brane}
Consider the
the super-Minkowski spacetime
$\mathbb{R}^{10,1\vert \mathbf{32}}$ (Example \ref{superMinkowskiSuperLieAlgebra}) underlying 11-dimensional supergravity. By \cite[(3.27a)]{DF} \cite[Prop. 4.3]{FSS16}
the elements
\begin{eqnarray}
  \label{M2BraneCocycle}
  \mu_{{}_{M2}} &:=& \tfrac{i}{2}\overline{\psi} \Gamma_{a_1 a_2} \psi \wedge e^{a_1} \wedge e^{a_2}
\\
  \mu_{{}_{M5}} &:=& \tfrac{1}{5!}\overline{\psi} \Gamma_{a_1 \cdots a_5} \psi \wedge e^{a_1} \wedge \cdots \wedge e^{a_5}
 \nonumber
\end{eqnarray}
in $\mathrm{CE}(\mathbb{R}^{10,1\vert \mathbf{32}})$
satisfy the relations
\begin{equation}
  \label{M2M5CocyclesOn11dSuperMinkowski}
  d \mu_{{}_{M2}} = 0
  \qquad
  \text{and}
  \qquad
  d \mu_{{}_{M5}} = -\tfrac{1}{2} \mu_{{}_{M2}} \wedge \mu_{{}_{M2}}
  \,.
\end{equation}
By \cite{FSS13, FSS15} $\mu_{{}_{M2}}$ defines the WZW-term for the Green-Schwarz
sigma-model of the M2-brane, and the following combination defines the WZW term for
the GS sigma-model of the M5-brane:
\begin{equation}
  \label{M5cocycleInTDualityPair}
  \widetilde{\mu}_{{}_{M5}}
    :=
  2 \mu_{{}_{M5}} + c_3 \wedge \mu_{{}_{M2}}
  \;\in\;
  \mathrm{CE}(\mathfrak{m}2\mathfrak{brane})
  \,,
\end{equation}
where
$$
  \mathfrak{m}2\mathfrak{brane}
  :=
  \mathrm{hofib}(\mu_{{}_{M2}})
$$
is the M2-brane extension of 11d super-Minkowski spacetime (\cite[Section 4.1]{FSS13}), which by Prop. \ref{hofibofSuperLInfinityAlgebras}
is given by
$$
  \mathrm{CE}(\mathfrak{m}2\mathfrak{brane})
  \;=\;
  \mathrm{CE}(\mathbb{R}^{10,1\vert \mathbf{32}})[ c_3 ]/
  \left(
    d c_3 = \mu_{{}_{M2}}
  \right)
  \,.
$$
\end{example}

\medskip
We now indeed uncover that T-self-duality associated with the M-brane cocycle.

\begin{prop}[M5-brane cocycle is spherical T-dual to itself]
 \label{M5CocycleIsSPhericalTDualToItself}
The M5-brane cocycle (\ref{M5cocycleInTDualityPair}) on the M2-brane extension of 11-dimensional super-Minkowski spacetime (Example \ref{HigherTDualCorrespondenceForM2Brane}) is
spherical T-dual (Def. \ref{HigherTDualityCorrespondences}) to itself, as exhibited in the following diagram
\begin{equation}
  \label{CorrespondenceSphericalTDualityM5}
  \xymatrix@C=-10pt{
    &
    {\phantom{AAAAAAA}}
    &&
    \mathfrak{m}2\mathfrak{brane} \times_{\mathbb{R}^{10,1\vert \mathbf{32}}} \mathfrak{m}2\mathfrak{brane}
    ~\ar[dr]^{p_B}
    ~\ar[dl]_{p_A}
    &&
    {\phantom{AAAAAAA}}
    \\
    b^{6} \mathfrak{u}_1
    &&
    \mathfrak{m}2\mathfrak{brane}
    \ar[dr]_{\pi^A := \mathrm{hofib}(\mu_{{}_{M2}}) \phantom{AAAA}}
    \ar[ll]_-{\widetilde{\mu}_{{}_{M5}}}
    &&
    \mathfrak{m}2\mathfrak{brane}
    \ar[dl]^{\phantom{AAA}\pi^B := \mathrm{hofib}(\mu_{{}_{M2}})}
    \ar[rr]^-{\widetilde{\mu}_{{}_{M5}}}
    &&
    b^{6} \mathfrak{u}_1
    \\
    &&
    & \mathbb{R}^{10,1\vert \mathbf{32}}
    \ar[dl]^{\mu_{{}_{M2}}}
    \ar[dr]_{\mu_{{}_{M2}}}
    \\
    &&
    b^{3} \mathfrak{u}_1 && b^{3} \mathfrak{u}_1
  }
\end{equation}
$$
\xymatrix@C=6em{
  (\mathfrak{m}2\mathfrak{brane},\widetilde{\mu}_{{}_{M5}} )
  \;\;\ar@{<->}[rr]^{\cal T} &&
  \;\; (\mathfrak{m}2\mathfrak{brane},\widetilde{\mu}_{{}_{M5}} )\;.
  }
$$
\end{prop}
\begin{proof}
Establishing the diagram amounts to checking that the cocycles
are compatible.
With the equivalent reformulation from Prop. \ref{ClassifyingDataHigherTDuality}
we need to show that
$$
  d \mu_{{}_{M2}} = 0
$$
and
$$
  \begin{aligned}
    d (2 \mu_{{}_{M5}})
    &=
    - \mu_{{}_{M2}} \wedge \mu_{{}_{M2}}
    \,.
  \end{aligned}
$$
But this is precisely the M2/M5-cocycle presented in equation (\ref{M2M5CocyclesOn11dSuperMinkowski}).
\end{proof}

\subsection{Toroidal T-duality on exceptional M-theory spacetime}
\label{ToroidalTDualityOnExceptionalMTheorySpacetime}


\begin{defn}[Maximal central extension of the $N=32$ superpoint]
\label{MaximalCentralExtensionOfNIs32Superpoint}
Consider the $N=32$ superpoint, hence the super $L_\infty$-algebra
$$
  \mathbb{R}^{0\vert 32}
  \;\in\;
  sL_\infty \mathrm{Alg}^{\mathrm{fin}}_{\mathbb{R}}
$$
defined dually by (cf. expression \eqref{SupermMinkowskiCE})
$$
  \mathrm{CE}( \mathbb{R}^{0\vert 32} )
  \;=\;
  \mathbb{R}[
    \hspace{-4mm}
      \underset{
        \mathrm{deg} = (1, \mathrm{odd)}
      }{\;
      \underbrace{(
        \psi^\alpha)
      }
      }
    \hspace{-3mm}_{\alpha \in \{1, 2, \cdots, 32\}}
  ]/( d \psi^\alpha = 0 )
  \,.
$$
We write
$$
  \xymatrix{
    \mathbb{R}^{10,1\vert \mathbf{32}}_{\mathrm{exc}}
    \ar[d]_-{ \star }
    \\
    \mathbb{R}^{0\vert \mathbf{32}}
  }
$$
for its maximal invariant central extension according to \cite[Def. 6]{HuertaSchreiber}.
We will call this the \emph{exceptional tangent superspacetime}.
\end{defn}
\begin{prop}[Spin-action on exceptional tangent superspacetime]
\label{SpinActionOnExceptionalTangentSuperspacetime}
{\bf (i)} The exceptional tangent superspacetime super Lie algebra $\mathbb{R}^{10,1\vert \mathbf{32}}_{\mathrm{exc}}$
from Def. \ref{MaximalCentralExtensionOfNIs32Superpoint} is, up to isomorphism, dually given by the
following Chevalley-Eilenberg algebra:
\begin{equation}
  \label{CEAlgebraOfExceptionalTangentSuperspacetime}
  \mathrm{CE}( \mathbb{R}^{10,1\vert \mathbf{32}}_{\mathrm{exc}} )
  \;=\;
  \mathbb{R}\big[
    \hspace{-3mm}\underset{ \mathrm{deg} = (1,\mathrm{even}) }{\underbrace{(e^a)}}
    \hspace{-3mm}, \;\;
    \underset{(1,\mathrm{even}) }{ \underbrace{ (B_{a_1 a_2}) }} \;,\;
    \underset{(1,\mathrm{even}) }{ \underbrace{ (B_{a_1 \cdots a_5}) }}\;,
    \underset{ (1,\mathrm{odd}) }{ \underbrace{ (\psi^\alpha) } }
     \big]/
  \left(
    \begin{array}{l}
      d \psi^\alpha = 0, \;\; \; \;\;\; \;\; d B_{a_1 a_2} = \tfrac{i}{2}\overline{\psi} \Gamma_{a_1 a_2} \psi,
      \\
      d e^a = \overline{\psi}\Gamma^a \psi, \; \;
       d B_{a_1 \cdots a_5} = \tfrac{1}{5!}\overline{\psi} \Gamma_{a_1 \cdots a_5} \psi
    \end{array}
  \right)
  \,,
\end{equation}
where $\Gamma$ denotes a Clifford algebra representation on the 32-dimensional vector
space spanned by the $\psi^\alpha$, corresponding to the irreducible real representation $\mathbf{32}$ of $\mathrm{Spin}(10,1)$.

\item {\bf (ii)} Hence the bosonic body of the exceptional super tangent spacetime is
\begin{equation}
  \label{ExceptionalTangentSpacetimeBosonicPart}
  \big(
    \mathbb{R}^{10,1\vert \mathbf{32}}_{\mathrm{exc}}
  \big)_{\mathrm{bos}}
  \simeq
  \mathbb{R}^{10,1}
  \oplus
  \Exterior^2\left( \mathbb{R}^{10,1} \right)^\ast
  \oplus
  \Exterior^5\left( \mathbb{R}^{10,1} \right)^\ast\;.
\end{equation}

\item {\bf (iii)} This induces on $\mathbb{R}^{10,1\vert \mathbf{32}}_{\mathrm{exc}}$ a $\mathrm{Spin}(10,1)$-action in the sense of
Def. \ref{SuperLInfinityAlgebraWithGroupAction}, and the projection to the superpoint factors $\mathrm{Spin}(10,1)$-equivariantly over
the 11d super-Minkowski spacetime (Example \ref{superMinkowskiSuperLieAlgebra}) as
$$
  \xymatrix{
    \mathbb{R}^{10,1\vert \mathbf{32}}_{\mathrm{exc}}
    \ar[dr]^{ \mathrm{hofib}( \mu_{\mathrm{exc}} ) }
    \ar[dd]_{ \mathrm{hofib}(\psi \wedge \overline{\psi}) }
    \\
    & \mathbb{R}^{10,1\vert \mathbf{32}}\;.
      \ar[dl]^{\; \mathrm{hofib}(\overline{\psi} \Gamma_a \psi) }
    \\
    \mathbb{R}^{0\vert \mathbf{32}}
  }
$$
\end{prop}
\begin{proof}
The maximal invariant central extension, in the sense of \cite{HuertaSchreiber}, of $\mathbb{R}^{0\vert \mathbf{32}}$
is classified by the maximal tuple of linearly independent 2-cocycles. Since the $\psi^\alpha$ commute with each other,
there is one 2-cocycle for every symmetric $32 \times 32$-matrix. These may be identified with the
elements $\Gamma_a$, $\Gamma_{a_1 a_2}$, $\Gamma_{a_1 \cdots a_5}$ of a Clifford algebra representation of the
$\mathbf{32}$ of $\mathrm{Spin}(10,1)$.

Under this identification and with notation from Def. \ref{HigherCupProducts} the cocycle classifying the maximal central extension is,
up to isomorphism, given by
\begin{equation}
  \label{528TorusCocycle}
  \Big(
    (\overline{\psi}\Gamma_a \psi), \;
    \underset{
      \mu_{\mathrm{exc}}
    }{
      \underbrace{
        \big(
          ( \tfrac{i}{2} \overline{\psi} \Gamma_{a_1 a_2} \psi ),
          ( \tfrac{1}{5!} \overline{\psi} \Gamma_{a_1 \cdots a_5} \psi )
        \big)
      }
    }
  \Big)
  \;:\;
  \mathbb{R}^{0\vert \mathbf{32}}
    \longrightarrow
  b (\mathfrak{u}_1)^{528}\;.
\end{equation}
This implies the claim by  Prop. \ref{hofibofSuperLInfinityAlgebras}.
\end{proof}

\begin{remark}[The 11d exceptional super tangent spacetime]
  \label{ExceptionalSuperTangentSpacetime}
  That ``exceptional generalized tangent spaces'' of the form
  $\mathbb{R}^{d-1,1} \oplus \Exterior^2(\mathbb{R}^{d-1,1})^\ast \oplus \Exterior^5(\mathbb{R}^{d-1,1})^\ast$
  should play a role in M-theory compactified on a $d$-dimensional fiber was proposed in
  \cite{Hull07, PachecoWaldram08}. The full compactification with $d = 11$
  $$
    \mathbb{R}^{10,1} \oplus \Exterior^2(\mathbb{R}^{10,1})^\ast \oplus \Exterior^5(\mathbb{R}^{10,1})^\ast
  $$
  that appears in (\ref{ExceptionalTangentSpacetimeBosonicPart})
  has been the basis of various conjectures in \cite{Wes03}, see also \cite[Sec. 2.2]{Bar12}.
  That this 11d exceptional tangent bundle (\ref{ExceptionalTangentSpacetimeBosonicPart})
  happens to be the bosonic body of the D'Auria-Fr{\'e} (DF) algebra
  $\mathbb{R}^{10,1\vert \mathbf{32}}_{\mathrm{exc},s}$ (Def. \ref{FermionicExtensionOfExceptionalTangentSuperspacetime}, already due to \cite{DF})
  has been highlighted in \cite{Vaula07}.
However, it seems that no relation between the DF-algebra $\mathbb{R}^{10,1\vert \mathbf{32}}_{\mathrm{exc},s}$ and the  actual idea of M-theoretic generalized geometry promoted in \cite{Hull07, PachecoWaldram08, Bar12} has been observed before.
 This is indeed what we discuss below in Section \ref{ExceptionalGenralisedGeometry}.
\end{remark}

\begin{defn}[528-Bein on exceptional super tangent spacetime]
\label{528BeinOnExceptionalTangentSuperspacetime}
It is useful to abbreviate the CE-generators in expression \eqref{CEAlgebraOfExceptionalTangentSuperspacetime} as
$$
  \begin{aligned}
    \vec{\mathcal{E}}
    & := (\mathcal{E}_A)
    \\
    & := \big( (e_a), (B_{a_1 a_2}), (B_{a_1 \cdots a_5}) \big)\;.
  \end{aligned}
$$
This $(\mathcal{E}_A)$ may be thought of as the 528-bein on the exceptional tangent superspacetime.
\end{defn}

The following basic example will play a key role in the proof of Prop. \ref{HCohomologyOfDecomposedCFieldVanishesInSubalgebra} below:

\begin{example}
The \emph{exceptional volume form} on the exceptional tangent superspacetime is
  \begin{equation}
    \label{TheExceptionalVolumeForm}
    \begin{aligned}
    \mathrm{vol}_{528}
    & :=
    \underoverset{A = 1}{528}{\Exterior} \mathcal{E}_A
    \\
    & :=
    \Big(
      \underset{0 \leq a \leq 10}{\bigwedge} e^a
    \Big)
    \wedge
    \Big(
      \underset{0 \leq a_1, \lt a_2 \leq 10}{\bigwedge} B^{a_1 a_2}
    \Big)
    \wedge
    \Big(
      \underset{0 \leq a_1, \lt \cdots \lt a_5 \leq 10}{\bigwedge} B^{a_1 \cdots a_5}
    \Big)\;.
    \end{aligned}
  \end{equation}
\end{example}

\begin{defn}[Inner product on exceptional tangent superspacetime]
  \label{UniversalCupProductForTExc}
  On $b (\mathfrak{u}_1)^{528}$ consider, for $s \in \mathbb{R} \setminus \{0\}$,
   the following universal cup product (Def. \ref{HigherCupProducts})
  $$
    \begin{aligned}
    \big\langle
      \mathrm{pr}_1^\ast(\vec{\mathrm{dd}}),\;
      \mathrm{pr}_2^\ast(\vec{\mathrm{dd}})
    \big\rangle_{\mathrm{exc},s}
    & =
    (s+1)
    \mathrm{pr}_1^\ast(\mathrm{dd}^a)
    \wedge
    \mathrm{pr}_2^\ast(\mathrm{dd}_a)
    -
    \mathrm{pr}_1^\ast(\mathrm{dd}^{a_1 a_2})
    \wedge
    \mathrm{pr}_2^\ast(\mathrm{dd}_{a_1 a_2})
    \\
    & \phantom{=}
    \;\;\;+
    (1 + \tfrac{s}{6})
    \mathrm{pr}_1^\ast(\mathrm{dd}^{a_1 \cdots a_5})
    \wedge
    \mathrm{pr}_2^\ast(\mathrm{dd}_{a_1 \cdots a_5})\;.
    \end{aligned}
  $$
\end{defn}

\medskip
\begin{prop}[528-Toroidal T-duality correspondence on exceptional tangent super-spacetime]
  \label{ExceptionalToroidalTDuality}
  {\bf (i)} On the exceptional tangent superspacetime $\mathbb{R}^{10,1\vert\mathbf{32}}_{\mathrm{exc}}$ (Def. \ref{MaximalCentralExtensionOfNIs32Superpoint})
  the following CE-element is a $\mathrm{Spin}(10,1)$-invariant 3-cocycle (Example \ref{InvariantCocycle}) for each $s \in \mathbb{R} \setminus \{0\}$:
  \begin{equation}
    \label{StringCocycleOnExceptionalSpacetime}
    \hspace{-5mm}
    \begin{aligned}
      \big\langle
        \vec{\mathcal{E}}  \wedge \overline{\psi} \vec \Gamma \psi
      \big\rangle_{\mathrm{exc},s}
     & =
     (s+1)
     e_a \wedge \overline{\psi} \Gamma^a \psi
     -
     B_{a_1 a_2} \wedge \overline{\psi} \Gamma^{a_1 a_2} \psi
     +
     (1 + \tfrac{s}{6})
     B_{a_1 \cdots a_5} \wedge \overline{\psi} \Gamma^{a_1 \cdots a_5} \psi
         \\
    & \phantom{=}
     \;\in\;
     \mathrm{CE}\big(
        \mathbb{R}^{10,1\vert \mathbf{32}}_{\mathrm{exc}}
     \big)^{\mathrm{Spin}(10,1)}
     \,,
    \end{aligned}
  \end{equation}
  where on the left we used the 528-bein from Def. \ref{528BeinOnExceptionalTangentSuperspacetime} and the
  inner product from Def. \ref{UniversalCupProductForTExc}.

  \item {\bf (ii)} The cocycle (\ref{StringCocycleOnExceptionalSpacetime}) is in 528-toroidal T-duality correspondence with itself (Def. \ref{HigherTDualityCorrespondences})
  with respect to the exceptional tangent superspacetime regarded as a 528-torus fibration over the superpoint via Def. \ref{UniversalCupProductForTExc}:
  $$
    \xymatrix@C=3em{
      b^2 \mathfrak{u}_1
      &
      &
      \mathbb{R}^{10,1 \vert \mathbf{32}}_{\mathrm{exc}}
      \ar[dr]_{ \mathrm{hofib}(\psi \wedge \overline{\psi}) }
      \ar[ll]_{ \left\langle \vec{\mathcal{E}} \wedge \overline{\psi} \vec \Gamma \psi \right\rangle_{\mathrm{exc},s} }
      &&
      \mathbb{R}^{10,1 \vert \mathbf{32}}_{\mathrm{exc}}
      \ar[dl]^{ \mathrm{hofib}(\psi \wedge \overline{\psi}) }
      \ar[rr]^-{ \left\langle \vec{\mathcal{E}} \wedge \overline{\psi} \vec \Gamma \psi \right\rangle_{\mathrm{exc},s} }
      &&
      b^2 \mathfrak{u}_1
      \\
      &
      &
      &
      \mathbb{R}^{0 \vert 32}
    }
  $$
  \item {\bf (iii)}
  The cocycle \ref{StringCocycleOnExceptionalSpacetime} is also in 517-toroidal T-duality
  with itself (Def. \ref{HigherTDualityCorrespondences}) factoring via Prop.\ref{SpinActionOnExceptionalTangentSuperspacetime}
  over the 11d super-Minkowski spacetime $\mathbb{R}^{10,1\vert \mathbf{32}}$
  $$
    \xymatrix@C=3em{
      b^2 \mathfrak{u}_1
      &
      &
      \mathbb{R}^{10,1 \vert \mathbf{32}}_{\mathrm{exc}}
      \ar[dr]_{  }
      \ar[ll]_{ \left\langle \vec{\mathcal{E}} \wedge \overline{\psi} \vec \Gamma \psi \right\rangle_{\mathrm{exc},s} }
      &&
      \mathbb{R}^{10,1 \vert \mathbf{32}}_{\mathrm{exc}}
      \ar[dl]^{  }
      \ar[rr]^-{ \left\langle \vec{\mathcal{E}} \wedge \overline{\psi} \vec \Gamma \psi \right\rangle_{\mathrm{exc},s} }
      &&
      b^2 \mathfrak{u}_1
      \\
      &
      &
      &
      \mathbb{R}^{10,1\vert\mathbf{32}}
    }
  $$
  $$
  \xymatrix{
    \big(
      \mathbb{R}^{10,1\vert \mathbf{32}}_{\mathrm{exc}}
      \,,\,
      \big\langle
        \vec{\mathcal{E}}  \wedge \overline{\psi} \vec \Gamma \psi
      \big\rangle_{\mathrm{exc},s}
    \big)
    \ar@{<->}[rr]^-{{\cal T}_{\rm self}} &&
    \big(
      \mathbb{R}^{10,1\vert \mathbf{32}}_{\mathrm{exc}}
      \,,\,
      \big\langle
        \vec{\mathcal{E}}  \wedge \overline{\psi} \vec \Gamma \psi
      \big\rangle_{\mathrm{exc},s}
    \big)\;.
    }
  $$
\end{prop}
\begin{proof}
  The cocycle property to be checked requires that
  \begin{equation}
    \label{CocycleConditionOnExceptional}
    \begin{aligned}
      d \big\langle \vec{\mathcal{E}} \wedge \overline{\psi} \vec{\Gamma} \psi \big\rangle_{\mathrm{exc},s}
       = &
      \big\langle d \vec{\mathcal{E}} \wedge \overline{\psi} \vec{\Gamma} \psi \big\rangle_{\mathrm{exc},s}
      \\
       = &
      (s + 1) \left( \overline{\psi} \Gamma_a \psi \right) \wedge \left( \overline{\psi} \Gamma^a \psi \right)
      -
      \tfrac{i}{2}
      \left( \overline{\psi} \Gamma_{a_1 a_2} \psi \right) \wedge \left( \overline{\psi} \Gamma^{a_1 a_2} \psi \right)
      \\
      & \;\;+
      (1+ \tfrac{s}{6})
      \tfrac{1}{5!}
      \left( \overline{\psi} \Gamma_{a_1 \cdots a_5} \psi \right) \wedge \left( \overline{\psi} \Gamma^{a_1 \cdots a_5} \psi \right)
      \\
      = & 0\;.
    \end{aligned}
  \end{equation}
  This does indeed vanish by a trilinear Fierz identity for $\mathbf{32}$ that was originally observed in \cite[(6.4)]{DF};
  our parameterization by $s$ follows \cite[(23)]{BAIPV04}:
  \begin{equation}
    \label{528FierzIdentity}
    (s + 1) \left( \overline{\psi} \Gamma_a \psi \right) \wedge \overline{\psi} \Gamma^a
    -
    \tfrac{i}{2}
    \left( \overline{\psi} \Gamma_{a_1 a_2} \psi \right) \wedge \overline{\psi} \Gamma^{a_1 a_2}
    +
    (1+ \tfrac{s}{6})
    \tfrac{1}{5!}
    \left( \overline{\psi} \Gamma_{a_1 \cdots a_5} \psi \right) \wedge \overline{\psi} \Gamma^{a_1 \cdots a_5}
    \;=\;
    0
    \,.
  \end{equation}
  With this it is straightforward to verify that we have a T-duality self-correspondence
  according to Def. \ref{HigherTDualityCorrespondences}.
\end{proof}
\begin{remark}[Fermionic 2-cocycles on exceptional tangent superspacetime]
 \label{Fermionic2CocyclesOnExceptionalTangentSuperspacetime}
 The Fierz identity (\ref{528FierzIdentity}) is stronger than the cocycle condition
  \eqref{CocycleConditionOnExceptional} that it implies,
 since it says that already the cubic fermion term inside the quartic term vanishes by itself.
 This means that the single bosonic 3-cocycle (\ref{StringCocycleOnExceptionalSpacetime}) is in fact a
 linear combination of the following 32 fermionic 2-cocycles:
 \begin{equation}
   \label{Fermionic2Cocycles}
   \big\langle \vec{\mathcal{E}} \wedge \overline{\psi} \vec \Gamma \big\rangle
   \;:\;
   \mathbb{R}^{10,1\vert \mathbf{32}}_{\mathrm{exc},s}
   \longrightarrow
   b (\mathbb{R}^{0\vert 1})^{32}\;.
 \end{equation}
\end{remark}

\medskip

\subsection{Spherical T-duality for M5-branes on exceptional M-theory spacetime}
\label{SphericalTDualityOnExceptionalMTheorySpacetime}

We have seen two different extensions of M-theory super-spacetime above: In Section \ref{SphericalTDualityOfM5BranesOnM2ExtendedSpacetime}
we considered the M2-brane extension by the higher degree M2-cocycle and found spherical T-duality for M5-branes on it,
while in Section \ref{ToroidalTDualityOnExceptionalMTheorySpacetime} we discussed the exceptional tangent space extension by 517 cocycles of degree 2 on which we found toroidal topological T-duality.
Here we discuss how the spherical T-duality on the M2-brane extended super-spacetime
passes to a spherical T-duality on the exceptionally extended super-spacetime.

\begin{defn}[exceptional generalised super spacetime]
  \label{FermionicExtensionOfExceptionalTangentSuperspacetime}
  For $s \in \mathbb{R} \setminus \{0\}$, write
  $$
    \mathbb{R}^{10,1\vert \mathbf{32}}_{\mathrm{exc},s}
    \;\in\;
    sL_\infty \mathrm{Alg}^{\mathrm{fin}}_{\mathbb{R}}
  $$
  for the fermionic central extension of the exceptional tangent superspacetime $\mathbb{R}^{10,1\vert \mathbf{32}}$ (Def. \ref{MaximalCentralExtensionOfNIs32Superpoint})
  which is classified by the 32 fermionic 2-cocycles
  $$
    \mu_{\mathrm{exc},s}
    :=
    \big\langle \vec{\mathcal{E}} \wedge \overline{\psi}\vec \Gamma \big\rangle
  $$
  (\ref{Fermionic2Cocycles}) from Remark \ref{Fermionic2CocyclesOnExceptionalTangentSuperspacetime}:
  $$
    \xymatrix{
      \mathbb{R}^{10,1\vert \mathbf{32}}_{\mathrm{exc},s}
      \ar[d]_{ \mathrm{hofib}( \mu_{\mathrm{exc},s} ) }
      \\
      \mathbb{R}^{10,1\vert \mathbf{32}}_{\mathrm{exc}}
      \ar[rr]^{ \mu_{\mathrm{exc},s} }
      &&
      b(\mathbb{R}^{0 \vert 1})^{32}\;.
    }
  $$
  By Prop. \ref{hofibofSuperLInfinityAlgebras} this means that we may take
  \begin{equation}
    \label{ExceptionalFermionicGenerator}
    \mathrm{CE}\big(
      \mathbb{R}^{10,1\vert \mathbf{32}}_{\mathrm{exc},s}
    \big)
    \;=\;
    \mathrm{CE}\big(
      \mathbb{R}^{10,1\vert \mathbf{32}}_{\mathrm{exc}}
    \big)
   [(\eta^\alpha)]/\big( d \overline{\eta} =  \big\langle \vec{\mathcal{E}} \wedge \overline{\psi}\vec \Gamma \big\rangle \big)\;.
  \end{equation}
\end{defn}

\begin{prop}[Transgression of M2-cocycle via decomposed C-field]
  \label{TransgressionElementForM2Cocycle}
  For $s \in \mathbb{R} \setminus \{0\}$, the fermionic extension of exceptional tangent superspacetime
  $\mathbb{R}^{10,1\vert \mathbf{32}}_{\mathrm{exc},s}$ (Def. \ref{FermionicExtensionOfExceptionalTangentSuperspacetime})
  regarded a fibered over 11d super-Minkowski spacetime
  $\mathbb{R}^{10,1\vert \mathbf{32}}$
  $$
   \xymatrix@R=1.5em{
      \Exterior^2 (\mathbb{R}^{10,1})^\ast \oplus \Exterior^5 (\mathbb{R}^{10,1})^\ast \oplus \mathbf{32}
      \; \ar@{^{(}->}[rr]^-{\iota}
      &&
      \mathbb{R}^{10,1\vert \mathbf{32}}_{\mathrm{exc},s}
     \ar@{->>}[dd]_{}^{\pi_{\mathrm{exc},s}}
      \\
    \\
    &&
    \mathbb{R}^{10,1\vert \mathbf{32}}\;,
    }
  $$
  carries a transgression element (Def. \ref{TransgressionElements})
  $c_{\mathrm{exc},s} \;\in\; \mathrm{CE}(\mathbb{R}^{10,1\vert \mathbf{32}})$
  for the M2-brane 4-cocycle (\ref{M2BraneCocycle}):
  \begin{equation}
    \label{TransgressionOfM2CocycleViaDecomposedCField}
     \xymatrix{
        c_{\mathrm{fib},s}
        &&&
        \mathrm{CE}
        \big(
          \Exterior^2 (\mathbb{R}^{10,1})^\ast
          \oplus
          \Exterior^5 (\mathbb{R}^{10,1})^\ast
          \oplus
          \mathbf{32}
        \big)^{\mathrm{Spin}(10,1)}
        \\
        c_{\mathrm{exc},s}
          \ar@{|->}[r]^{d}
          \ar@{|->}[u]_{\iota^\ast}
        & d c_{\mathrm{exc},s}
        &&
        \mathrm{CE}
          \big(
            \mathbb{R}^{10,1\vert \mathbf{32}}_{\mathrm{exc},s}
          \big)^{\mathrm{Spin}(10,1)}
        \ar[u]_{\iota^\ast}
        \\
        &
        \mu_{{}_{M2}}
        \ar@{|->}[u]_{\pi_{\mathrm{exc},s}^\ast}
        &&
        \mathrm{CE}\big(\mathbb{R}^{10,1\vert \mathbf{32}}\big)^{\mathrm{Spin}(10,1)}
        \ar[u]_{\pi^\ast}
      }
    \end{equation}
\end{prop}
\begin{proof}
  The condition
  $$
    d c_{\mathrm{exc},s} = \pi_{\mathrm{exc},s}^\ast( \mu_{{}_{M2}} )
  $$
  had been solved for two values of $s$ in \cite[Section 6]{DF},
  for the other values of $s$ in \cite[Section 3]{BAIPV04}. Explicitly, in terms of the CE-generators
  from  (\ref{CEAlgebraOfExceptionalTangentSuperspacetime}) and (\ref{ExceptionalFermionicGenerator}),
   the transgression element reads
\begin{equation}
  \label{DecomposedCFieldAsSumOfBosonicAndFermionicContribution}
  c_{\mathrm{exc},s}
  \;=\;
  (c_{\mathrm{exc},s})_{\mathrm{bos}}
  +
  (c_{\mathrm{exc},s})_{\mathrm{ferm}}
\end{equation}
with
\begin{equation}
  \label{DecomposedCFieldBosonicContribution}
\begin{aligned}
   (c_{\mathrm{exc},s})_{\mathrm{bos}}
   &  =
   \raisebox{-12pt}{
   $
   \left.
   \begin{array}{l}
     \phantom{+} \alpha_0(s)  B_{a b} \wedge e^a \wedge e^b
     +
     \alpha_1(s) B^{a_1}{}_{a_2} \wedge B^{a_2}{}_{a_3} \wedge B^{a_3}{}_{a_1}
    \\
    +
    \alpha_2(s) B_{b_1 a_1 \cdots a_4} \wedge B^{b_1}{}_{b_2} \wedge B^{b_2 a_1 \cdots a_4}
    \\
    +
    \alpha_4(s)
    \epsilon_{\alpha_1 \cdots \alpha_6 b_1 \cdots b_5} B^{a_1 a_2 a_3}{}_{c_1 c_2} \wedge B^{a_4 a_5 a_6 c_1 c_2} \wedge B^{b_1 \cdots b_5}
   \end{array}
   \right\} c_{\mathrm{fib},s}
   $
   }
   \\
   & \phantom{=} \;\;\;\;\;
   -  \alpha_3(s) \epsilon_{\alpha_1 \cdots \alpha_5 b_1 \cdots b_5 c} B^{a_1 \cdots a_5} \wedge B^{b_1 \cdots b_5} \wedge e^c
  \end{aligned}
\end{equation}
and
\begin{equation}
  \label{DecomposedCFieldFermionicContribution}
  (c_{\mathrm{exc},s})_{\mathrm{ferm}}
  =
  -
  \tfrac{1}{2}
  \overline{\eta}_\alpha \wedge \psi^{\beta}
  \wedge
  \Big(
    \beta_1(s) (\Gamma_a)^\alpha{}_\beta \; e^a
    +
    \beta_2(s) (\Gamma_{a b})^\alpha{}_\beta \; B^{a b}
    +
    \beta_3(s) (\Gamma_{a_1 \cdots a_5})^\alpha{}_{\beta} \; B^{a_1 \cdots a_5}
  \Big)\;,
\end{equation}
for analytic functions $\alpha_i, \beta_j$ of the parameter $s \in \mathbb{R} \setminus \{0\}$ with the following zeros
\begin{equation}
  \label{CoefficientZerosInDecomposedCField}
  \begin{array}{lcl}
    \alpha_0(s) \neq 0
    \\
    \alpha_1(s) = 0 &\Leftrightarrow& s = -3
    \\
    \alpha_2(s) = 0 &\Leftrightarrow& s = -6
    \\
    \alpha_3(s) = 0 &\Leftrightarrow& s = -6
    \\
    \alpha_4(s) = 0 &\Leftrightarrow& s = -6
    \\
    \beta_1(s) = 0 &\Leftrightarrow& s = -3/2
    \\
    \beta_2(s) = 0 &\Leftrightarrow& s = -3
    \\
    \beta_3(s) = 0 &\Leftrightarrow& s = -6.
  \end{array}
\end{equation}

\vspace{-6mm}
\end{proof}

\medskip
Now by Prop. \ref{TransgressionElementForM2Cocycle} we may ask whether the spherical T-duality
correspondence of the M5-brane cocycle from Example \ref{HigherTDualCorrespondenceForM2Brane}
transfers along the comparison map
$$
  \xymatrix{
    \mathbb{R}^{10,1\vert \mathbf{32}}_{\mathrm{exc},s}
    \ar[rr]^{ \mathrm{comp}_s }
    \ar[dr]_{\pi_{\mathrm{exc},s}}
    &&
    \mathfrak{m}2\mathfrak{brane}
    \ar[dl]^{\;\; \mathrm{hofib}( \mu_{{}_{M2}} ) }
    \\
    & \mathbb{R}^{10,1\vert \mathbf{32}}
  }
$$
to an isomorphism on the M5-brane twisted cohomology of the super exceptional super tangent spacetime of Def. \ref{tildemuM5TwistedCohomologyOfTangentSuperspacetime}.
By Theorem \ref{HigherTDualityForDecomposedFields} this requires analysis of the C-cohomology of the decomposed C-field:

\begin{prop}[C-cohomology of decomposed supergravity C-field]
  \label{HCohomologyOfDecomposedCFieldVanishesInSubalgebra}
  For $s \in \mathbb{R} \setminus \{ 0,-3/2, -3, -6 \}$,
  the C-cohomology (Def. \ref{CCohomology}) of the decomposed C-field
  $c_{\mathrm{exc},s} \in \mathrm{CE}( \mathbb{R}^{10,1\vert\mathbf{32}}_{\mathrm{exc},s} )^{\mathrm{Spin}(10,1)}$
  (Prop. \ref{TransgressionElementForM2Cocycle}) is spanned by
  elements which are wedge products $f(\psi) \wedge \mathrm{vol}_{528}$ of
  elements generated from $\overline{\psi} \wedge \psi$ with the volume form $\mathrm{vol}_{528}$ (expression \eqref{TheExceptionalVolumeForm}) of the 528-dimensional exceptional superspacetime.
\end{prop}
\begin{proof}
  Write $n_{\mathrm{bos}}$ and $n_{\mathrm{ferm}}$ for the numbers of bosonic or fermionic generators,
  respectively, in a wedge product. Then consider the degrees
  $$
    \mathrm{deg}_0 := \tfrac{1}{2}n_{\mathrm{ferm}}
    \qquad \text{and} \qquad
    \mathrm{deg}_1 := n_{\mathrm{bos}} - \tfrac{1}{2}n_{\mathrm{ferm}}
    \,.
  $$
  With respect to this bigrading the two summands of the decomposed C-field (\ref{DecomposedCFieldAsSumOfBosonicAndFermionicContribution})
  have the following bidegrees:
  $$
    c_{\mathrm{exc},s}
    \;=\;
    \underset{
      (\mathrm{deg}_1, \mathrm{deg}_0) = (3,0)
    }{
    \underbrace{
      (c_{\mathrm{exc},s})_{\mathrm{bos}}
    }}
      +
    \underset{
       (0,1)
    }{
    \underbrace{
      (c_{\mathrm{exc},s})_{\mathrm{ferm}}
    }
    }\;.
  $$
  Therefore, the $c_{\mathrm{exc},s}\wedge$-complex is the direct sum of three total complexes
  of the following three double complexes
  $$
    \left(
      \mathrm{CE}( \mathbb{R}^{10,1\vert\mathbf{32}}_{\mathrm{exc},s} )^{\mathrm{Spin}(10,1)}\vert_{\mathrm{deg}_1 \,\mathrm{mod}\, 3 = \epsilon  }
      \; ,
      d_0 := (c_{\mathrm{exc},s})_{\mathrm{ferm}}\wedge
      \,,\,
      d_1 := (c_{\mathrm{exc},s})_{\mathrm{bos}}\wedge
    \right)
  $$
  for the off-set
  $$
    \epsilon \in \{0,1,2\}.
  $$
  Since these bicomplexes are concentrated in a half plane, we may compute the $c_{\mathrm{exc},s}$-cohomology by the
  corresponding double complex spectral sequences (see page \pageref{page} for illustration)
  \begin{equation}
  \label{DoubleComplexSpectralSequencesForCCohomology}
  \hspace{-5mm}
  \begin{aligned}
    E_2^{\bullet, \bullet}
    =
    H^\bullet_{(c_{\mathrm{exc},s})_{\mathrm{bos}}}
    \Big(
      H^\bullet_{(c_{\mathrm{exc},s})_{\mathrm{ferm}}}
      \big(
        \mathrm{CE}( \mathbb{R}^{10,1\vert\mathbf{32}}_{\mathrm{exc},s} )^{\mathrm{Spin}(10,1)}\vert_\epsilon
      \big)
     \Big)
    \;\xymatrix@=2em{\ar@{=>}[r]&}\;
    H^\bullet_{c_{\mathrm{exc},s}}
    \left(
     \mathrm{CE}( \mathbb{R}^{10,1\vert\mathbf{32}}_{\mathrm{exc},s} )^{\mathrm{Spin}(10,1)}\vert_\epsilon
    \right)
    \\
  \end{aligned}
  \end{equation}
  Now in the $\mathrm{Spin}(10,1)$-invariant subalgebra the fermions always appear paired, as a linear combination of the
  528 degree-2 elements which are quadratic in the gravitino field
  \begin{equation}
    \label{QuadraticGravitinoGenerators}
    (\overline{\psi}\wedge \psi):=\Big(
    \overline{\psi}_\alpha \wedge \psi^{\beta}
    (\Gamma_a)^\alpha{}_\beta
    \;,\;
    \overline{\psi}_\alpha \wedge \psi^{\beta}
    (\Gamma_{a_1 a_2})^\alpha{}_\beta
    \;,\;
    \overline{\psi}_\alpha \wedge \psi^{\beta}
    (\Gamma_{a_1 \cdots a_5})^\alpha{}_{\beta}
    \Big)\;,
  \end{equation}
  as well as the 528 degree-2 elements which are products of a gravitino field with the auxiliary fermion $\eta$,
  $$
    (\mathrm{dp}_A)
    :=
    \Big(
    \overline{\eta}_\alpha \wedge \psi^{\beta}
    (\Gamma_a)^\alpha{}_\beta
    \;,\;
    \overline{\eta}_\alpha \wedge \psi^{\beta}
    (\Gamma_{a_1 a_2})^\alpha{}_\beta
    \;,\;
    \overline{\eta}_\alpha \wedge \psi^{\beta}
    (\Gamma_{a_1 \cdots a_5})^\alpha{}_{\beta}
    \Big)\;.
  $$
  Under the assumption on $s$, indeed all these 528 elements are non-vanishing in $(c_{\mathrm{exc},s})_{\mathrm{ferm}}$, by  \eqref{CoefficientZerosInDecomposedCField}.

  In terms of these quadratic fermionic elements, the fermionic part (\ref{DecomposedCFieldFermionicContribution}) of the
  decomposed C-field has the simple form
  $$
    (c_{\mathrm{exc},s})_{\mathrm{ferm}}
    =
    \mathrm{dp}_A \wedge \mathrm{dx}^A
  $$
  where
  \begin{equation}
    \label{ExceptionalVielbein}
    (\mathrm{dx}^A)
    :=
    \left(
      e^a, B^{a_1 a_2}, B^{a_1 \cdots a_5}
    \right)
  \end{equation}
  denotes the collection of all the bosonic generators, which may be thought of as the 528-vielbein on the
  exceptional spacetime.

  The C-cohomology of such ``odd symplectic forms'' $\mathrm{dp}_A \wedge \mathrm{dx}^A$ has been computed in \cite{Sev05} (there called H-cohomology), and in our case it is spanned by
  the terms proportional to
  \begin{equation}
    \label{TheCCohomologyRepresentatives}
    f(\psi)
    \wedge
    \underset{A}{\Exterior}\; \mathrm{dx}^A
    =
    f(\psi)
    \wedge
    \mathrm{vol}_{528}
    \,,
  \end{equation}
  where $f(\psi)$ is any element generated from the elements $\overline{\psi} \wedge \psi$ from expression \eqref{QuadraticGravitinoGenerators} alone
  and where $\mathrm{vol}_{528}$ is the exceptional volume form (\ref{TheExceptionalVolumeForm}).
  A homotopy operator that witnesses the vanishing of the C-cohomology away from these elements is
  $$
    \partial_{\mathrm{dp}_A}\partial_{\mathrm{dx}^A}
    \,.
  $$
  This manifestly commutes with the $\mathrm{Spin}(10,1)$-action in our case, so that the statement passes to the $\mathrm{Spin}(10,1)$-invariant algebra.

  This C-cohomology of $(c_{\mathrm{exc},s})_{\mathrm{ferm}}\wedge$ constitutes the first page of the three spectral sequences,
  respectively. Inspection of the degrees shows that none of the higher differentials can be non-vanishing,
  hence that the spectral sequences already collapse on this page.
  This implies the claim.
\end{proof}
\newpage
\label{page}
\begin{tabular}{cl}
{offset} $\epsilon= 0$
&
\raisebox{100pt}{
\scalebox{.8}{
  $$
    \xymatrix{
      & & & &&&
      \\
      & & & &  &&
      \\
      & & & &&  & &
      \\
      & & &  \ddots &
      \\
      &
      &
      &
      & {}_{\mathllap{(522,6)}}\bullet
       \ar[r]|{d_1}
       \ar[rrd]|{d_2}
       \ar[rrrdd]|{d_3}
       &
      \\
      &
      & &&& {}_{\mathllap{(525,3)}}\bullet &&&&&&
      \\
      \ar[rrrrrrrr]_>{ n_{\mathrm{bos}} - n_{\mathrm{ferm}}/2 }_>>>>>>>>>>>>>>>>>>>>>>>>>>{(528,0)} && &&&& \bullet  &&&  &&
      \\
      &  \ar[uuuuuu]^>{ n_{\mathrm{ferm}}/2 } & &&& &&&
    }
  $$
}}
\\
{offset} $\epsilon = 1$
&
\raisebox{100pt}{
\scalebox{.8}{
  $$
    \xymatrix{
      & & & &&&
      \\
      & & & &&&
      \\
      & & \ddots &
      \\
      &
      &
      & {}_{\mathllap{(521,7)}}\bullet
       \ar[r]|{d_1}
       \ar[rrd]|{d_2}
       \ar[rrrdd]|{d_3}
       &
      \\
      &
      & && {}_{\mathllap{(524,4)}}\bullet &&&&&&
      \\
      &&&&& {}_{\mathllap{(527,1)}}\bullet &&&
      \\
      \ar[rrrrrrrr]_>{ n_{\mathrm{bos}} - n_{\mathrm{ferm}}/2 } && &&&   &&&  &&
      \\
      &  \ar[uuuuuu]^>{ n_{\mathrm{ferm}}/2 } & &&& &&&
    }
  $$
}}
\\
{offset} $\epsilon= 2$
&
\raisebox{100pt}{
\scalebox{.8}{
  $$
    \xymatrix{
      & & & &&&
      \\
      &  \ddots &
      \\
      &
      &
       {}_{\mathllap{(520,8)}}\bullet
       \ar[r]|{d_1}
       \ar[rrd]|{d_2}
       \ar[rrrdd]|{d_3}
       &
      \\
      &
      & & {}_{\mathllap{(523,5)}}\bullet &&&&&&
      \\
      &&&& {}_{\mathllap{(526,2)}}\bullet &&&
      \\
      \\
      \ar[rrrrrrrr]_>{ n_{\mathrm{bos}} - n_{\mathrm{ferm}}/2 } && &&&   &&&  &&
      \\
      &  \ar[uuuuuu]^>{ n_{\mathrm{ferm}}/2 } & &&& &&&
    }
  $$
}}
\end{tabular}

\noindent {\bf The $E_1$-pages of the three spectral sequences (\ref{DoubleComplexSpectralSequencesForCCohomology}) which jointly compute the C-cohomology of the decomposed supergravity C-field $c_{\mathrm{exc},s}$ (Prop. \ref{TransgressionElementForM2Cocycle}).}
The fat dots indicate the non-vanishing C-cohomology classes (\ref{TheCCohomologyRepresentatives}) of the fermionic component $(c_{\mathrm{exc},s})_{\mathrm{ferm}}$; they are all represented by a wedge product of bi-fermions $\psi \wedge \overline{\psi} $ with the bosonic volume form $\mathrm{vol}_{528}$ of the exceptional tangent super-spacetime (\ref{TheExceptionalVolumeForm}). Some higher differentials are shown in order to visualize that the spectral sequences all collapse already on the $E_1$-page.

\newpage

As a consequence we are able to identify spherical T-duality in exceptional geometry:

\begin{defn}[M5-twisted cohomology of super exceptional tangent super spacetime]
  \label{tildemuM5TwistedCohomologyOfTangentSuperspacetime}
  Let $s \in \mathbb{R} \setminus \{0\}$ with $\mathbb{R}^{10,1\vert\mathbf{32}}_{\mathrm{exc},s} \xrightarrow{\mathrm{comp}_s} \mathbb{R}^{10,1\vert \mathbf{32}}$ the corresponding
  super exceptional tangent super spacetime (Def. \ref{FermionicExtensionOfExceptionalTangentSuperspacetime}).
 \item {\bf (i)}  Write
  \begin{equation}
    \label{M5BraneCocycleOnExceptionalSpacetime}
    \widetilde \mu_{{}_{M5,s}}
    :=
    (\mathrm{comp}_s)^\ast( \widetilde \mu_{{}_{M5}} )
    =
    2 \mu_{{}_{M5}} + c_{\mathrm{exc},s} \wedge \mu_{{}_{M2}}
    \;\;
    \in
    \mathrm{CE}\big( \mathbb{R}^{10,1\vert\mathbf{32}}_{\mathrm{exc},s} \big)
  \end{equation}
  for the pullback of the M5-brane cocycle
  $\widetilde{\mu}_{{}_{M5}} := 2 \mu_{{}_{M5}} + c \wedge \mu_{{}_{M2}} \in \mathrm{CE}\left( \mathbb{R}^{10,1\vert \mathbf{32}}\right)$ (expressions \eqref{M5cocycleInTDualityPair}), for which the C-field factor $c$ is replaced by the decomposed C-field $c_{\mathrm{exc},s}$ from Prop. \ref{TransgressionElementForM2Cocycle}.

 \item {\bf (ii)}  This induces the corresponding 6-periodic $\tilde \mu_{{}_{M5}}$-twisted $\mathrm{Spin}(10,1)$-invariant cohomology groups
  $$
    H^{\bullet + \tilde \mu_{{}_{M5,s}}}\big(  \mathbb{R}^{10,1\vert \mathbf{32}}_{\mathrm{exc},s} / \mathrm{Spin}(10,1) \big)
    \;:=\;
    H^\bullet
    \big(
      \mathrm{CE}\big(
        \mathbb{R}^{10,1\vert \mathbf{32}}_{\mathrm{exc},s/s'}
      \big)
      ,
      d + \tilde \mu_{{}_{M5,s}}
    \big)
  $$
  of the super exceptional super spacetime via Def. \ref{TwistedCohomology}.
  We write
  \begin{equation}
    \label{c}
    \tilde H^{\bullet + \tilde \mu_{{}_{M5,s}}}
    \big(  \mathbb{R}^{10,1\vert \mathbf{32}}_{\mathrm{exc},s} / \mathrm{Spin}(10,1) \big)
    \;:=\;
    H^\bullet
    \Big(
      \mathrm{CE}\big(
        \mathbb{R}^{10,1\vert \mathbf{32}}_{\mathrm{exc},s/s'}
      \big)/\big\langle f(\psi)\wedge\mathrm{vol}_{528}, f(\psi)d \mathrm{vol}_{528} \big\rangle
      ,
      d + \tilde \mu_{{}_{M5,s}}
    \Big)
  \end{equation}
  for the corresponding cohomology groups after quotienting out the subcomplex
  spanned by the $c_{\mathrm{exc},s}$-cohomology, hence, by Prop. \ref{HCohomologyOfDecomposedCFieldVanishesInSubalgebra},
  the multiples of wedge products of $\mathrm{Spin}(10,1)$-invariants in the $\psi$-generators
  with the exceptional volume form $\mathrm{vol}_{528}$
  (Def. \ref{TheExceptionalVolumeForm}).
\end{defn}

\begin{prop}
[Spherical T-duality on exceptional super spacetime]
  \label{SphericalTDualityOnExtendedSuperspacetime}
  For $s \in \mathbb{R} \setminus \{0,-3/2, -3, -6\}$ we have a spherical T-duality isomorphism
  for decomposed form fields (Theorem \ref{HigherTDualityForDecomposedFields}) between
  the corresponding exceptional superspacetimes $\mathbb{R}^{10,1\vert \mathbf{32}}_{\mathrm{exc},s}$,
  $\mathbb{R}^{10,1\vert \mathbf{32}}_{\mathrm{exc},s'}$ (Def. \ref{FermionicExtensionOfExceptionalTangentSuperspacetime})
  in that we have an isomorphism ${\cal T}_{\mathrm{exc},s}$ of $\tilde \mu_{{}_{M5,s}}$-twisted
  cohomology groups
  according to Def. \ref{tildemuM5TwistedCohomologyOfTangentSuperspacetime} fitting
  into the diagram below
   $$
  \xymatrix@R=5pt{
    H^{ (\bullet +  3)  + \tilde \mu_{M5}}\big(\mathfrak{m}2\mathfrak{brane} / \mathrm{Spin}(10,1)\big)
    \ar[ddd]_-{ (\mathrm{comp}_{s})^\ast }
    \ar@{<-}[rrrr]^-{ {\cal T} }_-{\simeq}
    &&&&
    H^{\bullet + \tilde \mu_{M5}}\big( \mathfrak{m}2\mathfrak{brane} /\mathrm{Spin}(10,1)\big)
    \ar[ddd]^-{ (\mathrm{comp}_{s})^\ast }
    \\
    \\
    \\
    \tilde H^{ (\bullet + 3) + \tilde \mu_{M5,s}}\big(  \mathbb{R}^{10,1\vert\mathbf{32}}_{\mathrm{exc},s}/\mathrm{Spin}(10,1) \big)
    \ar@{<-}[rrrr]^-{ {\cal T}_{\mathrm{exc},s} }_-{\simeq}
    &&&&
    \tilde H^{ \bullet + \tilde \mu_{M5, s } }\big( \mathbb{R}^{10,1\vert\mathbf{32}}_{\mathrm{exc},s}/\mathrm{Spin}(10,1) \big)
    \\
    [- \omega_{\mathrm{w}} + c_{\mathrm{exc},s} \wedge \omega_{\mathrm{nw}})]
    \;
    \ar@{<-|}[rrrr]
    &&&&
    \;[ \omega_{\mathrm{nw}} + c_{\mathrm{exc},s} \wedge \omega_{\mathrm{w}} ]\;.
  }
$$
\end{prop}
\begin{proof}
  The result is obtained by using the result of Prop. \ref{HCohomologyOfDecomposedCFieldVanishesInSubalgebra} in the statement of Theorem \ref{HigherTDualityForDecomposedFields}.
\end{proof}

\begin{remark}[M5-twisted cocycles from heterotic Green-Schwarz mechanism]
  \label{M5TwistedCocyclesFromHeteroticGreenSchwarz}
Proposition \ref{SphericalTDualityOnExtendedSuperspacetime} is a higher/M-theoretic analog of the string theoretic T-duality isomorphism
of \cite[Prop. 6.4]{FSS16}, recalled in the top line of (\ref{TypeIIT}). It it now desireable to obtain also the corresponding analog of the bottom line in (\ref{TypeIIT}),
namely examples of nontrivial $\tilde \mu_{{}_{M5,s}}$-twisted and $\mathrm{Spin}(10,1)$-invariant cocycles.
Due to the nature of the twisted differential $d + \tilde \mu_{{}_{M5,s}}$
with its ingredients from (\ref{M2BraneCocycle}), (\ref{CEAlgebraOfExceptionalTangentSuperspacetime}), (\ref{DecomposedCFieldFermionicContribution})
and (\ref{M5BraneCocycleOnExceptionalSpacetime}), the cocycle condition for $d + \tilde \mu_{{}_{M5,s}}$
translates to a complicated-looking higher-order condition on spinor identities. Here we will leave its solution
as an open mathematical problem. Nevertheless we
now offer an informal argument that nontrivial such cocycles do exist at least in degree 2 or 3 mod 6.
Namely, in terms of string/M-theory, this open problem is the question for the M-theoretic analog of the
D-brane charges $\mu_{{}_{Dp}}$ in (\ref{TypeIIT}) when the role of the fundamental type II string
$\mu_{{}_{F1}}^{\mathrm{IIA/B}}$ is taken by the M5-brane $\tilde \mu_{{}_{M5,s}}$ (as indicated in
the table at the beginning of section \ref{HigherTDualityCorrespondences}, see \cite{tcu} for proposals and further discussion).
In view of this, there is the following ``physics proof'' of the existence of non-trivial cocycles:

\item {\bf (i)}   Recall that the twisted cocycle condition in Prop. \ref{SphericalTDualityOnExtendedSuperspacetime},
   by Def. \ref{TwistedCohomology}, reads
    \begin{align*}
      d \mu_{\bullet} =& \;0\;,
      \\
      d \mu_{\bullet + 6} =& - \tilde \mu_{{}_{M5}} \wedge \mu_n\;.
    \end{align*}
   By the logic of the torsion constraints of supergravity, reviewed in Section \ref{Introduction}, this
   must correspond to an equation of motion for field strength super differential forms as
 \begin{align*}
      d F_{\bullet} = &\; 0
      \\
      d F_{\bullet + 6} = & - H_7 \wedge F_\bullet
      \,,
 \end{align*}
  where $H_7$ is the 7-form flux to which the 5-brane couples, while the nature of the fluxes
  $F_\bullet$ and $F_{\bullet+6}$ is to be determined.

 \item {\bf (ii)} But for degrees $\bullet =2$ mod 6 it was recognized in \cite[Section 3]{Sati09} that the
 effective background equations of motion  in heterotic string theory, whose Yang-Mills sector may be rewritten as
  \begin{equation}
    \label{TheSatiMechanism}
      d \hspace{-6mm}
      \underset{
        {\begin{array}{c} \mbox{\tiny gauge field} \\ \mbox{\tiny magnetic flux}  \end{array}}
      }{
        \underbrace{F_{2}}
      }
      \hspace{-6mm}
       = 0
      \qquad  \qquad
      \text{and}
      \qquad \qquad
      d \hspace{-5mm}
      \underset{
        {\begin{array}{c} \mbox{\tiny gauge field} \\ \mbox{\tiny electric flux}  \end{array}}
      }{
      \underbrace{
        F_{8}
      }}
      \hspace{-5mm}
      =
      - H_7
      \wedge F_2
      \,,
  \end{equation}
  imply that the heterotic gauge field strength $F_2$ and its Hodge dual $F_8$ jointly
  constitute a single cocycle (in the sense discussed in Section \ref{TwistedSuperLInfinityCohomology})
  $$
    \mathcal{F}_{2\,\mathrm{mod}\, 6}
    :=
    \left(
      F_2,  F_8
    \right)
    \;
    \in H^{2\,\mathrm{mod}\,6 + H_7}_{\mathrm{dR}}(X) \otimes \mathfrak{g}
  $$
  in $H_7$-twisted de Rham cohomology (with coefficients in the gauge Lie algebra), where $H_7$ is the flux form to which the NS5-brane couples.

  \item {\bf (iii)} In \cite{Sati09} it was furthermore observed that this is directly analogous to how the RR-fields in type II string theory,
  i.e. the fluxes corresponding to the cocycles $\mu_{{}_{Dp}}$ (\ref{TypeIIT}) to which the D-branes couple,
   as in Section \ref{OrdinaryTypeIITDuality}, jointly constitute a cocycle in $H_3$-twisted de Rham cohomology.
    But under the lift of heterotic string theory to heterotic M-theory \cite{HoravaWitten96a, HoravaWitten96b}
  the $H_7$ in (\ref{TheSatiMechanism}) is identified with the flux corresponding to $\tilde \mu_{{}_{M5}}$ \cite{FSS15}.

 \item {\bf (iv)}  In conclusion this means that from perturbative string theory we expect
  non-trivial twisted cocycles in Prop. \ref{SphericalTDualityOnExtendedSuperspacetime}
  to exist in degree 2 mod 6 (or in degree 3 mod 6 if they are subject to double dimensional reduction \cite[Section 3]{FSS16}) and to correspond to the  M-theoretic lift of the heterotic gauge field and its magnetic dual.
  Of course this is part of the open question for the M-theoretic origin of the gauge fields -- see Remark \ref{NonAbelianDegreesOfFreedom} -- and certainly deserves to be discussed elsewhere.

  \item {\bf (v)} There may be more twisted cocycles, of course. Indeed, Prop. \ref{SphericalTDualityOnExtendedSuperspacetime} predicts
  that if
  $
    \mathcal{F}_{2 \,\mathrm{mod}\,6} = (F_2, F_8)
  $
  is a non-trivial twisted cocycle in degree 2 mod 6, then there must also be a, presently mysterious, spherical-T-dual twisted cocycle
  $$
    \mathcal{F}_{5 \,\mathrm{mod}\,6}
    :=
    \mathcal{T}(F_2, F_8)
  $$
  in degree 5 mod 6.

   Finally it is curious to note that the derivation in \cite[Section 3]{Sati09} shows that the
   degree-7 twisted cocycle relation (\ref{TheSatiMechanism}) is really a cohomological incarnation
   of the Green-Schwarz mechanism in heterotic string theory, the origin of all string theoretic grand unification.
\end{remark}

\medskip
\begin{remark}[M5-brane twisted cocycles involving the exceptional volume element?]
  \label{InvolvingTheExceptionalVolumeElement}
  The restriction
  to the ``reduced'' twisted cohomology groups $\tilde H^{\bullet + \mu_{{}_{M5,s}}}$ (Prop. \ref{HCohomologyOfDecomposedCFieldVanishesInSubalgebra})
  from Def. \ref{tildemuM5TwistedCohomologyOfTangentSuperspacetime} means that
  the spherical T-duality on exceptional superspace from Prop. \ref{SphericalTDualityOnExtendedSuperspacetime}
  would be somewhat ``blurred'' on those $\mathrm{Spin}(10,1)$-invariant cocycles $\tilde \mu_{{}_{M5}}$-twisted cocycles, if any, which involve summands
  that are multiples $f(\psi)\wedge \mathrm{vol}_{528}$ of the exceptional volume form $\mathrm{vol}_{528}$ and Spin-invariant combinations $f(\psi)$ of the super-vielbein.
  It seems plausible that in fact no such cocycles exist; we will revisit this elsewhere.
\end{remark}

\subsection{Exceptional generalized supergeometry}
\label{ExceptionalGenralisedGeometry}

Having  established spherical T-duality on $\mathbb{R}^{10,1\vert \mathbf{32}}_{\mathrm{exc},s}$ (in Prop. \ref{SphericalTDualityOnExtendedSuperspacetime})
we discuss the relevance of Prop. \ref{TransgressionElementForM2Cocycle} for the exceptional generalized geometry of the supergravity C-field (see Remark \ref{ExceptionalSuperTangentSpacetime}).
In fact, Prop. \ref{SphericalTDualityOnExtendedSuperspacetime} will be seen as establishing
 a duality on the exceptional moduli space of C-field configurations.

\medskip

First notice the ``moduli problem'' for C-field configurations: The (2,2)-component of the C-field strength $G_4$ on super-spacetime is constrained to equal the M2-brane super cocycle (\ref{M2BraneCocycle}) in every super tangent space.
However, under the identification of super tangent spaces with $\mathbb{R}^{10,1\vert \mathbf{32}}$ established by the super vielbein $(E^A) := (e^a, \psi^\alpha)$
\cite[(145)]{BST2}
$$
  (G_4)_{(2,2)}
   \;=\;
  \underset{
    \mu_{{}_{M2}}
  }{
  \underbrace{
    \tfrac{i}{2} (\Gamma_{a b})^\alpha{}_\beta \overline{\psi}_\alpha \wedge \psi^\beta \wedge e^a \wedge e^b
  }}
$$
this still leaves, even in the absence of any bosonic flux, hence even if
$$
  (G_4)_{4,0} = 0
  \,,
$$
a moduli space of possible C-field configurations. That is, the C-field itself is, of course, locally a differential
3-form potential $C$ for $G_4$
$$
  d C = \mu_{{}_{M2}}
  \,.
$$
Here we are identifying via Prop. \ref{LeftInvariantDifferentialFormsOnLieGroupRepresentCEElements}
$$
  \mu_{{}_{M2}}
    \in
  \mathrm{CE}(\mathbb{R}^{10,1\vert\mathbf{32}})
    \simeq
  \Omega^4_{\mathrm{LI}}(\mathbb{R}^{10,1\vert \mathbf{32}})
  \xymatrix{\ar@{^{(}->}[r]&}
  \Omega^4(\mathbb{R}^{10,1\vert \mathbf{32}})
$$
with a super-differential 4-form on $\mathbb{R}^{10,1}$ that happens to be left-invariant (hence invariant under super-translations).

\medskip
Now, of course, the actual C-field $C \in \Omega^3(\mathbb{R}^{10,1\vert \mathbf{32}})$ is never left-invariant,
because if it were then $d C= \mu_{{}_{M2}}$ would mean that $[\mu_{{}_{M2}}] = 0 \in H^\bullet(\mathfrak{g}) \simeq H^\bullet_{\mathrm{dR}, \mathrm{LI}}(\mathbb{R}^{10,1\vert \mathbf{32}})$, which is not the case.
But, in fact, the C-field itself is to be thought of as part of the connection data on a 2-gerbe whose curvature 4-form
is $\mu_{{}_{M2}}$, and the transformation properties for this connection data is more flexible in that it allows
invariance up to higher gauge transformation. According to \cite[Section 2.2.2]{SSS1},\cite[Def. 4.3.6]{FSS12}
such 2-gerbe connections may be encoded by transgression elements as in Def. \ref{TransgressionElements}.
For the M2-cocycle $\mu_{{}_{M2}}$ such a transgression element is precisely what Prop. \ref{TransgressionElementForM2Cocycle} establishes:
The transgression element is the decomposed C-field (\ref{TransgressionOfM2CocycleViaDecomposedCField}) in the D'Auria-Fr{\'e} algebra
\cite[Section 6]{DF}, \cite[Section 3]{BAIPV04}.

\medskip
Consequently,  Example \ref{PrimitivesFromSections} says that
the transgression of $\mu_{{}_{M2}}$ (\ref{M2BraneCocycle}) on
$\mathbb{R}^{10,1\vert \mathbf{32}}$ allows to obtain well-behaved C-field configurations
by pullback along linear sections $\sigma$ of the exceptional super spacetime regarded as a bundle
over 11d super-spacetime:
$$
  \raisebox{20pt}{
  \xymatrix{
    \mbox{\bf Moduli space}
    \ar@{}[d]|{ \mbox{ \footnotesize \bf Classifying map} }
    &
    \mathbb{R}^{10,1\vert \mathbf{32}}_{\mathrm{exc},s}
    \ar[d]^{\pi_{\mathrm{exc},s}}
    \\
    \mbox{\bf Spacetime}
    &
    \mathbb{R}^{10,1\vert \mathbf{32}}
    \ar@/^1pc/[u]^{\sigma}
  }
  }
 \;\;
  \xymatrix{
  \ar@{=>}[r] &}
  \;\;
  d \hspace{-1.5mm}
  \underset{
    {\begin{array}{c} \mbox{\tiny C-field} \\ \mbox{\tiny configuration} \end{array}}
   }{\underbrace{(\sigma^\ast c_{\mathrm{exc},s})}}  \hspace{-2.5mm}=\mu_{{}_{M2}}
  \,.
$$
Therefore:
{\it
\begin{enumerate}
  \item
    Each of the fermionic extensions $\mathbb{R}^{10,1 \vert \mathbf{32}}_{\mathrm{exc},s}$ (Def. \ref{FermionicExtensionOfExceptionalTangentSuperspacetime}) of the
    exceptional super-spacetime $\mathbb{R}^{10,1\vert \mathbf{32}}$ (Def. \ref{MaximalCentralExtensionOfNIs32Superpoint}),
    for each value of $s \in \mathbb{R} \setminus \{0\}$, serves as a moduli space for C-field configurations.
  \item
   The decomposed $C$-field $c_{\mathrm{exc},s} \in \mathrm{CE}( \mathbb{R}^{10,1\vert \mathbf{32}}_{\mathrm{exc},s})$ (Prop. \ref{TransgressionElementForM2Cocycle})
   is the corresponding universal field on the moduli space, whose pullback along classifying maps $\sigma$ yield the
   actual C-field configurations on super-Minkowski spacetime.
\end{enumerate}
}

We illustrate how this works with the following basic examples.

\begin{example}[C-fields via splitting of exceptional tangent bundle]
Let $C \in \Omega^3(\mathbb{R}^{10,1\vert \mathbf{32}})$ be a bosonic differential 3-form, with components
$C = C_{a_1 a_2 a_3} d x^a  \wedge d x^b \wedge d x^c$
then a section is obtained by contraction of vectors in $C$
\begin{equation}
  \label{SectionForPlainC}
  \xymatrix{
    v + \iota_v C & \mathbb{R}^{10,1} \oplus \Exterior^2( \mathbb{R}^{10,1} ) \oplus \Exterior^5( \mathbb{R}^{10,1} )\;.
    \\
    v
    \ar@{|->}[u] & \mathbb{R}^{10,1} \ar[u]_{\sigma_{{}_C}}
  }
\end{equation}
Pullback along this section of the generators
$$
  B_{a b} \in \mathrm{CE}(\mathbb{R}^{10,1\vert \mathbf{32}}_{\mathrm{exc},s})
  \simeq
  \Omega^\bullet_{\mathrm{LI}}(\mathbb{R}^{10,1\vert \mathbf{32}}_{\mathrm{exc},s}\;,
$$
from Prop. \ref{SpinActionOnExceptionalTangentSuperspacetime} yields the corresponding ``wrapping modes'' of $C$, namely
$$
  \sigma_{{}_C}^\ast B_{ab} = C_{a b c} d x^c
  \;,
$$
and hence the pullback of the decomposed C-field (\ref{DecomposedCFieldAsSumOfBosonicAndFermionicContribution})
has bosonic component (\ref{DecomposedCFieldBosonicContribution}) which at $s = -3$ (see (\ref{CoefficientZerosInDecomposedCField})) reproduces the 3-form $C$:
$$
  \begin{aligned}
    \sigma_{{}_C}^\ast
       (c_{\mathrm{exc},-3})_{\mathrm{bos}}
    & =
    \sigma_{{}_C}^\ast(B_{ab} \wedge e^a \wedge e^c)
    \\
    & =
    C\;.
  \end{aligned}
$$
This is the way that C-field configurations have been proposed to be encoded by \emph{exceptional generalized geometry} in \cite{Hull07, PachecoWaldram08, Bar12};
see also Remark \ref{ExceptionalSuperTangentSpacetime}.
\end{example}

However, notice that this looks different at different values of $s$:

\begin{example}[Decomposed bosonic C-field at different values of $s$]
\label{SectionForPlainCAtDifferentParameter}
For other values of $s$, the very same section $\sigma_{{}_C}$ (\ref{SectionForPlainC})
induces \emph{different} C-field configurations via pullback. This is
due to the second summand $\propto B^a{}_b \wedge B^b{}_c \wedge B^c{}_a$
in (\ref{DecomposedCFieldBosonicContribution}), namely a linear combination of $C$ with the 3-form
$$
  \sigma_{{}_C}^\ast( B^a{}_b \wedge B^b{}_c \wedge B^c{}_a )
  \; \propto \;
  \mathrm{CS}(A)\;,
$$
where on the right we have the (flat) Chern-Simons form for $C$ regarded as an $\mathfrak{so}(10,1)$-valued differential 1-form $A$,
$$
  (A_\mu)^a{}_b := C_\mu{}^a{}_b
  \,.
$$
\end{example}

\medskip
\begin{remark}[Non-abelian gauge degrees of freedom in M-theory]
  \label{NonAbelianDegreesOfFreedom}
  It is a famous open problem that the following two facts are superficially incompatible:
  \begin{enumerate}
    \item On the one hand, M-theory must contain avatars of nonabelian gauge fields,
    since these are seen in its string theoretic weak coupling limit in various guises.
    \item On the other hand, its
   low-energy-limit in the form of 11d supergravity theory only contains, with the C-field, an abelian, albeit higher-, gauge field.
  \end{enumerate}
  More specifically, various anomaly-cancellation arguments (see \cite{3stack} for review and homotopy-theoretic discussion
  in line with the present perspective) show that the supergravity C-field ought to
  contain a contribution that is locally given by the Chern-Simons 3-form of non-abelian gauge field,
  which plain 11d supergravity knows nothing about such a Chern-Simons summand.
  However, notice that in the perspective on M-theory via super $L_\infty$-homotopy theory, the equations of motion of
  11d supergravity are the consequence of just one of several super tangent space-wise super $L_\infty$-algebraic
  structures, namely of the torsion constraint (\ref{VielbeinDifferential}), which implies
  the equations of motion of 11d supergravity by \cite{CL, Howe97}.
  But, in addition, there is also the constraint (\ref{BSTMembrane}) on the M2-brane's WZW term.
  If one demands that this constrained be solved by super $L_\infty$-algebraic means, namely by transgression
  (Def. \ref{TransgressionElements}), then Example \ref{SectionForPlainCAtDifferentParameter} shows that this makes the Chern-Simons  term of a nonabelian gauge field appear as a summand of the C-field.
  See \cite{multiple} for details on how the latter arises in the context of multiple M5-brane theories.
\end{remark}


Finally, we observe that the perspective of exceptional generalized geometry allows us to
explain the ``meaning'' of the extra fermionic generators $\eta^\alpha$ from (\ref{ExceptionalFermionicGenerator}), which was the cause of some concern in \cite{ADR16}:

\begin{example}
[Interpreting the extra fermionic generator via heterotic M-theory]
\label{InterpretingTheExtraFermionicGenerator}
A linear section of the exceptional super-tangent bundles with non-vanishing component in the extra fermion generators is given in particular by
a choice of fermion $(\chi^\alpha)$ via
$$
  \xymatrix{
    v + (v^a (\Gamma_a)^\alpha{}_\beta \chi^\beta)
    \\
    v
    \ar@{|->}[u]_{\sigma_\chi}
  }
$$
Pullback along this section of the extra fermionic generators from (\ref{ExceptionalFermionicGenerator})
$$
  \eta^\alpha \in \mathrm{CE}(\mathbb{R}^{10,1\vert \mathbf{32}}_{\mathrm{exc},s})
  \simeq
  \Omega^\bullet_{\mathrm{LI}}(\mathbb{R}^{10,1\vert \mathbf{32}}_{\mathrm{exc},s})
$$
yields
$$
  \sigma_\chi^\ast( \eta^\alpha ) = d x^a (\Gamma_a)^\alpha{}_\beta \chi^\beta
$$
and hence takes the following value on the fermionic component (\ref{DecomposedCFieldFermionicContribution})
of the decomposed C-field:
$$
  \sigma_\chi^\ast( (c_{\mathrm{exc},s})_{\mathrm{ferm}} )
  \; \propto \;
  \overline{\psi}_\alpha (\Gamma_{a b})^\alpha{}_{\beta} \chi^\beta e^a \wedge e^b
  \,.
$$
After an identification of bulk fermions with Ho{\v r}ava-Witten-boundary fermions as in \cite{EvslinSati02},
this is the form in which the dilatino has to appear as a summand in the supergravity C-field in heterotic M-theory,
according to \cite[(4.21)]{ADR86}.
\end{example}

Note that this works because our identification of the bosonic part of the D'Auria-Fr{\'e} (DF) algebra (Def. \ref{FermionicExtensionOfExceptionalTangentSuperspacetime}) with  the exceptional tangent bundle of \cite{Hull07} immediately implies that we have to interpret the extra fermionic component
of the DF algebra as providing the supersymmetrization of the exceptional generalized geometry
proposal for M-theory.

\subsection{Spherical T-duality of M5-branes over 7d super-spacetime}
\label{SphericalTDualityOver7dSpacetime}

Under dimensional reduction, the non-wrapping part of the M2/M5-brane supercocycle remains
an M2/M5-brane supercocycle in lower dimensions with higher supersymmetry. In some lower dimensions
it remains even at minimal supersymmetry. This is the case notably in $d =7$.
(Notice that the $N=2$ supergravity in seven dimensions   \cite{TvN}\cite{SSz} can be obtained from
eleven dimensions in a consistent way \cite{Kan}.)

\medskip
The direct analog of the spherical T-duality for M5-branes from Example
\ref{HigherTDualCorrespondenceForM2Brane} holds in 7d . Instead of
repeating the whole applicable discussion in detail, we encapsulate this in the following.

\begin{remark}
{\bf (i)}  The cochains on minimal 7d super-Minkowski spacetime
  $$
    \mu^{7d}_{{}_{M2}}, \;\;\; \mu^{7d}_{{}_{M5}}
    \;\in\;
    \mathrm{CE}
    \big(
      \mathbb{R}^{ 6,1 \vert \mathbf{16} }
    \big)
  $$
  satisfy the analog of the relation (\ref{M2M5CocyclesOn11dSuperMinkowski}) (see also {\cite[expression (4.36)]{ADR16}}):
\begin{equation}
  d \mu^{7d}_{{}_{M2}} = 0
  \qquad
  \text{and}
  \qquad
  d \mu^{7d}_{{}_{M5}} = -\tfrac{1}{2} \mu^{7d}_{{}_{M2}} \wedge \mu^{7d}_{{}_{M2}}
  \,,
\end{equation}
hence constitute a rational 4-sphere valued supercocycle
$$
  \big(
    \mu^{7d}_{{}_{M2}}, \; \mu^{7d}_{{}_{M5}}
  \big)
  \;:\;
  \mathbb{R}^{6,1\vert \mathbf{16}}
   \xymatrix{\ar[r] &}
  \mathfrak{l}(S^4)\;.
$$
\item {\bf (ii)} Consequently, the discussion in section \ref{SphericalTDualityOfM5BranesOnM2ExtendedSpacetime} applies and provides an example of topological spherical T-duality (Def. \ref{HigherTDualityCorrespondences}) for super M5-branes on the corresponding M2-extended super-spacetimes
just as in Prop. \ref{M5CocycleIsSPhericalTDualToItself}
$$
  \xymatrix{
    b^{6}\mathbb{R}
    &&
    \mathfrak{m}2\mathfrak{brane}^{7d}
    \ar[dr]_{\pi^A := \mathrm{hofib}(\mu^{7d}_{{}_{M2}}) \phantom{AAAA}}
    \ar[ll]_-{\widetilde{\mu}^{7d}_{{}_{M5}}}
    &&
    \mathfrak{m}2\mathfrak{brane}^{7d}
    \ar[dl]^{\phantom{AAA}\pi^B := \mathrm{hofib}(\mu^{7d}_{{}_{M2}})}
    \ar[rr]^-{\widetilde{\mu}^{7d}_{{}_{M5}}}
    &&
    b^{6}\mathbb{R}
    \\
    &&
    & \mathbb{R}^{6,1\vert \mathbf{16}}
    \ar[dl]^{\mu^{7d}_{{}_{M2}}}
    \ar[dr]_{\mu^{7d}_{{}_{M2}}}
    \\
    &&
    b^{3} \mathbb{R} && b^{3} \mathbb{R}
  }
$$
$$
\xymatrix@C=6em{
  \big(\mathfrak{m}2\mathfrak{brane},\widetilde{\mu}^{7d}_{{}_{M5}} \big)
  \;\;\ar@{<->}[rr]^{\cal T} &&
  \;\; \big(\mathfrak{m}2\mathfrak{brane},\widetilde{\mu}^{7d}_{{}_{M5}} \big)\;.
  }
$$
\item {\bf (iii)} Moreover, there is again an exceptional superspacetime $\mathbb{R}^{6,1\vert \mathbf{16}}_{\mathrm{exc},s}$ as in Def. \ref{FermionicExtensionOfExceptionalTangentSuperspacetime} over which the 7d C-field has a transgression element and decomposes
as in Prop. \ref{TransgressionElementForM2Cocycle}; this is due to \cite[Section 4]{ADR16}.
Accordingly there is a 7d analog of spherical T-duality on
exceptional spacetime as in Prop. \ref{SphericalTDualityOnExtendedSuperspacetime}.
\end{remark}

\subsection{Parity isomorphism on brane charges over exceptional M-theory spacetime}
\label{ParitySymmetry}

In this section we prove a ``parity isomorphism'' for $\widetilde{\mu}_{{}_{M5}}$-twisted cohomology
on exceptional superspacetime (Prop. \ref{ParitySymmetryForDecomposedCField} below),
different from but akin to the spherical T-duality isomorphism in Prop. \ref{SphericalTDualityOnExtendedSuperspacetime}. Before we construct the isomorphism in Prop. \ref{ReflectionAutomorphismOnExceptionalTangentSuperSpacetime} below,
first we survey some background on the role that such an isomorphism may be expected to play in M-theory.

\medskip
The parity symmetry appears at the level of supergravity as
the transformation
 $C_3 \mapsto -C_3$ together with an odd number of space or time
reflections \cite{DLP}.
 In M-theory,  parity acts by orientation-reversal, together with $G_4 \to -G_4$.

 \medskip
 The considerations of parity in M-theory affect the description of
 some of its objects. For instance, the usual requirement that the fivebrane
 worldvolume be oriented can be relaxed, by virtue of M-theory conserving
 parity \cite{Wi1}.
For the membrane, the theories describing multiple membranes
should preserve parity, which places constraints on the structure
constants appearing in the Lagrangian \cite{BLMP}.
Chern-Simons-matter theories describing fractional M2-branes \cite{ABJ},
 arising from backgrounds
${\rm AdS}_4 \times S^7/\Z_k$ with torsion class
in $H^4(S^7/\Z_k; \Z)\cong \Z_k$, lead to symmetries
of the form $U(n+ \ell)_k \times U(n)_{-k}$, where $k$ is the
Chern-Simons level. The parity symmetry acting as
$C_3 \mapsto -C_3$ takes $k$ to $-k$ and then
$G_4 \mapsto k - G_4$, so that the rank $\ell$ goes to
$\ell-k$ in structure groups of level $-k$ \cite{ABJ}.

\medskip
From global geometric and topological perspective,
 M-theory is parity invariant, and so should in principle be formulated
 in a way that makes sense on unoriented, and possibly
 orientable manifolds (see \cite{Mo}).  How
 to formulate M-theory on such manifolds?
 The parity problem was discussed in \cite{DFM} \cite{Mo}
 with the purpose of extending
 the $E_8$-model to unoriented manifolds, via orientation
 double covers and Deck
 transformations. The authors point out (see particularly the summary
 discussion in \cite{Mo}) that this problem is still unsatisfactorily
 addressed and hence a proper formulation is still lacking.

\medskip
Now we consider the degree four field to be captured by an $E_8$ bundle
over spacetime, as in \cite{Wi1}\cite{DMW} \cite{DFM}, characterized by
an integral degree four class $a$.  In terms of this, the quantization condition
$a=G_4 + \tfrac{1}{2}\lambda$ holds, where $\lambda=\tfrac{1}{2}p_1$
is the first Spin characteristic class. The parity transformation acts
on the degree four classes as
$$
a \mapsto \lambda - a
\qquad
\text{and}
\qquad
G_4 \mapsto -G_4\;.
$$
 In this global formulation, one has
 \begin{itemize}
\item When the first Spin characteristic class
 $\lambda=\tfrac{1}{2}p_1$ is zero, i.e. on String manifolds,
  then this is simply $a \mapsto -a$. Note that considering such higher structures
in the rational setting has been discussed extensively in \cite{SW}.

  \item When considering M-theory on a
  circle, $Y^{11}=S^1 \times X^{10}$, the parity symmetry acts by reflection
  of the $S^1$ factor together with $a \mapsto \lambda -a$.
 \item If $G_4=0$ in cohomology then the corresponding configuration would
  be parity-invariant.
 \end{itemize}

\medskip
Observe that when we take our spacetimes to be a String manifold
(as done in \cite{loop}\cite{E8}) then
the parity transformation acts simply by a minus sign on both
$a$ and $G_4$.
We could then take the degree four configurations to correspond to two
$E_8$ bundles, related by a parity transformation.
We will aim to
find a home of this transformation in the context of higher rational
T-duality.

\begin{example}
[Parity as rational T-duality of $E_8$ bundles]
Since the homotopy group of $E_8$ are concentrated in degrees $(3, 15, \cdots)$,
the group $E_8$ has the same homotopy type
as $K(\Z, 3)$ up to degree 14, and hence the classifying spaces $BE_8$ and $K(\Z, 4)$
have the same homotopy type up to dimension 15. We now make
the further observation that $E_8$ and $SU(2)$ have the same rational
homotopy and cohomology in the above range. This means that overall
we have the identifications
$$
``E_8 \simeq_{14,\Q} K(\Q, 3) \simeq_\Q SU(2) \simeq_\Q S^3"\;.
$$
Consequently, within equivalence in rational homotopy theory we are free to view our $E_8$
bundle over an 11-dimensional base space as a 3-sphere bundle.  This is then summarized as follows
$$
\xymatrix@R=1.5em{
 S^3_\Q \simeq K(\Q, 3) \simeq_{14,\Q} E_8
\ar[r]
&
E
\ar[dr]^\pi
&&
E'
\ar[dl]_{\pi'}
&
E_8 \simeq_\Q K(\Q, 3) \simeq S^3_\Q
\ar[l]
\\
&& Y &
}
$$
Taking the class of the bundle $E$ to be $a$ and the class of the
bundle $E'$ to be $-a$ then puts the two bundles as a parity dual pair, which
fits into our discussion of T-duality for rational sphere bundle as a special
case.
A parity-invariant formulation of the $E_8$ model is given in \cite{DFM}
 by passing to the orientation double
cover $Y_d$ of spacetime $Y$ and declaring the $C$-field to be the
parity invariant cocycle on $Y_d$. Hence we could use $Y_d$ in place of $Y$.
\end{example}

\medskip
We now return to our supercocycles and study the effect of parity on them.

\begin{prop}[Reflection automorphism on exceptional tangent super-spacetime]
\label{ReflectionAutomorphismOnExceptionalTangentSuperSpacetime}
There is an action of $\mathbb{Z}/2$ on $\mathrm{CE}( \mathbb{R}^{10,1\vert \mathbf{32}}_{\mathrm{exc},s} )$
where the non-trivial element $\rho \in \mathbb{Z}/2$ acts dually by
\footnote{Beware that this is saying that the $B_{a_1 a_2}$-generator picks up a sign precisely if its indices do \emph{not} take the value 10.}
$$
  \rho^\ast
  \;:\;
  \left\{
 \begin{aligned}
    e^a
    & \longmapsto
    \left\{
      \begin{array}{rcl}
        - e^a &\vert& a = 10
        \\
        e^a &\vert& \mbox{otherwise}
      \end{array}
    \right.
    \\
    B_{a_1 a_2}
    &
    \longmapsto
    \left\{
      \begin{array}{rcl}
        - B_{a_1 a_2} &\vert& a_1, a_2 \neq 10
        \\
        B_{a_1 a_2} &\vert& \mbox{otherwise}
      \end{array}
    \right.
    \\
    B_{a_1 \cdots a_5}
    &
    \longmapsto
    \left\{
      \begin{array}{rcl}
        - B_{a_1 \cdots a_5} &\vert& \mbox{one of the} \; a_i = 10
        \\
        B_{a_1 \cdots a_5} &\vert& \mbox{otherwise}
      \end{array}
    \right.
    \\
    \psi & \longmapsto \phantom{-}\Gamma_{10} \psi
    \\
    \eta & \longmapsto - \Gamma_{10} \eta
 \end{aligned}
 \right.
$$
\end{prop}

\begin{proof}
It is clear that $\rho^\ast$ is an isomorphism in the underlying graded algebra. What needs to be checked is that it does
respect the differentials.
Hence first we need to show that under $\rho^\ast$ the bispinorial expressions
$\overline{\psi} \Gamma^a \psi$, $\overline{\psi} \Gamma^{a_1 a_2} \psi$, $\overline{\psi} \Gamma^{a_1 \cdots, a_5} \psi$
pick up the same signs as the corresponding elements $e^a$, $B^{a_1 a_2}$ and $B^{a_1 \cdots a_5}$ do.
This is equivalent to saying that their contractings are preserved by $\rho^\ast$.
Using the identities $\overline{\psi} = \psi^\dagger \Gamma_0$,
$(\Gamma_{10})^\dagger = - \Gamma_{10}$,
and $\Gamma_{10} \Gamma_{10} = -1$ we establish $\rho$-invariance of the supercocycles
as follows:
$$
  \begin{aligned}
    \rho( \overline{\psi} \Gamma_a \psi \wedge e^a )
    & =
    \underset{0 \leq a \leq 9}{\sum} \overline{ \Gamma_{10} \psi } \Gamma_a \Gamma_{10} \psi \wedge e^a
    +
    \overline{ \Gamma_{10} \psi } \Gamma_{10} \Gamma_{10} \psi \wedge (- e^{10})
    \\
    & =
    \underset{0 \leq a \leq 9}{\sum} \psi^\dagger (- \Gamma_{10}) \Gamma_0 \Gamma_a \Gamma_{10} \psi \wedge e^a
    +
    \psi^\dagger (- \Gamma_{10}) \Gamma_0 \Gamma_{10} \Gamma_{10} \psi \wedge (- e^{10})
    \\
    & =
    +
    \overline{\psi} \Gamma_a \psi \wedge e^a\;,
 \end{aligned}
$$
$$
  \begin{aligned}
    \rho( \tfrac{1}{2} \overline{\psi} \Gamma_{a_1 a_2} \psi \wedge B^{a_1 a_2} )
    & =
    \underset{0 \leq a_1 < a_2 \leq 9}{\sum} \overline{ \Gamma_{10} \psi } \Gamma_{a_1 a_2} \Gamma_{10} \psi \wedge (- B^{a_1 a_2})
    +
    \underset{0 \leq a \leq 9}{\sum} \overline{ \Gamma_{10} \psi } \Gamma_{a, 10} \psi \wedge B^{a, 10}
    \\
    & =
    \underset{0 \leq a_1 < a_2 \leq 9}{\sum} \psi^\dagger (-\Gamma_{10}) \Gamma_{0} \Gamma_{a_1 a_2} \Gamma_{10} \psi \wedge (- B^{a_1 a_2})
    +
    \underset{0 \leq a \leq 9}{\sum} \psi^\dagger (- \Gamma_{10}) \Gamma_0 \Gamma_{a, 10} \psi \wedge B^{a, 10}
    \\
    & =
    +
    \tfrac{1}{2} \overline{\psi} \Gamma_{a_1 a_2} \psi \wedge B^{a_1 a_2}\;,
 \end{aligned}
$$
and
$$
\hspace{-2mm}
  \begin{aligned}
    \rho( \tfrac{1}{5!} \overline{\psi} \Gamma_{a_1 \cdots a_5} \psi \wedge B^{a_1 \cdots a_5} )
     =&
   \hspace{-1mm} \underset{{}_{0 \leq a_1 < \cdots < a_5 \leq 9}}{\sum}
   \hspace{-3mm} \overline{ \Gamma_{10} \psi } \Gamma_{a_1 \cdots a_5} \Gamma_{10} \psi \wedge (B^{a_1 \cdots a_5})
    +
    \hspace{-2mm}\underset{{}_{0 \leq a_1 < \cdots < a_4  \leq 9}}{\sum}
    \hspace{-3mm} \overline{ \Gamma_{10} \psi } \Gamma_{a_1 \cdots a_4, 10} \psi \wedge B^{a_1 \cdots a_4, 10}
    \\
     =&
    \underset{{}_{0 \leq a_1 < \cdots < a_5 \leq 9}}{\sum} \psi^\dagger (- \Gamma_{10}) \Gamma_0  \Gamma_{a_1 \cdots a_5} \Gamma_{10} \psi \wedge (B^{a_1 \cdots a_5})
    \\
    &  \phantom{=} +
        \underset{{}_{0 \leq a_1 < \cdots < a_4  \leq 9}}{\sum} \hspace{-2mm}
         \psi^\dagger (- \Gamma_{10}) \Gamma_0 \Gamma_{a_1 \cdots a_4, 10} \psi \wedge B^{a_1 \cdots a_4, 10}
    \\
     =&
    +
     \tfrac{1}{5!} \overline{\psi} \Gamma_{a_1 \cdots a_5} \psi \wedge B^{a_1 \cdots a_5}\;.
 \end{aligned}
$$
Similarly, the following computation shows that $d$ and $\rho^*$ commute
$$
  \begin{aligned}
    d \rho^\ast( \overline{\eta} )
    & =
    d (\overline{ - \Gamma_{10} \eta})
    \\
    & =
    d \eta^\dagger \Gamma_{10} \Gamma_0
    \\
    & =
    - d \overline{\eta} \Gamma_{10}
    \\
    & =
    - \overline{\psi} \Gamma_a  \Gamma_{10} e^a
    - \overline{\psi} \Gamma_{a_1 a_2} \Gamma_{10} B^{a_1 a_2}
    - \overline{\psi} \Gamma_{a_1 \cdots a_5} \Gamma_{10} B^{a_1 \cdots a_5}
    \\
    & =
    - \psi^\dagger \Gamma_0 (- \Gamma_{10}) \Gamma_a   \rho^\ast(e^a)
    - \psi^\dagger \Gamma_0 (- \Gamma_{10}) \Gamma_{a_1 a_2} \rho^\ast(B^{a_1 a_2})
    - \psi^\dagger \Gamma_0 (- \Gamma_{10}) \Gamma_{a_1 \cdots a_5} \rho^\ast(B^{a_1 \cdots a_5})
    \\
    & =
    \overline{ \rho^\ast{\psi} } \Gamma_a \rho^\ast(e^a)
    +
    \overline{ \rho^\ast{\psi} } \Gamma_{a_1 a_2} \rho^\ast(B^{a_1 a_2})
    +
    \overline{ \rho^\ast{\psi} } \Gamma_{a_1 \cdots a_5} \rho^\ast(B^{a_1 \cdots a_5})
    \\
    & = \rho^\ast\left( d \overline{\eta} \right)\;.
  \end{aligned}
$$

\vspace{-5mm}
\end{proof}

We use this result to determine the effect on the M-brane supercocycles.

\medskip
\begin{prop}[Parity symmetry of decomposed supergravity C-field]
  \label{ParitySymmetryForDecomposedCField}
  Under the reflection automorphism $\rho$ from Prop. \ref{ReflectionAutomorphismOnExceptionalTangentSuperSpacetime} we have that

  \item {\bf (i)} the decomposed supergravity C-field (Prop. \ref{TransgressionElementForM2Cocycle}) changes sign:
  $
    \rho^\ast( c_{\mathrm{exc},s} ) = - c_{\mathrm{exc},s}
  $;
  \item {\bf (ii)} the decomposed M5-brane cocycle (\ref{M5BraneCocycleOnExceptionalSpacetime}) is invariant:
  $
    \rho^\ast(\tilde \mu_{{}_{M5,s}}) = \tilde \mu_{{}_{M5,s}}
  $.
\end{prop}
\begin{proof}
{\bf (i)} The first statement follows by inspection. For instance, the transformation of the first summand
 $B_{a b} \wedge e^a \wedge e^b$
  of the bosonic component $(c_{\mathrm{exc},s})_{\mathrm{bos}}$
  may be computed as follows:
  $$
    \begin{aligned}
      \rho^\ast(B_{a_1 a_2} \wedge e^{a_1} \wedge e^{a_2})
      & =
      2 \underset{0 \leq a_1 < a_2 \leq 9}{\sum} (- B_{a_1 a_2}) \wedge e^{a_1} \wedge e^{a_2}
      +
      2 \underset{0 \leq a \leq 9}{\sum} B_{a 10} \wedge e^{a} \wedge (- e^{10})
      \\
      & =
      - B_{a_1 a_2} \wedge e^{a_1} \wedge e^{a_2}\;,
    \end{aligned}
  $$
  while the transformation of the second summand
  may be computed as
  $$
    \begin{aligned}
      \rho^\ast\left(
        B^a{}_b \wedge B^b{}_c \wedge B^c{}_a
      \right)
      & =
      \underset{0 \leq a, b \leq 9}{\sum} (-B^{a}{}_b)
      \big(
        \underset{0 \leq c \leq 9}{\sum} (- B^b{}_{c}) \wedge (- B^c{}_a)
        +
        B^b{}_{10} \wedge B^{10}{}_a
      \big)
      +
      \mbox{cyclic}
      \\
      & =
      -
      B^a{}_b \wedge B^b{}_c \wedge B^c{}_a\;,
    \end{aligned}
  $$
  and similarly for the other bosonic summands.

  For the fermionic term $(c_{\mathrm{exc},s})_{\mathrm{ferm}}$ we already checked at the beginning of the proof of Prop. \ref{ReflectionAutomorphismOnExceptionalTangentSuperSpacetime}
  that it becomes invariant under $\rho^\ast$ if we replaced the factor of $\eta$ by $\psi$. But the transformations of
  $\eta$ and $\psi$ under $\rho^\ast$ are the same up to a minus sign.

 \noindent {\bf (ii)}  Regarding the second statement, by
  the same reasoning as in proof of Prop. \ref{ReflectionAutomorphismOnExceptionalTangentSuperSpacetime}, we have
  $$
    \rho^\ast(\mu_{{}_{M5}})
    =
    \rho^\ast( \tfrac{1}{5!}\overline{\psi}\Gamma_{a_1 \cdots a_5} \psi \wedge e^{a_1} \wedge \cdots \wedge e^{a_5} )
    =
    \mu_{{}_{M5}}
  $$
  and
  $$
    \rho^\ast(\mu_{{}_{M2}})
    =
    \rho^\ast( \tfrac{i}{2!}\overline{\psi}\Gamma_{a_1 a_2} \psi \wedge e^{a_1} \wedge e^{a_2} )
    =
    - \mu_{{}_{M2}}
    \,.
  $$
  Hence the statement follows from the previous one:
  $$
    \rho^\ast(\tilde \mu_{{}_{M5,s}})
    =
    \rho^\ast( 2 \mu_{{}_{M5}} +  c_{\mathrm{exc},s} \wedge \mu_{{}_{M2}})
    =
    \tilde \mu_{{}_{M5,s}}
    \,.
  $$

 \vspace{-6mm}
\end{proof}

Indeed, this is compatible with the results of parity
in the topological sector in \cite{DFM}.
As a direct consequence, we can establish the following.

\begin{prop}[Parity as an isomorphism on twisted cohomology]
  For $s \in \mathbb{R} \setminus \{0\}$,
  the $\mathbb{Z}/2$-action (\ref{ReflectionAutomorphismOnExceptionalTangentSuperSpacetime})
  on the exceptional superspacetime $\mathbb{R}^{10,1\vert \mathbf{32}}_{\mathrm{exc},s}$
  induces an isomorphism of its $\tilde \mu_{{}_{M5,s}}$-twisted cohomology (Def. \ref{TwistedCohomology}):
  $$
    \xymatrix{
      H^{\bullet + \tilde \mu_{{}_{M5,s}}}( \mathbb{R}^{10,1\vert\mathbf{32}}_{\mathrm{exc},s} )
      \ar[rr]^{\rho^\ast}_{\simeq}
      &&
      H^{\bullet + \tilde \mu_{{}_{M5,s}}}( \mathbb{R}^{10,1\vert\mathbf{32}}_{\mathrm{exc},s} )
    }\;.
  $$
\end{prop}

\medskip
Curiously, observe the following interesting effect on the top cocycle.

\begin{prop}
The 528-volume element (\ref{TheExceptionalVolumeForm})
is invariant under the reflection automorphism of Prop. \ref{ReflectionAutomorphismOnExceptionalTangentSuperSpacetime}:
$$
  \rho^\ast(\mathrm{vol}_{528})
  =
  \mathrm{vol}_{528}
$$
\end{prop}
\begin{proof}
  The factors in $\mathrm{vol}_{528}$ that change sign under $\rho^\ast$
  are
  1. those $e^a$ for which $a =10$, of which there is one, which is an odd number;
  2. those $B_{a_1 a_2}$ with $a_1 \lt a_2$ for which $a_1 \neq 10$ and $a_2 \neq 10$, of which there are $\left( 10 \atop 2\right) = 45$, which is also an odd number;
  3. those $B_{a_1 \cdots a_5}$ with $a_1 < \cdots < a_5$ for which $a_5 = 10$, of which there are $\left(  10 \atop 4 \right) = 210$, which is an even number.

In total this means that under $\rho^\ast$ the element $\mathrm{vol}_{528}$ picks up an even number of signs, hence is fixed.
\end{proof}


\end{document}